%% file: main.tex
\pdfoutput=1
\documentclass{article}

\input{macros}

\title{Score Design for Multi-Criteria Incentivization}
\author{%
    Anmol Kabra\thanks{Work done while at TTI-Chicago. A condensed version of this paper appeared at Foundations of Responsible Computing (FORC) 2024.}\\
    \small \texttt{anmol@cs.cornell.edu}\\
    \small Cornell University
    \and
    Mina Karzand\\
    \small \texttt{mkarzand@ucdavis.edu}\\
    \small UC Davis
    \and
    Tosca Lechner\\
    \small \texttt{tlechner@uwaterloo.ca}\\
    \small UWaterloo
    \and
    Nati Srebro\\
    \small \texttt{nati@ttic.edu}\\
    \small TTI-Chicago
    \and
    Serena Wang\\ 
    \small \texttt{serenalwang@berkeley.edu}\\
    \small UC Berkeley and Google
}
\date{}

\begin{document}

\maketitle

\begin{abstract}
We present a framework for designing scores to summarize performance metrics. 
Our design has two multi-criteria objectives: (1) improving on scores should improve all performance metrics, and (2) achieving pareto-optimal scores should achieve pareto-optimal metrics.
We formulate our design to minimize the dimensionality of scores while satisfying the objectives.
We give algorithms to design scores, which are provably minimal under mild assumptions on the structure of performance metrics.
This framework draws motivation from real-world practices in hospital rating systems, where misaligned scores and performance metrics lead to unintended consequences.
\end{abstract}

\input{intro}

\input{obj__improvement}

\input{obj__optimality}

\input{obj__both}

\input{computing_matrix_ranks}

\input{conclusion}

\input{acknowledgements}

\begingroup
\hypersetup{linkcolor=red}
\bibliographystyle{plainnat}
\bibliography{refs}
\endgroup

\input{appendix}

\end{document}

%% file: macros.tex
\PassOptionsToPackage{noend}{algpseudocode}
\PassOptionsToPackage{backref=page}{hyperref}
\PassOptionsToPackage{dvipsnames}{xcolor}
\input{global_macros}
\usepackage[toc,page]{appendix}
\usepackage[margin=1in]{geometry}
\usepackage{tikz}
\usepackage{tikz-3dplot}
\usepackage{pgfplots}
\pgfplotsset{compat=1.16}
\usepackage[sort,numbers]{natbib}
\usepackage{enumerate}

\definecolor{mygreen}{rgb}{0.1, 0.8, 0.1}
\colorlet{myred}{red!70!black}

\hypersetup{
  colorlinks=true,
  linkcolor=black,
  citecolor=red,
}

\newtheorem{theorem}{Theorem}
\newtheorem{corollary}[theorem]{Corollary}
\newtheorem{lemma}[theorem]{Lemma}
\newtheorem{proposition}[theorem]{Proposition}
\newtheorem{claim}[theorem]{Claim}
\newtheorem*{claim*}{Claim}

\theoremstyle{definition}
\newtheorem{definition}[theorem]{Definition}
\newtheorem{example}[theorem]{Example}

\newtheorem*{assumption*}{Assumption}

\theoremstyle{remark}
\newtheorem{remark}[theorem]{Remark}
\newtheorem*{remark*}{Remark}

\usepackage{apptools}
\AtAppendix{\counterwithin{theorem}{section}}

\newcommand{\innerproduct}[2]{\dotprod{#1}{#2}}
\newcommand{\stcomplement}[1]{{#1}^{\mathsf{c}}}
\newcommand{\longiff}[1]{\;\xLeftrightarrow{#1}\;}
\newcommand{\longimplies}[1]{\;\xRightarrow{#1}\;}
\newcommand{\longeq}[1]{\;\stackrel{#1}{=\joinrel=}\;}

\DeclareMathOperator{\rowspan}{rowspan}
\let\spanv\relax % remove existing definition of span
\DeclareMathOperator{\spanv}{span}

\newcommand{\ParetoOpt}{\mathrm{ParetoOpt}}
\newcommand{\ConeRank}{\mathsf{ConeRank}}
\newcommand{\ConeGeneratingRank}{\mathsf{ConeGeneratingRank}}
\newcommand{\ConeSubsetRank}{\mathsf{ConeSubsetRank}}
\newcommand{\CR}{\mathsf{CR}}
\newcommand{\CGR}{\mathsf{CGR}}
\newcommand{\CSR}{\mathsf{CSR}}
\newcommand{\NonNegativeRank}{\mathrm{NonNegativeRank}}

\newcommand{\score}{S}
\newcommand{\restrictionCS}{\text{Res-CS}\xspace}
\newcommand{\restrictionLM}{\text{Res-LM}\xspace}
\newcommand{\restrictionL}{\text{Res-L}\xspace}

\definecolor{Gred}{RGB}{219, 50, 54}
\definecolor{Ggreen}{RGB}{60, 186, 84}
\definecolor{Gblue}{RGB}{72, 133, 237}
\definecolor{Gyellow}{RGB}{247, 178, 16}
\definecolor{ToCgreen}{RGB}{0, 128, 0}
\definecolor{myGold}{RGB}{231,141,20}
\definecolor{myBlue}{rgb}{0.19,0.41,.65}
\definecolor{myPurple}{RGB}{175,0,124}

\usepackage[colorinlistoftodos,prependcaption,textsize=scriptsize]{todonotes}
\providecommand{\tdotoggle}{0}
 \ifnum\tdotoggle=1
 \fi
\newcommand{\mytodo}[1]{\ifnum\tdotoggle=1{#1}\fi}

\newcommand{\unsure}[1]{\mytodo{\todo[linecolor=myGold,backgroundcolor=myGold!25,bordercolor=myGold]{#1}}}

\newcommand{\tableoftodos}{\ifnum\tdotoggle=1 \listoftodos[Comments/To Do's] \fi}

\newcommand{\commentAnmol}[1]{\unsure{Anmol: #1}}

\newcommand{\inbox}[1]{
\vspace*{0.5em}
\noindent\fbox{%
    \parbox{\dimexpr\linewidth-2\fboxsep-2\fboxrule}{%
        #1
    }%
}
}

\newcommand{\WinL}{W_{\calL}}
\newcommand{\WnotinL}{W_{\cdot \setminus \calL}}
\newcommand{\bfWinL}{\bfW_{\calL}}
\newcommand{\bfWnotinL}{\bfW_{\cdot \setminus \calL}}
\newcommand{\VinL}{V_{\calL}}
\newcommand{\VnotinL}{V_{\cdot \setminus \calL}}
\newcommand{\bfVinL}{\bfV_{\calL}}
\newcommand{\bfVnotinL}{\bfV_{\cdot \setminus \calL}}

\newcommand{\boldheading}[1]{\textbf{#1}}

%% file: global_macros.tex
\usepackage[utf8]{inputenc}
\usepackage[usenames,dvipsnames,table]{xcolor}
\usepackage{tcolorbox}
\usepackage{xspace}

\usepackage{amsmath,amssymb,amsthm,amsfonts}
\usepackage{centernot}
\usepackage{mathtools}
\usepackage{bbm}
\usepackage{bm}
\usepackage{dsfont}
\usepackage{nicefrac}

\usepackage{algorithm,algorithmicx}
\usepackage{algpseudocode}

\usepackage{graphicx}
\usepackage{caption}
\usepackage{subcaption}
\usepackage{hyperref}
\usepackage{nameref}
\usepackage{cleveref}
\usepackage{multirow}
\usepackage{arydshln}% for dashed hlines in tables
\usepackage{tablefootnote}
\usepackage{refcount}
\usepackage{booktabs}

\usepackage{xparse}
\ExplSyntaxOn

\NewDocumentCommand{\definealphabet}{mmmm}
 {% #1 = prefix, #2 = command, #3 = start, #4 = end
  \int_step_inline:nnn { `#3 } { `#4 }
   {
    \cs_new_protected:cpx { #1 \char_generate:nn { ##1 }{ 11 } }
     {
      \exp_not:N #2 { \char_generate:nn { ##1 } { 11 } }
     }
   }
 }

\ExplSyntaxOff

\definealphabet{bb}{\mathbb}{A}{Z}
\definealphabet{cal}{\mathcal}{A}{Z}
\definealphabet{bf}{\bm}{a}{z}
\definealphabet{bf}{\bm}{A}{Z}
\newcommand{\bfzero}{\bm{0}}
\newcommand{\bfone}{\bm{1}}

\newcommand{\bfalpha}{\bm{\alpha}}
\newcommand{\bfbeta}{\bm{\beta}}

\newcommand{\bflambda}{\bm{\lambda}}
\newcommand{\bfmu}{\bm{\mu}}
\newcommand{\bfnu}{\bm{\nu}}

\newcommand{\inparens}[1]{\left( #1 \right)}
\newcommand{\inbraces}[1]{\left\{ #1 \right\}}
\newcommand{\inbraks}[1]{\left[ #1 \right]}
\newcommand{\set}[1]{\inbraces{#1}}

\newcommand{\abs}[1]{\left\lvert #1 \right\rvert}
\newcommand{\norm}[1]{\left\lVert #1 \right\rVert}
\newcommand{\dotprod}[2]{\left\langle #1, #2 \right\rangle}

\DeclareMathOperator*{\argmin}{arg\,min}

\DeclareMathOperator{\spanv}{\mathrm{Span}}
\DeclareMathOperator{\rank}{\mathrm{rank}}

\DeclareMathOperator{\affine}{\mathrm{aff}}
\DeclareMathOperator*{\proj}{\mathbf{proj}}

\DeclareMathOperator{\conv}{\mathrm{Conv}}
\DeclareMathOperator{\cone}{\mathrm{Cone}}
\DeclareMathOperator{\lineality}{\mathrm{lin}}

\newcommand{\round}[2]{{#1}^{(#2)}}

\newcommand{\poly}{{\mathrm{poly}}}
\newcommand{\polylog}{{\mathrm{polylog}}}

%% file: intro.tex
\section{Introduction}
\label{sec:introduction}

The use of numerical metrics to evaluate performance and guide decision-making is common practice in healthcare, education, business, and public policy.
It is common for agencies to design \textit{surrogate scores} that summarize performance metrics, in a way that aligns incentives with performance metrics.
Often the scored entities strategically optimize surrogates and end up degrading on metrics, a phenomenon commonly known as \textit{unintended consequences} and pithily conveyed by Goodhart's law~\cite{goodhart1984problems,strathern1997improving}:
\begin{center}
    \it ``When a measure becomes a target, it ceases to be a good measure.''
\end{center}
Agencies thus aim to ensure that optimizing scores leads to improved metrics.
As the number of performance metrics can be large in practice~\cite{wadhera2020quality,muller2018tyranny}, agencies must design \textit{succinct} multi-dimensional surrogate scores.
We present a framework to study this \textit{minimal design problem}, and propose score designs that prevent unintended consequences.

Our work is directly motivated by real-world examples in safety-critical domains such as healthcare and education, where manifestations of Goodhart's law exemplify the serious ramifications of unintended consequences.
When Pacificare, a healthcare provider, incentivized hospitals in 2003 to perform certain medical procedures to improve quality of care, several unrepresented metrics deteriorated~\cite{koretz2017}.
Similar misalignment between performance metrics and score-based hospital ratings, used by the Medicare agency (CMS), has been widely critiqued~\cite{nytimes2013the,chicagoboothreview2020hospital,kim2022hospital,aggarwal2021association,schardt2020increase,alexander2020doctors}.
Even so, CMS uses these score-based ratings to incentivize hospital policies~\cite{cms_vbp,conrad2015theory}.
Hence, it aims to design scores so that improving on scores also improves all performance metrics.
This goal motivates the \textit{improvement objective} in our framework.
In a similar vein, rating agencies such as USNews aim to incentivize efficient use of hospital resources through published scores~\cite{usnews_rankings}.
On multi-dimensional metrics, the efficiency goal~\cite{committee2001crossing} naturally translates into the notion of pareto-efficiency, which motivates the \textit{optimality objective} in our framework.

We present a framework for designing scores to summarize performance metrics.
We give three natural design restrictions that align with real-world interpretability desiderata~\cite{cms_starratings,usnews_rankings}, and propose score designs that satisfy the multi-criteria objectives under these restrictions.
Striving for succinct scores, we formulate our design to minimize the dimensionality of scores.
We give polynomial-time algorithms to design these succinct scores, which are provably minimal under mild assumptions on the structure of performance metrics.
While existing work on score design for incentivization studies scalar scores~\cite{kleinberg2020classifiers,haghtalab2021maximizing,rolf2020balancing,wang2023operationalizing}, we design scores of smallest dimensionality to satisfy the multi-criteria objectives.
These objectives are unsatisfiable with scalar scores in general.

\subsection{Designing surrogate scores from performance metrics}
\label{sec:intro_model}

In our model, the agency aims to design a surrogate score function $S : \calF \to \calS$ given a set of performance metrics $\calF$ of hospitals.

Hospitals report to agencies like CMS and USNews on hundreds of performance metrics such as condition-specific death rates, readmission rates, and percentages of patients receiving satisfactory care~\cite{cms_starratings,cms_2024starratings_report,usnews_rankings}.
We can denote the values of $d$ metrics of a hospital with a real-valued vector $\bff \in \calF \subseteq \bbR^d$.
Since $d$ is large and metrics can be related through confounding variables~\cite{ash2012statistical,kurian2021predicting}, the agency wants to summarize the $d$ metrics as $k$ scores with values $\calS \subseteq \bbR^k$, where $k$ is small as possible.
For instance, \Cref{example:1__2d__correlated_metrics} suggests that, to summarize COVID and pneumonia death rate metrics, the agency can choose either of the two metrics as the score, so that $k = 1$.
Whereas for pneumonia death rate and excess antibiotic use metrics, \Cref{example:1__2d__anticorrelated_metrics} argues that selecting both metrics as scores is necessary, and so $k = 2$.

\paragraph*{Surrogate design objectives}

Anticipating that the hospital would target the incentives by optimizing the score function $S$, the agency wants to design $S$ in such a way that optimizing them ensures that the hospital does well on the performance metrics.
We formalize this goal with two design objectives, which utilize an ordering on the sets $\calF$ and $\calS$, denoted by $\succeq_{\calF}$ and $\succeq_{\calS}$.
The two objectives are motivated from CMS and USNews hospital rating agencies~\cite{cms_starratings,usnews_rankings}.

\begin{enumerate}
    \item \boldheading{Improvement objective.}
    Improving on surrogate scores should result in improving on performance metrics.
    In particular,
    \begin{align}
    \label{eqn:definition_improvement_general}
        \text{for } \bff, \bff' \in \calF,\quad \text{ if } S(\bff') \succeq_{\calS} S(\bff) \text{ then } \bff' \succeq_{\calF} \bff.
    \end{align}

    \item \boldheading{Optimality objective.}
    Pareto-optimal points of surrogate scores should be pareto-optimal points of performance metrics.
    In particular,
    \begin{align}
    \label{eqn:definition_optimality_general}
        \ParetoOpt(S) \subseteq \ParetoOpt(\calF).
    \end{align}
\end{enumerate}

Throughout the paper, we analyze the setting $\calF \subseteq \bbR^d$ and $\calS \subseteq \bbR^k$ and use elementwise order of vectors for $\succeq_{\calF}$ and $\succeq_{\calS}$.

\paragraph*{Surrogate design restrictions}
Due to interpretability and public reporting obligations, rating agencies like CMS and USNews design scores by selecting subsets of the list of performance metrics or by taking weighted averages~\cite{cms_2024starratings_report,cms_starratings,cms_mipspayment,cms_2021risk_medicare,usnews_rankings}.
Moreover, monotonicity of scores in performance metrics is a desirable property for CMS, as it ensures that a hospital striving to improve all performance metrics sees improved score values~\cite{cms_2024starratings_report,cms_2021risk_medicare}.

We formulate these requirements as three different restrictions on $S$.
These restrictions impose a linear form on $S : \bff \mapsto \bfA \bff$ with $\bfA \in \bbR^{k \times d}$ satisfying certain structural constraints.

\begin{enumerate}
    \item \boldheading{Coordinate Selection (\restrictionCS).}
    Each of the $k$ coordinates of scores are chosen from $d$ coordinates of performance metrics.
    That is, for all $i \in [k]$ there exists $j \in [d]$ such that $S(\bff)_i = \bff_j$ for all $\bff \in \calF$.
    Equivalently, $S : \bff \mapsto \bfA \bff$ where rows of $\bfA$ are 1-hot vectors.
    
    \item \boldheading{Linear and Monotone (\restrictionLM).}
    The $k$ coordinates of scores are linear combinations of $d$ coordinates of performance metrics, and improving on performance metrics should result in improving on surrogate scores.
    That is, $S : \bff \mapsto \bfA \bff$ where for $\bff, \bff' \in \calF$, if $\bff' \ge \bff$ then $\bfA \bff' \ge \bfA \bff$.
    
    \item \boldheading{Linear (\restrictionL).}
    The coordinates of surrogate scores are linear combinations of coordinates of performance metrics.
    That is, $S : \bff \mapsto \bfA \bff$ without any further constraints on $\bfA$.
\end{enumerate}

\paragraph*{Minimal design problem}
Since the number of performance metrics $d$ can be large~\cite{cms_2024starratings_report,cms_starratings,usnews_rankings},
a natural goal is to \textit{succinctly} summarize metrics with scores that are accessible to patients and policymakers.
This goal of succinctness translates into designing a multi-dimensional function $S : \bbR^d \to \bbR^k$ with \textit{the smallest output dimension} $k$.
For a combination of design objective and design restriction, the \textit{minimal design problem} is determining the smallest dimensionality $k$ and providing an algorithm outputs a surrogate score function $S$ with this $k$.

\subsection{Our contributions}
In this paper, we study the minimal design problem.
Our key contributions are:

\begin{enumerate}
    \item We formalize surrogate score design for incentivizing multiple criteria, motivated from real-world practices of two hospital rating systems, CMS and USNews. 

    \item We fully determine the minimal design problems of all combinations of objectives and restrictions introduced in \Cref{sec:intro_model}, and propose efficient score design algorithms (\Cref{alg:1__design_strategy,alg:2__design_strategy}).
    We summarize our results in \Cref{tab:results__summary}.
    \begin{enumerate}
        \item 
        We show that the smallest dimensionalities $k$ are dictated by structural properties of the affine hull of performance metrics $\calF$.

        \item 
        Identifying a relationship between improvement and optimality objectives (\Cref{thm:improvement_monotone__imply__optimality}), we determine the minimal design problem for simultaneously satisfying both objectives.
    \end{enumerate}

    \item We give polynomial-time algorithms for computing the three matrix ranks (\Cref{sec:computing_matrix_ranks}) used to determine the minimal design for the improvement objective, thus enabling efficient design algorithms.
    In doing so, we develop novel techniques to decompose and enclose polyhedral cones, augmenting research in computational geometry on manipulating polyhedral cones.
\end{enumerate}
\begin{table}[!h]
    \centering
    \begin{tabular}{c c c c}
        \toprule
        Restriction & Improvement (\S\ref{sec:feasibility_improvement}) & Optimality (\S\ref{sec:feasibility_optimality}) & Both (\S\ref{sec:feasibility_both})\\
        \midrule
        \restrictionCS & $\ConeSubsetRank(\bfZ)$ & $r$ & $\ConeSubsetRank(\bfZ)$ \\
        \restrictionLM & $\ConeGeneratingRank(\bfZ)$ & $1$ & $\ConeGeneratingRank(\bfZ)$ \\
        \restrictionL & $\ConeRank(\bfZ)$ & $1$ & $\ConeGeneratingRank(\bfZ)$ \\
        \bottomrule
    \end{tabular}
    \caption{
    We list smallest dimensionalities $k$ for the minimal design problem of all combinations of objectives and restrictions.
    Here columns of $\bfZ$ are an orthonormal basis of the linear subspace associated with $r$-dimensional affine hull of $\calF$.
    We define the three matrix ranks $\ConeSubsetRank, \ConeGeneratingRank, \ConeRank$ in \Cref{thm:1__f_linear__suff}.
    For the improvement objective, the listed dimensionalities are also necessary, when $\calF$ has non-empty relative interior (\Cref{thm:1__f_linear__nonempty_relint_f_with_origin__necsuff}).
    }
    \label{tab:results__summary}
\end{table}

\subsection{Related work}
\label{sec:related_work}

Recent work has highlighted the plight of score-based incentivization when scores that do not align with performance metrics.
In healthcare, design objectives of hospital rating agencies often vary across agencies.
Two popular examples are the Medicare agency (CMS), which incentivizes healthcare investment across care metrics through a five-star score~\cite{cms_starratings,conrad2015theory}, and the USNews agency, which promotes highly-specialized medical departments~\cite{usnews_rankings}.
When hospitals target these score-based ratings, they often degrade on a few performance metrics~\cite{koretz2017}.
For example, CMS's score-based ratings have been found to encourage hospitals to selectively treat patients for minimizing readmission rates~\cite{alexander2020doctors,dranove1994physician,clemens2014physicians}, and have exacerbated unequal access to healthcare~\cite{kim2022hospital,aggarwal2021association,schardt2020increase}.
Such unintended consequences are prevalent in fields that use scores as an incentive mechanism~\cite{bandalos2018measurement}, for instance, in standardized testing~\cite{koretz2017} and financial credit ratings~\cite{holmstrom1994incentive,white2013credit,bar2011credit,grove2015high}.

Our framework extends recent work on score design in principal-agent theory~\cite{kleinberg2020classifiers,haghtalab2021maximizing,rolf2020balancing,wang2023operationalizing,guerdan2023counterfactual,liu2023actionability,hartline2022optimal,hartline2020optimization,alon2020multiagent,ahmadi2022setting} by designing scores for multi-criteria objectives.
Kleinberg and Raghavan~\cite{kleinberg2020classifiers} compare linear with monotone scalar score design for incentivizing \textit{effort} from agents.
On a similar front, Haghtalab et al.~\cite{haghtalab2021maximizing} study scalar score design with a linear threshold restriction.
Score design has also been studied through a causality lens to optimize the average treated outcome~\cite{wang2023operationalizing,guerdan2023counterfactual,liu2023actionability}.
Finally, Rolf et al.~\cite{rolf2020balancing} use noisy score observations to approximate the pareto-frontier of performance metrics.
Our framework's optimality objective and design restrictions capture this line of work on scalar scores.
However, our improvement objective is a novel contribution, and this objective turns to be unsatisfiable with scalar scores (\Cref{thm:1__f_linear__nonempty_relint_f_with_origin__necsuff}).
Hence, our score design problems are inherently multi-criteria.

Technically, our design algorithms utilize novel techniques to decompose and enclose polyhedral cones, building on work in computational geometry on finding frames of polyhedral cones~\cite{dula1998algorithm,lopez2011algorithm,wets1967algorithms} and enclosing convex hulls~\cite{gale1953inscribing,kumar2013fast,orourke1986optimal,toussaint1983solving}.
Our definition of $\ConeRank$ (\Cref{thm:1__f_linear__suff}) is similar to $\NonNegativeRank$, which is extensively studied in the context of non-negative matrix factorization~\cite{gillis2020nonnegative,gillis2013fast,donoho2003does,vavasis2010complexity,kumar2013fast}.

\subsection{Notation}
\label{sec:notation}

We represent scalars as $\lambda, c \in \bbR$, and vectors and matrices as $\bfw \in \bbR^n, \bfW \in \bbR^{m \times n}$.
We denote the nonnegative orthant with $\bbR^n_+$.
We generally write matrices as a stack of rows, $\bfW = [\bfw_1; \dots; \bfw_m]$, often denoting the set of rows with $W$.
We say that matrix $\bfW$ (or set $W$) generates cone $\calK_W$ if $\calK_W = \cone(W) = \set{ \bfx \in \bbR^n \mid \bfx = \bflambda \bfW, \bflambda \in \bbR^m_+ }$.
We denote a vector of zeros (or ones) as $\bfzero_n \in \bbR^n$ (or $\bfone_n$), and the $n$-by-$n$ identity matrix as $\bfI_n$, dropping subscripts when unambiguous.

\subsection{Organization}
The rest of the paper is organized as follows.
We analyze the minimal design problem for improvement objective in \Cref{sec:feasibility_improvement}, and optimality objective in \Cref{sec:feasibility_optimality}.
Subsequently in \Cref{sec:feasibility_both}, we analyze the minimal design problem for simultaneously satisfying both improvement and optimality objectives.
Our minimal design for the improvement objective uses three notions of matrix ranks.
In \Cref{sec:computing_matrix_ranks} we extensively study these matrix ranks and present intuition for algorithms to compute them.
We include proofs of all results in \Cref{app:omitted_proofs}.
In \Cref{app:algorithms} we give full details of these algorithms.
We provide relevant background on convex analysis and geometry in \Cref{app:preliminaries}, and include all technical lemmas in \Cref{app:key_lemmas}.

%% file: obj__improvement.tex
\section{Minimal design problem for improvement objective}
\label{sec:feasibility_improvement}

We propose a surrogate score design for satisfying the improvement objective under the three design restrictions.
Then we illustrate our design strategy on simple examples of performance metrics $\calF$, highlighting relationships between the geometry of $\calF$ and the succinctness of scores.
Finally, we show that our proposed design is minimal under a mild assumption on $\calF$, implying that score design for improvement objective is inherently multi-criteria.

We first simplify the improvement objective in \Cref{eqn:definition_improvement_general} to identify geometric objects that represent \textit{movement} and \textit{improvement directions}.
Score function $S : \bff \mapsto \bfA \bff$ on domain $\calF$ satisfies improvement when for all $\bff, \bff' \in \calF$, if $\bfA (\bff' - \bff) \ge \bfzero$ then $(\bff' - \bff) \ge \bfzero$.
Denoting the \textit{movement directions at center $\bff$} with $\calF_{\bff} = \set{ \bfg = \bff' - \bff \in \bbR^d \mid \text{ for all } \bff' \in \calF }$, we can rearrange terms to get
\begin{align}
    \text{for all centers } \bff \in \calF, \text{ movement directions } \bfg \in \calF_{\bff}, \quad \text{ if } \bfA \bfg \ge \bfzero \text{ then } \bfI \bfg \ge \bfzero
\end{align}

Here the \textit{set of score improvement directions} is exactly $\calK_A^\ast = \set{ \bfg \in \bbR^d \mid \bfA \bfg \ge \bfzero }$, which is the dual of polyhedral cone $\calK_A$ generated from rows of $\bfA$.
Similarly, the \textit{set of metric improvement directions} is $\calK_I^\ast = \set{ \bfg \in \bbR^d \mid \bfI \bfg \ge \bfzero } = \bbR^d_+$, which is the dual of polyhedral cone $\calK_I = \bbR^d_+$ generated from rows of $\bfI$.
So intuitively, score function $\score : \bff \mapsto \bfA \bff$ satisfies improvement if and only if every movement direction (in $\calF_{\bff}$) that is a score improvement direction (in $\calK_A^\ast$) is also a metric improvement direction (in $\calK_I^\ast$):
\begin{align}
\label{eqn:1__f_linear__restate_cones}
    S \text{ satisfies improvement} &\iff \text{for all } \bff \in \calF,\quad \calF_{\bff} \cap \calK_A^\ast \subseteq \calK_I^\ast.
\end{align}

\subsection{Design proposal for improvement objective}
\label{sec:feasibility_improvement__suff}

When performance metrics $\calF \subseteq \bbR^d$ is a full-dimensional set, score design is trivial where the most succinct score design is $S(\bff) = \bff$.
Note that while performance is measured in many dimensions~\cite{wadhera2020quality,muller2018tyranny}, the number of confounding variables of performance metrics is often smaller due to correlated metrics~\cite{ash2012statistical,kurian2021predicting}.
This typically induces a low-dimensional structure on $\calF$, observed in practice and assumed in theory~\cite{bandalos2018measurement,barclay2022concordance,ash2012statistical,kurian2021predicting}.
We do not assume such low-dimensional structure of $\calF$, but the smallest dimensionality $k$ of score function $S$ is impacted by the intrinsic dimension of $\calF$.
The affine hull of $\calF$ is a natural geometric choice to capture its intrinsic dimension.

\begin{definition}
\label{defn:affine_hull_f}
    Define the affine hull of $\calF$, $\affine(\calF)$, as the intersection of all affine subspaces in $\bbR^d$ containing $\calF$.
    Let $\calL$ be the linear subspace associated with $\affine(\calF)$, i.e. $\calL$ is the translation of $\affine(\calF)$ so that for all centers $\bff \in \calF$, movement directions $\calF_{\bff} \subseteq \calL$.
\end{definition}

By utilizing this subspace $\calL$ containing all possible movement directions $\calF_{\bff}$, we propose a score design in \Cref{alg:1__design_strategy} with dimensionalities given in \Cref{thm:1__f_linear__suff}.
We introduce three \textit{matrix ranks}---$\ConeSubsetRank \, (\CSR)$, $\ConeGeneratingRank \, (\CGR)$, and $\ConeRank \, (\CR)$---to characterize the score design dimensionalities for the three respective design restrictions---Coordinate Selection (\restrictionCS), Linear and Monotone (\restrictionLM), Linear (\restrictionL).
These three matrix ranks capture the geometric properties of performance metrics $\calF$ that dictate the dimensionality of optimal score design for the three restrictions.

\begin{theorem}
    \label{thm:1__f_linear__suff}
    Let columns of $\bfZ$ be an orthonormal basis of linear subspace $\calL$ associated with $\affine(\calF)$, and let $r = \dim \affine(\calF)$.
    For each design restriction, there exists $S : \calF \to \bbR^k$, designed using \Cref{alg:1__design_strategy}, that satisfies the improvement objective with the following dimensionalities.
    
    \setlength{\tabcolsep}{3pt}
    \noindent
    \begin{center}
        \begin{tabular}{l|l@{}l} % remove space between last two columns
             & \multicolumn{2}{c}{Dimensionality $k \ge$} \\
            \midrule
            \restrictionCS & $\ConeSubsetRank(\bfZ)$ & $\coloneqq \min_q \set{ q \mid \calK_Z = \calK_V \text{ for some } \bfV \in\bbR^{q \times r} \text{ s.t. } \bfV \subseteq \bfZ }$ \\
            \restrictionLM & $\ConeGeneratingRank(\bfZ)$ & $\coloneqq \min_q \set{ q \mid \calK_Z = \calK_V \text{ for some } \bfV \in \bbR^{q \times r} }$ \\
            \restrictionL & $\ConeRank(\bfZ)$ & $\coloneqq \min_q \set{ q \mid \calK_Z \subseteq \calK_V \text{ for some } \bfV \in \bbR^{q \times r} }$ \\
        \end{tabular}
    \end{center}
\end{theorem}

\begin{algorithm}[!h]
    \begin{algorithmic}[1]
        \State Given: performance metrics $\calF$ and a design restriction.
        \State Find $\bfZ$ whose columns are an orthonormal basis of subspace $\calL$ associated with $\affine(\calF)$.
        \State Find $\bfV$ that attains\footnotemark~the matrix rank corresponding to the design restriction.
        \State Find $\bfA$ that satisfies $\bfV = \bfA \bfZ$ and design $S : \bff \mapsto \bfA \bff$.
    \end{algorithmic}
    \caption{Design strategy for improvement objective}
    \label{alg:1__design_strategy}
\end{algorithm}

\Cref{thm:1__f_linear__suff} follows from the following key insight of \Cref{eqn:1__f_linear__restate_cones}: ``for $S : \bff \to \bfA \bff$ to satisfy the improvement objective, score improvement directions need to be metric improvement directions \textbf{only} for movement directions $\calF_{\bff}$, which are contained in subspace $\calL$.''
In fact, satisfying the improvement objective boils down to ensuring that score improvement directions are a subset of metric improvement directions \textit{in the coefficient space} w.r.t. subspace $\calL$.
The respective improvement directions $\calK_A^\ast$ and $\calK_I^\ast$ are generated by rows of $\bfA$ and $\bfI$, which have coefficients that are rows of $\bfV = \bfA \bfZ$ and $\bfZ$, where columns of $\bfZ$ are an orthonormal basis of subspace $\calL$.
It turns out that improvement directions in the coefficient space are precisely the duals $\calK_V^\ast$ and $\calK_Z^\ast$ of polyhedral cones generated from rows of $\bfV$ and $\bfZ$.
So to satisfy the improvement objective, we need to ensure $\calK_V^\ast \subseteq \calK_Z^\ast$, or $\calK_Z \subseteq \calK_V$.

With the three matrix ranks, we capture the additional structure on $\bfA$ imposed by the three design restrictions (\Cref{sec:intro_model}).
\restrictionL restriction does not further impose structure on $\bfA$, and so we only need to \textit{enclose} cone $\calK_Z$ with $\calK_V$.
\restrictionLM restriction further requires function $S$ to be monotone in $\calF$, which intuitively means that every metric improvement direction needs to be a score improvement direction, i.e., $\calK_Z^\ast \subseteq \calK_V^\ast$.
So to satisfy \restrictionLM, we must \textit{generate} cone $\calK_Z$ with $\calK_V$.
Finally, \restrictionCS restriction requires selecting the $k$ score function coordinates from $d$ metrics.
In the coefficient space, this requirement means that rows of $\bfV$ are chosen from rows of $\bfZ$ and $\calK_V$ generates $\calK_Z$.
Hence, the three matrix ranks precisely capture structure on $\bfA$ imposed by the improvement objective and the design restrictions.
We include the proof of \Cref{thm:1__f_linear__suff} in \Cref{thm_w_proof:1__f_linear__suff}.

\footnotetext{%
For a matrix rank, e.g. $\CSR$, we say that $\bfV$ ``attains'' it if $\bfV \subseteq \bfZ$ (rows of $\bfV$ are chosen from rows of $\bfZ$), $\calK_Z = \calK_V$, and the number of rows of $\bfV$ equals $\CSR(\bfZ)$.%
}

\subsection{Geometry of metrics dictates succinctness of scores}

We now illustrate \Cref{alg:1__design_strategy} with several examples of metrics $\calF$.
We instantiate performance metrics in our examples with familiar notions of hospital metrics, to intuitively bridge our analysis and algorithm with practical score design.
In doing so, we discuss how the geometry of $\calF$ dictates the shape of polyhedral cone $\calK_Z$, influencing the dimensionality of minimal score design for the three design restrictions.
Finally, we provide high-level descriptions of techniques to to implement \Cref{alg:1__design_strategy} efficiently.

\begin{figure}[!h]
    \centering
    \begin{subfigure}{0.49\textwidth}
        \centering
        \input{tikz/1example_f_2d_pointed}
        \caption{When the two metrics are correlated (Ex.~\ref{example:1__2d__correlated_metrics}), we can choose either metric in $S : \calF \to \bbR^1$.}
        \label{fig:1example_f_2d__correlated}
    \end{subfigure}
    \hfill
    \begin{subfigure}{0.49\textwidth}
        \centering
        \input{tikz/1example_f_2d_nonpointed}
        \caption{When the two metrics are anti-correlated (Ex.~\ref{example:1__2d__anticorrelated_metrics}), we must choose both metrics in $S : \calF \to \bbR^2$.}
        \label{fig:1example_f_2d__anticorrelated}
    \end{subfigure}
    \caption{
    To design scores for two metrics ($\calF \subseteq \bbR^2$), we can inspect the correlation between metrics---the correlation dictates the succinctness of $S : \calF \to \bbR^k$ for satisfying improvement.
    }
\end{figure}
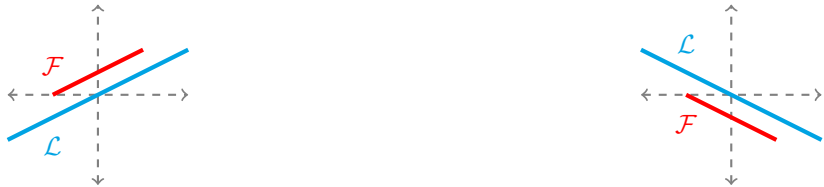

\begin{example}[Two correlated metrics$\implies$choose either for score design]
\label{example:1__2d__correlated_metrics}
    CMS evaluates hospitals on numerous performance metrics like condition-specific death rates, readmission rates, and safety standards~\cite{cms_starratings}.
    Often comorbidities of medical conditions can lead to positive correlations between metrics.
    In the case of two \textit{perfectly} positively correlated metrics, \Cref{alg:1__design_strategy} suggests to choose either of the two metrics to design $S : \calF \to \bbR^1$.

    Consider two metrics---(i) pneumonia death rate and (ii) COVID-19 death rate---that have a positive correlation due to comorbidities.
    Assume that for a hospital, these two death rates take values $\calF = \set{ \bff \in \bbR^2 \mid -f_1 + 2 f_2 = 1, -1 \le f_1 \le 1}$, lying in a 1-dimensional affine subspace of $\bbR^2$ (\Cref{fig:1example_f_2d__correlated}, red).
    As the affine hull $\affine(\calF) = \set{ \bff \mid -f_1 + 2 f_2 = 1 }$ is 1-dimensional, the associated linear subspace $\calL = \set{ \bff \mid -f_1 + 2 f_2 = 0 }$ (\Cref{fig:1example_f_2d__correlated}, blue) containing all movement directions $\calF_{\bff}$ is 1-dimensional.
    Per Line~2 of \Cref{alg:1__design_strategy}, we arrange an orthonormal basis for $\calL$ as columns of $\bfZ \propto \begin{bmatrix}2 \\ 1\end{bmatrix}$, whose rows generate the polyhedral cone $\calK_Z = \set{ 2\lambda_1 + \lambda_2 \mid \lambda_1, \lambda_2 \ge 0 } = \bbR_+$.
    Note that the metric improvement directions in the coefficient space are the dual cone $\calK_Z^\ast = \bbR_+$.
    
    To satisfy improvement objective under a design restriction, we need to find matrix $\bfV$ that attains the corresponding matrix rank.
    For all three matrix ranks, the cone $\calK_V$ generated by rows of $\bfV$ needs to \textit{enclose} cone $\calK_Z$.
    Equivalently, in the coefficient space, score improvement directions $\calK_V^\ast$ need to be a subset of metric improvement directions $\calK_Z^\ast$.
    The choice of $\bfV = [2] \in \bbR^{1 \times 1}$ yields the desired property $\calK_Z \subseteq \calK_V$.
    In fact, we get $\calK_Z = \calK_V$ and $\bfV \subseteq \bfZ$, and so all three matrix ranks have value $1$.

    Finally, we can recover $\bfA = [1, 0]$ such that $\bfV = \bfA \bfZ$, and design $S(\bff) = [1, 0] \cdot \bff = f_1$.
    It is easy to verify that this $S$ satisfies the improvement objective (we could also have chosen $\bfV = [1]$ previously to design $S(\bff) = [0, 1] \cdot \bff = f_2$).
    Hence, when the two metrics are perfectly positively correlated, choosing one for score design suffices.
\end{example}
    
\begin{example}[Two anti-correlated metrics$\implies$must choose both for score design]
\label{example:1__2d__anticorrelated_metrics}
    Performance metrics used by CMS can also be negatively correlated when a hospital must balance its effort to simultaneously improve all metrics.
    In the case of two \textit{perfectly} negative correlated metrics, \Cref{alg:1__design_strategy} suggests to use both metrics to design $S : \calF \to \bbR^2$, as no 1-dimensional score function can satisfy improvement objective.

    Consider two metrics---(i) pneumonia death rate and (ii) excessive antibiotic use---that have a negative correlation as improving on one degrades the other.
    Assume that these two metrics take values $\calF = \set{ \bff \in \bbR^2 \mid -f_1 - 2 f_2 = 1, -1 \le f_1 \le 1 }$, lying in a 1-dimensional affine subspace of $\bbR^2$ (\Cref{fig:1example_f_2d__anticorrelated}, red).
    Similar to \Cref{example:1__2d__correlated_metrics}, the subspace $\calL = \set{ \bff \mid -f_1 + 2f_2 = 0 }$ (\Cref{fig:1example_f_2d__anticorrelated}, blue) associated to $\affine(\calF)$ is 1-dimensional.
    But the rows of orthonormal basis $\bfZ \propto \begin{bmatrix}2 \\ -1 \end{bmatrix}$ generate cone $\calK_Z = \set{ 2 \lambda_1 - \lambda_2 \mid \lambda_1, \lambda_2 \ge 0 } = \bbR$, which contains a linear subspace within.
    This means that the metric improvement directions in the coefficient space are the dual cone $\calK_Z^\ast = \set{\bfzero}$, i.e., there are no non-trivial directions to simultaneously improve both metrics.

    To satisfy improvement objective, score improvement directions in the coefficient space $\calK_V^\ast$ need to be a subset of metric improvement directions $\calK_Z^\ast = \set{\bfzero}$, or equivalently $\calK_Z \subseteq \calK_V$.
    Hence, we choose $\bfV = \begin{bmatrix} 2 \\ -1 \end{bmatrix} \propto \bfZ$ with 2 rows.
    Note that $\bfV$ with just 1 row would generate either cone $\bbR_+$ or cone $-\bbR_+$, and fail to enclose cone $\calK_Z = \bbR$.
    Hence, all three matrix ranks have value 2 even though all movement directions $\calF_{\bff}$ lie in a 1-dimensional subspace $\calL$.

    Finally, we can recover $\bfA = \bfI_2$ such that $\bfV = \bfA \bfZ$ and design the trivial $S(\bff) = \bff$.
    Due to the perfect negative correlation in metrics, we must choose both in the score design.
\end{example}

\begin{example}[Restriction with monotonicity$\implies$higher dimensionality]
\label{example:matrix_rank__3d__circularcone}
    When the number of metrics is large, understanding correlations among them can be unintuitive.
    Hence, we rely on structure of polyhedral cones for score design, specifically improvement directions of scores $\calK_V^\ast$ and metrics $\calK_Z^\ast$ (in the coefficient space).
    We find that score function dimensionality $k$ under \restrictionCS and \restrictionLM restrictions can be much larger than under \restrictionL, as $\CSR, \CGR \gg \CR$.

    Consider the case of four metrics where two of them balance the other two, i.e., a toy example where performance metrics take values $\calF = \affine(\calF) = \set{ \bff \in \bbR^4 \mid [1, -1, 1, -1] \cdot \bff = 0 }$.
    Here the four metrics lie in a 3-dimensional linear subspace of $\bbR^4$ and $\calF = \affine(\calF) = \calL$. 
    Hence, three orthonormal vectors in $\bbR^4$ form a basis of $\calL$ such that the rows of $\bfZ$ generate the ``square'' cone $\calK_Z$ in $\bbR^3$ (\Cref{fig:example_matrix_rank__3d_square_cone}, red):
    \begin{align*}
        \bfZ &= \frac{1}{2} \cdot \begin{bmatrix}
            1 & 1 & 1\\
            1 & -1 & 1\\
            1 & -1 & -1\\
            1 & 1 & -1
        \end{bmatrix} \in \bbR^{4 \times 3}.
    \end{align*}

    For \restrictionCS and \restrictionLM restrictions, we need to find matrix $\bfV$ such that $\calK_V = \calK_Z$.
    As all rows of $\bfZ$ are \textit{extreme rays} of $\calK_Z$, matrix $\bfV$ must have four rows $\bfV = \bfI_4 \bfZ$ (any $\bfV$ with fewer rows would not \textit{generate} the square cone).
    Hence, $\CSR(\bfZ) = \CGR(\bfZ) = 4$.
    But for \restrictionL restriction that does not require monotonicity, rows of $\bfV$ need only ensure $\calK_Z \subseteq \calK_V$.
    The following matrix $\bfV$ with three rows that generates a ``triangular'' cone $\calK_V$ (\Cref{fig:example_matrix_rank__3d_square_cone}, blue) \textit{enclosing} the square cone $\calK_Z$:
    \begin{align*}
        \bfV &= \frac{1}{2} \cdot \begin{bmatrix}
            1 & 0 & 2\\
            1 & 3 & -1\\
            1 & -3 & -1\\
        \end{bmatrix} \qquad \text{and so } \bfV = \bfA \bfZ \text{ with } \bfA = \frac{1}{4} \cdot \begin{bmatrix}
            3 & 3 & -1 & -1\\
            3 & -3 & -1 & 5\\
            -3 & 3 & 5 & -1\\
        \end{bmatrix}.
    \end{align*}
    \begin{figure}[!h]
        \centering
        \begin{subfigure}{0.49\textwidth}
            \centering
            \input{tikz/cone_inR3_3d_square_cone}
            \caption{
            Rows of $\bfZ$ are extreme rays of the generated ``square'' cone $\calK_Z$.
            The square cone can be enclosed by a ``triangular'' cone $\calK_V$.
            }
            \label{fig:example_matrix_rank__3d_square_cone}
        \end{subfigure}\hfill
        \begin{subfigure}{0.49\textwidth}
            \centering
            \input{tikz/cone_inR3_3d_circular_cone}
            \caption{
            All rows of $\bfZ \in \bbR^{d \times 3}$ are extreme rays of the generated ``circular'' cone $\calK_Z$.
            The circular cone can be enclosed by a ``triangular'' cone $\calK_V$.
            }
            \label{fig:example_matrix_rank__3d_circular_cone}
        \end{subfigure}
        \caption{
        Side and top views of cones $\calK_Z$ (red) generated by rows of $\bfZ$, whose columns are orthonormal basis of 3-dimensional subspace $\calL$.
        As $\CSR$ and $\CGR$ require \textit{generating} $\calK_Z$ with $\calK_V$, the matrix ranks depend on the number of extreme rays of $\calK_Z$, which can be much higher than $\dim \affine(\calF) = 3$.
        On the other hand, $\CR$ only requires \textit{enclosing} $\calK_Z$ with $\calK_V$; and so is independent of the number of extreme rays.
        }
    \end{figure}
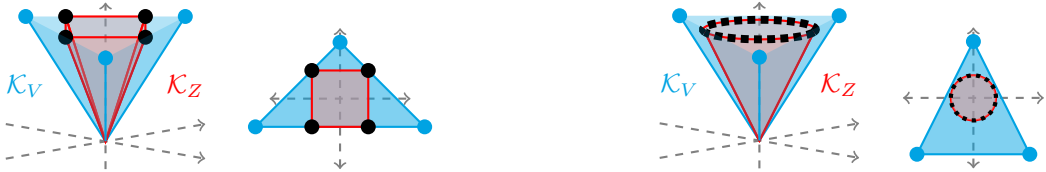
    
    Generally, $\CSR$ and $\CGR$ can be much larger than $\CR$ (\Cref{fig:example_matrix_rank__3d_circular_cone}).
    Since these three matrix ranks describe the dimensionality under the three restrictions (\Cref{thm:1__f_linear__suff}), restrictions that require monotonicity (\restrictionCS, \restrictionLM) lead to higher dimensionality in score design compared to \restrictionL.
    In other words, allowing negative values in matrix $\bfA$ can significantly reduce dimensionality of score design.
\end{example}

\begin{remark}[Competing metric improvement directions$\implies$higher dimensionality under \restrictionCS]
    When rows of $\bfZ$ generate cone $\calK_Z$ that is \textit{pointed}%
    \footnote{A cone $\calK$ is pointed if for all nonzero $\bfx \in \calK$, we have $-\bfx \notin \calK$. It is called \textit{non-pointed} otherwise.}%
    , we get $\CSR(\bfZ) = \CGR(\bfZ)$.
    But when cone $\calK_Z$ that is \textit{non-pointed}, we get $\CSR(\bfZ) > \CGR(\bfZ)$.
    $\calK_Z$ can be non-pointed when improving one metric degrades another, i.e., when metric improvement directions compete among themselves.
    In this setting, dimensionality under \restrictionCS is higher than that under \restrictionLM (see \Cref{example:matrix_rank__5d__nonpointed}).
\end{remark}

\paragraph*{Efficiently implementing \Cref{alg:1__design_strategy}}

Our proposed design strategy in \Cref{alg:1__design_strategy} can be efficiently implemented with algorithms that utilize the geometry of metrics $\calF$.
Elementary linear algebra operations can implement Lines~2~and~4 of \Cref{alg:1__design_strategy}, i.e., finding orthonormal basis $\bfZ$ and recovering $\bfA$ from $\bfV = \bfA \bfZ$.
It is also possible to efficiently implement Line~3, to find matrix $\bfV$ that attain the matrix ranks---$\ConeSubsetRank$, $\ConeGeneratingRank$, and $\ConeRank$.
We briefly discuss algorithms for Line~3 here, and point the reader to \Cref{sec:computing_matrix_ranks} for detailed algorithms.
These algorithms ensure that the full \Cref{alg:1__design_strategy} can be efficiently implemented.
We leverage a key property of polyhedral cones, \textit{pointedness}.

When the cone $\calK_Z$ generated from rows of $\bfZ$ is pointed, we can easily find $\bfV$ that attains the matrix ranks.
For $\ConeSubsetRank$, we can keep the rows of $\bfZ$ that are extreme rays of the polyhedral cone $\calK_Z$, as extreme rays minimally generate a pointed cone~\cite[Prop.~26.5.4]{border2022convex}.
$\ConeGeneratingRank$ turns out to be the same as $\ConeSubsetRank$, as every extreme ray of $\calK_Z$ is a row of matrix $\bfZ$~\cite[Prop.~26.5.4]{border2022convex}.
For $\ConeRank$, the matrix $\bfV$ attaining it must generate $\calK_V$ that encloses $\calK_Z$.
An intuitive procedure can find this $\bfV$: can scale rows of $\bfZ$ to lie on a hyperplane, and find a simplex that encloses the convex hull of scaled rows~\cite{gale1953inscribing}.

When the cone $\calK_Z$ is non-pointed, the cone contains a linear subspace within.
Here we can utilize the unique Minkowski decomposition of polyhedral cones into two orthogonal components: the maximal linear subspace within, and a pointed remnant~\cite[Sec.~8.2]{schrijver1998theory}.
Then, for all three matrix ranks, we can generate/enclose non-pointed cone $\calK_Z$, by generating/enclosing the two orthogonal components separately.

\subsection{Proposed design is minimal}
\label{sec:feasibility_improvement__nec}

\Cref{thm:1__f_linear__suff} states that dimensionalities determined by the three matrix ranks---$\ConeSubsetRank$, $\ConeGeneratingRank$, and $\ConeRank$---are sufficient for score design.
It turns out that these dimensionalities are also necessary under a mild assumption on $\calF$ (\Cref{thm:1__f_linear__nonempty_relint_f_with_origin__necsuff}).
Hence, \Cref{thm:1__f_linear__suff,thm:1__f_linear__nonempty_relint_f_with_origin__necsuff} together imply that the \textit{three matrix ranks exactly determine the minimal design problem for improvement objective}.

\begin{theorem}
\label{thm:1__f_linear__nonempty_relint_f_with_origin__necsuff}
    Assume metrics $\calF \subseteq \bbR^d$ have non-empty relative interior with respect to $\affine(\calF)$.
    Then the listed dimensionalities $k$ in \Cref{thm:1__f_linear__suff} are necessary.
\end{theorem}

We briefly discuss the implication of metrics $\calF$ having non-empty relative interior on satisying the improvement objective.
Such a set $\calF$ contains a center $\bff^\ast \in \calF$ where every direction in subspace $\calL$ is a positively-scaled movement direction from $\calF_{\bff^\ast}$.
Intuitively, all score improvement directions are movement directions in the coefficient space.
As a result, we get an equivalence between satisfying improvement in the ambient space and the coefficient space, i.e., satisfying improvement in \Cref{eqn:1__f_linear__restate_cones} is equivalent to satisfying $\calK_Z \subseteq \calK_V$.
See \Cref{thm_w_proof:1__f_linear__nonempty_relint_f_with_origin__necsuff} for the proof.

\begin{remark}
\label{remark:relint_discussion}
In \Cref{fig:examples_relint_f} we illustrate examples of $\calF$ and their relative interior.
$\calF$ having non-empty relative interior is a reasonable condition in practice, as performance metrics used by rating agencies are often correlated and not isolated points~\cite{bandalos2018measurement,cms_starratings,usnews_rankings,barclay2022concordance,kurian2021predicting,ash2012statistical}.
For instance, CMS uses percentage-rate-based metrics, such as condition-specific death rates, readmission rates, and screening rates~\cite{cms_starratings,cms_2024starratings_report}.
This leads to real-valued metrics $\calF = [0, 1]^d$, which has non-empty relative interior.
We note that, when the relative interior is \textit{empty}, dimensionality $k$ significantly less than listed values in \Cref{thm:1__f_linear__suff} can suffice (\Cref{prop_w_proof:1__f_linear__empty_relint__counterexample}).
\end{remark}

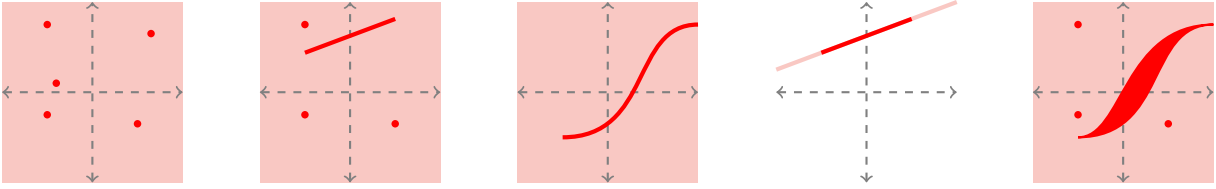
\begin{figure}[!h]
    \centering
    \begin{subfigure}{0.17\textwidth}
        \input{tikz/relint_empty_isolated}
    \end{subfigure}
    \hfill
    \begin{subfigure}{0.17\textwidth}
        \input{tikz/relint_empty_isolated_w_line}
    \end{subfigure}
    \hfill
    \begin{subfigure}{0.17\textwidth}
        \input{tikz/relint_empty_S}
    \end{subfigure}
    \hfill
    \begin{subfigure}{0.17\textwidth}
        \input{tikz/relint_nonempty_line}
    \end{subfigure}
    \hfill
    \begin{subfigure}{0.17\textwidth}
        \input{tikz/relint_nonempty_S}
    \end{subfigure}
    \caption{Examples of $\calF \subseteq \bbR^2$.
    The left three have empty relative interior, whereas the right two have non-empty relative interior with respect to $\affine(\calF)$, which is lightly shaded.
    }
    \label{fig:examples_relint_f}
\end{figure}

\begin{remark}[Choice of affine subspace and orthonormal basis]
    Our design strategy in \Cref{alg:1__design_strategy} can use \textit{any} orthonormal basis $\bfZ$ of the linear subspace $\calL_{\calH}$ associated with \textit{any} affine subspace $\calH$ containing metrics $\calF$.
    To design the \textit{minimal} $S : \calF \to \bbR^k$, we pick \textit{any} orthonormal basis of subspace $\calL$ associated with affine hull $\calH = \affine(\calF)$.
    This follows from \Cref{lemma_w_proof:matrix_rank__smallest_with_affine_invariant_to_rotation}, which states that three matrix ranks are (1) invariant to the choice of orthonormal basis for a fixed subspace $\calL_{\calH}$, and (2) minimized with the choice of $\calH = \affine(\calF)$.
\end{remark}

%% file: tikz/1example_f_2d_pointed.tex
\begin{tikzpicture}[thick,scale=0.6]
    \coordinate (O) at (0, 0);
    \coordinate (right) at (2, 0);
    \coordinate (left) at (-2, 0);
    \coordinate (top) at (0, 2);
    \coordinate (bottom) at (0, -2);

    \draw[<->,color=gray,dashed] (left) -- (right);
    \draw[<->,color=gray,dashed] (bottom) -- (top);

    \draw[ultra thick,color=red] (-1,-0) -- (1,1);
    \node[label={[text=red]above:$\calF$}] at (-1,0) {};
    
    \draw[ultra thick,color=Cerulean] (-2,-1) -- (2,1);
    \node[label={[text=Cerulean]below:$\calL$}] at (-1,-0.5) {};
\end{tikzpicture}

%% file: tikz/1example_f_2d_nonpointed.tex
\begin{tikzpicture}[thick,scale=0.6]
    \coordinate (O) at (0, 0);
    \coordinate (right) at (2, 0);
    \coordinate (left) at (-2, 0);
    \coordinate (top) at (0, 2);
    \coordinate (bottom) at (0, -2);

    \draw[<->,color=gray,dashed] (left) -- (right);
    \draw[<->,color=gray,dashed] (bottom) -- (top);

    \draw[ultra thick,color=red] (-1,0) -- (1,-1);
    \node[label={[text=red]below:$\calF$}] at (-1,0) {};
    
    \draw[ultra thick,color=Cerulean] (-2,1) -- (2,-1);
    \node[label={[text=Cerulean]above:$\calL$}] at (-1,0.5) {};
\end{tikzpicture}

%% file: tikz/cone_inR3_3d_square_cone.tex
\tdplotsetmaincoords{80}{45}
\tdplotsetrotatedcoords{-90}{180}{-90}
\begin{tikzpicture}[tdplot_main_coords,thick,scale=1.5]
    \coordinate (O) at (0, 0, 0);

    % axes
    \draw[->,color=gray,dashed] (-1.25,0,0) -- (1.25,0,0);
    \draw[->,color=gray,dashed] (0,-1.25,0) -- (0,1.25,0);
    \draw[->,color=gray,dashed] (0,0,-0.25) -- (0,0,1.25);

    % back face of triangle cone
    \draw[fill=Cerulean,fill opacity=0.3,draw=none] (O) -- (0,1,1) -- (-1,0,1) -- cycle;
    \draw[draw=Cerulean] (O) -- (0,1,1);
    \draw[draw=Cerulean] (O) -- (-1,0,1);

    % draw square cone
    \begin{scope}[thick,draw=red]
        % back face
        \draw[fill=Salmon,fill opacity=0.3] (O) -- (-0.25,0.75,1) -- (-0.75,0.25,1) -- cycle;
        % \draw[canvas is xy plane at z=1,fill=Salmon,fill opacity=0.3] (45-15:0.5) arc (45-15:225+15:0.5) -- (O) -- cycle;
        
        % left face
        \draw[fill=Salmon,fill opacity=0.3] (O) -- (-0.75,0.25,1) -- (0,-0.5,1) -- cycle;
        
        % right face
        \draw[fill=Salmon,fill opacity=0.3] (O) -- (-0.25,0.75,1) -- (0.5,0,1) -- cycle;
        
        % front face
        \draw[fill=Salmon,fill opacity=0.3] (O) -- (0,-0.5,1) -- (0.5,0,1) -- cycle;
    \end{scope}

    % draw black points on the square vertices
    \node[circle,inner sep=2pt,fill=black] at (-0.25,0.75,1) {};
    \node[circle,inner sep=2pt,fill=black] at (-0.75,0.25,1) {};
    \node[circle,inner sep=2pt,fill=black] at (0,-0.5,1) {};
    \node[circle,inner sep=2pt,fill=black] at (0.5,0,1) {};

    % left face of trianglular cone
    \draw[fill=Cerulean,fill opacity=0.2,draw=none] (O) -- (-1,0,1) -- (1,-1,1) -- cycle;
    % right face of trianglular cone
    \draw[fill=Cerulean,fill opacity=0.2,draw=none] (O) -- (0,1,1) -- (1,-1,1) -- cycle;
    \draw[draw=Cerulean] (O) -- (1,-1,1);

    \node[circle,inner sep=2pt,fill=Cerulean] at (-1,0,1) {};
    \node[circle,inner sep=2pt,fill=Cerulean] at (0,1,1) {};
    \node[circle,inner sep=2pt,fill=Cerulean] at (1,-1,1) {};
    
    \node[text=Cerulean] at (-0.5,-0.5,0.5) {$\calK_V$};
    \node[text=red] at (0.5,0.5,0.5) {$\calK_Z$};
\end{tikzpicture}
\quad
\begin{tikzpicture}[thick,scale=0.75]
    \coordinate (O) at (0, 0);
    \coordinate (right) at (1.25, 0);
    \coordinate (left) at (-1.25, 0);
    \coordinate (top) at (0, 1.25);
    \coordinate (bottom) at (0, -1.25);

    \draw[<->,color=gray,dashed] (left) -- (right);
    \draw[<->,color=gray,dashed] (bottom) -- (top);

    \draw[fill=Cerulean,fill opacity=0.5,draw=Cerulean] (0,1) -- (1.5,-0.5) -- (-1.5,-0.5) -- cycle;

    \draw[fill=Salmon,fill opacity=0.5,draw=red] (0.5,0.5) -- (0.5,-0.5) -- (-0.5,-0.5) -- (-0.5,0.5) -- cycle;

    \node[circle,inner sep=2pt,fill=Cerulean] at (0,1) {};
    \node[circle,inner sep=2pt,fill=Cerulean] at (1.5,-0.5) {};
    \node[circle,inner sep=2pt,fill=Cerulean] at (-1.5,-0.5) {};
    
    \node[circle,inner sep=2pt,fill=black] at (0.5,0.5) {};
    \node[circle,inner sep=2pt,fill=black] at (0.5,-0.5) {};
    \node[circle,inner sep=2pt,fill=black] at (-0.5,-0.5) {};
    \node[circle,inner sep=2pt,fill=black] at (-0.5,0.5) {};
\end{tikzpicture}

%% file: tikz/cone_inR3_3d_circular_cone.tex
\tdplotsetmaincoords{80}{45}
\tdplotsetrotatedcoords{-90}{180}{-90}
\begin{tikzpicture}[tdplot_main_coords,thick,scale=1.5]
    \coordinate (O) at (0, 0, 0);

    % axes
    \draw[->,color=gray,dashed] (-1.25,0,0) -- (1.25,0,0);
    \draw[->,color=gray,dashed] (0,-1.25,0) -- (0,1.25,0);
    \draw[->,color=gray,dashed] (0,0,-0.25) -- (0,0,1.25);

    % back face of triangle cone
    \draw[fill=Cerulean,fill opacity=0.3,draw=none] (O) -- (0,1,1) -- (-1,0,1) -- cycle;
    \draw[draw=Cerulean] (O) -- (0,1,1);
    \draw[draw=Cerulean] (O) -- (-1,0,1);

    % draw circular cone
    \begin{scope}[thick,draw=red]
        % back face
        \draw[canvas is xy plane at z=1,fill=Salmon,fill opacity=0.3] (45-15:0.5) arc (45-15:225+15:0.5) -- (O) -- cycle;
        
        % front face
        \draw[canvas is xy plane at z=1,fill=Salmon,fill opacity=0.3] (45-15:0.5) arc (45-15:-135+15:0.5) -- (O) -- cycle;
    \end{scope}

    % draw black points on the circular arc
    \begin{scope}[dotted, ultra thick, draw=black, line width=3pt]
        \draw[canvas is xy plane at z=1] (45-15:0.5) arc (45-15:225+15:0.5);
        \draw[canvas is xy plane at z=1] (45-15:0.5) arc (45-15:-135+15:0.5);
    \end{scope}

    % left face of trianglular cone
    \draw[fill=Cerulean,fill opacity=0.2,draw=none] (O) -- (-1,0,1) -- (1,-1,1) -- cycle;
    % right face of trianglular cone
    \draw[fill=Cerulean,fill opacity=0.2,draw=none] (O) -- (0,1,1) -- (1,-1,1) -- cycle;
    \draw[draw=Cerulean] (O) -- (1,-1,1);

    \node[circle,inner sep=2pt,fill=Cerulean] at (-1,0,1) {};
    \node[circle,inner sep=2pt,fill=Cerulean] at (0,1,1) {};
    \node[circle,inner sep=2pt,fill=Cerulean] at (1,-1,1) {};

    \node[text=Cerulean] at (-0.5,-0.5,0.5) {$\calK_V$};
    \node[text=red] at (0.5,0.5,0.5) {$\calK_Z$};
\end{tikzpicture}
\quad
\begin{tikzpicture}[thick,scale=0.75]
    \coordinate (O) at (0, 0);
    \coordinate (right) at (1.25, 0);
    \coordinate (left) at (-1.25, 0);
    \coordinate (top) at (0, 1.25);
    \coordinate (bottom) at (0, -1.25);

    \draw[<->,color=gray,dashed] (left) -- (right);
    \draw[<->,color=gray,dashed] (bottom) -- (top);

    \draw[fill=Cerulean,fill opacity=0.5,draw=Cerulean] (0,1) -- (1,-1) -- (-1,-1) -- cycle;

    \draw[fill=none,draw=red] (0,0) circle (0.4);
    \draw[fill=Salmon,fill opacity=0.5,dotted,draw=black,ultra thick] (0,0) circle (0.4);

    \node[circle,inner sep=2pt,fill=Cerulean] at (0,1) {};
    \node[circle,inner sep=2pt,fill=Cerulean] at (1,-1) {};
    \node[circle,inner sep=2pt,fill=Cerulean] at (-1,-1) {};
\end{tikzpicture}

%% file: tikz/relint_empty_isolated.tex
\begin{tikzpicture}[thick,scale=0.6]
    \coordinate (O) at (0, 0);
    \coordinate (right) at (2, 0);
    \coordinate (left) at (-2, 0);
    \coordinate (top) at (0, 2);
    \coordinate (bottom) at (0, -2);

    \draw[draw=none,fill=Salmon, opacity=0.5] (2,2) -- (-2,2) -- (-2,-2) -- (2,-2) -- cycle;
    
    \draw[<->,color=gray,dashed] (left) -- (right);
    \draw[<->,color=gray,dashed] (bottom) -- (top);

    \node[circle,inner sep=1pt,fill=red] at (-1, -0.5) {};
    \node[circle,inner sep=1pt,fill=red] at (1, -0.7) {};
    \node[circle,inner sep=1pt,fill=red] at (-1, 1.5) {};
    \node[circle,inner sep=1pt,fill=red] at (-0.8, 0.2) {};
    \node[circle,inner sep=1pt,fill=red] at (1.3,1.3) {};

\end{tikzpicture}

%% file: tikz/relint_empty_isolated_w_line.tex
\begin{tikzpicture}[thick,scale=0.6]
    \coordinate (O) at (0, 0);
    \coordinate (right) at (2, 0);
    \coordinate (left) at (-2, 0);
    \coordinate (top) at (0, 2);
    \coordinate (bottom) at (0, -2);

    \draw[draw=none,fill=Salmon, opacity=0.5] (2,2) -- (-2,2) -- (-2,-2) -- (2,-2) -- cycle;

    \draw[<->,color=gray,dashed] (left) -- (right);
    \draw[<->,color=gray,dashed] (bottom) -- (top);

    % line is y-2 / x-2 = 3/8
    \draw[ultra thick,color=red] (-1,0.875) -- (1,1.625);
    \node[circle,inner sep=1pt,fill=red] at (-1, -0.5) {};
    \node[circle,inner sep=1pt,fill=red] at (1, -0.7) {};
    \node[circle,inner sep=1pt,fill=red] at (-1, 1.5) {};
\end{tikzpicture}

%% file: tikz/relint_empty_S.tex
\begin{tikzpicture}[thick,scale=0.6]
    \coordinate (O) at (0, 0);
    \coordinate (right) at (2, 0);
    \coordinate (left) at (-2, 0);
    \coordinate (top) at (0, 2);
    \coordinate (bottom) at (0, -2);

    \draw[draw=none,fill=Salmon, opacity=0.5] (2,2) -- (-2,2) -- (-2,-2) -- (2,-2) -- cycle;

    \draw[<->,color=gray,dashed] (left) -- (right);
    \draw[<->,color=gray,dashed] (bottom) -- (top);

    \draw[ultra thick,color=red] (-1,-1) .. controls (1,-1) and (0.5,1.5) .. (2,1.5);
\end{tikzpicture}

%% file: tikz/relint_nonempty_line.tex
\begin{tikzpicture}[thick,scale=0.6]
    \coordinate (O) at (0, 0);
    \coordinate (right) at (2, 0);
    \coordinate (left) at (-2, 0);
    \coordinate (top) at (0, 2);
    \coordinate (bottom) at (0, -2);

    % line is y-2 / x-2 = 3/8
    \draw[ultra thick,color=Salmon, opacity=0.5] (-2,0.5) -- (2,2);

    \draw[<->,color=gray,dashed] (left) -- (right);
    \draw[<->,color=gray,dashed] (bottom) -- (top);

    \draw[ultra thick,color=red] (-1,0.875) -- (1,1.625);
\end{tikzpicture}

%% file: tikz/relint_nonempty_S.tex
\begin{tikzpicture}[thick,scale=0.6]
    \coordinate (O) at (0, 0);
    \coordinate (right) at (2, 0);
    \coordinate (left) at (-2, 0);
    \coordinate (top) at (0, 2);
    \coordinate (bottom) at (0, -2);

    \draw[draw=none,fill=Salmon, opacity=0.5] (2,2) -- (-2,2) -- (-2,-2) -- (2,-2) -- cycle;

    \draw[<->,color=gray,dashed] (left) -- (right);
    \draw[<->,color=gray,dashed] (bottom) -- (top);

    \filldraw[fill=red,color=red] (-1,-1) .. controls (1,-1) and (0.5,1.5) .. (2,1.5) -- (2, 1.5) .. controls (0,1.5) and (0,-1) .. (-1,-1);

    \node[circle,inner sep=1pt,fill=red] at (-1, -0.5) {};
    \node[circle,inner sep=1pt,fill=red] at (1, -0.7) {};
    \node[circle,inner sep=1pt,fill=red] at (-1, 1.5) {};
\end{tikzpicture}

%% file: obj__optimality.tex
\section{Minimal design problem for optimality objective}
\label{sec:feasibility_optimality}

We propose a surrogate score design for satisfying the optimality objective and discuss the minimality of our proposed design.
We use the standard definition of pareto-optimality.

\begin{definition}
\label{defn:pareto_opt}
Point $\bff \in \calF$ is pareto-optimal for maximizing $S$ if no other point in $\calF$ both improves $S(\bff)$ in all coordinates and strictly improves $S(\bff)$ in at least one coordinate.
\begin{align*}
    \ParetoOpt(S) &\coloneqq \set{\bff \in \calF \mid \text{for all } \bff' \in \calF, \text{ either } S(\bff') \not\ge S(\bff) \text{ or } S(\bff') = S(\bff) }.
\end{align*}
We write $\ParetoOpt(\calF)$ to denote the pareto-optimal points in $\calF$ w.r.t. the identity map.
\end{definition}

We simplify the optimality objective in \Cref{eqn:definition_optimality_general}---$\ParetoOpt(S) \subseteq \ParetoOpt(\calF)$---using movement directions $\calF_{\bff}$ at center $\bff$, score improvement directions $\calK_A^\ast$, and metric improvement directions $\calK_I^\ast$.
Intuitively, score function $S : \bff \mapsto \bfA \bff$ satisfies optimality if and only if movement directions $\calF_{\bff}$ that are \textit{non-strict score improvement directions} are also \textit{non-strict metric improvement directions}:
\begin{align}
\label{eqn:2__f_linear__restate_cones}
    \text{Optimality} &\iff \set{ \bff \in \calF \mid \calF_{\bff} \subseteq \stcomplement{(\calK_A^\ast)} \cup \ker \bfA } \subseteq \set{ \bff \in \calF \mid \calF_{\bff} \subseteq \stcomplement{(\calK_I^\ast)} \cup \ker \bfI }.
\end{align}

\subsection{Design proposal for optimality objective}
\label{sec:feasibility_optimality__suff}

We propose a score design in \Cref{alg:2__design_strategy} with dimensionalities given in \Cref{thm:2__f_linear__suff}.
We note that dimensionalities for score design are much smaller for the optimality objective than for the improvement objective (\Cref{thm:1__f_linear__suff}).
Specifically, for \restrictionLM and \restrictionL restrictions, a 1-dimensional score function $S : \calF \to \bbR$ suffices to satisfy optimality whereas multi-dimensional function $S$ is necessary for improvement (\Cref{thm:1__f_linear__nonempty_relint_f_with_origin__necsuff}).
This suggests that the optimality objective is significantly weaker than the improvement objective.

\begin{theorem}
    \label{thm:2__f_linear__suff}
    For each design restriction, there exists $S : \calF \to \bbR^k$, designed using \Cref{alg:2__design_strategy}, that satisfies the optimality objective with the following dimensionalities.
    \vspace{-0.2em}
    \noindent
    \begin{center}
        \begin{tabular}{l|c}
             & Dimensionality $k \ge$ \\
            \midrule
            \restrictionCS & $\dim \affine(\calF)$ \\
            \restrictionLM & $1$ \\
            \restrictionL & $1$ \\
        \end{tabular}
    \end{center}
    \vspace{-2em}
\end{theorem}
\begin{algorithm}[!h]
    \begin{algorithmic}[1]
        \State Given: $\calF$ and a design restriction.
        \If{Design restriction is \restrictionLM or \restrictionL}
            \State Design $S(\bff) = \bfa \cdot \bff$ with any positive vector $\bfa$.
        \ElsIf{Design restriction is \restrictionCS}
            \State Find $\bfZ$ whose columns are an orthonormal basis of subspace $\calL$ associated with $\affine(\calF)$.
            \State Let $\bfV$ be linearly independent rows of $\bfZ$.
            \State Find $\bfA$ that satisfies $\bfV = \bfA \bfZ$ and design $S : \bff \mapsto \bfA \bff$.
        \EndIf
    \end{algorithmic}
    \caption{Design strategy for optimality objective}
    \label{alg:2__design_strategy}
\end{algorithm}

For \restrictionLM and \restrictionL restrictions, the minimal design is straightforward: design $S : \bff \mapsto \bfa \cdot \bff$ using any vector $\bfa > \bfzero$~\cite{zadeh1963optimality}.
For \restrictionCS restriction, we utilize an isomorphism between movement directions $\calF_{\bff}$ and their coefficients $\calC_{\bff} \subseteq \bbR^r$ w.r.t. orthonormal basis $\bfZ \in \bbR^{d \times r}$ of subspace $\calL$ associated with $r$-dimensional $\affine(\calF)$.
The columns of $\bfZ$ span subspace $\calL$ and its rows correspond to coordinates of movement directions $\calF_{\bff}$.
Using this isomorphism, choosing $r$ linearly independent rows of $\bfZ$ as rows of $\bfV$ suffices to satisfy the optimality objective.
As $\bfV \subseteq \bfZ$, we can find $\bfA \in \bbR^{r \times d}$ with 1-hot rows such that $\bfV = \bfA \bfZ$, and design $S : \bff \mapsto \bfA \bff$ that satisfies the \restrictionCS restriction.
We include the proof in \Cref{thm_w_proof:2__f_linear__suff}.

\subsection{Discussion of minimality of proposed design}
\label{sec:feasibility_optimality__nec}

While our proposed design for improvement objective is minimal when $\calF$ has non-empty relative interior (\Cref{thm:1__f_linear__nonempty_relint_f_with_origin__necsuff}), our design for the optimality objective is \textit{not necessarily} minimal under the same condition on $\calF$.
The challenge is that $\ParetoOpt(\calF)$, the optimal trade-off surface~\cite{boyd2004convex}, depends on the boundary of $\calF$.
To demonstrate this, we give three examples of $d$-dimensional $\calF$ with non-empty relative interior---for one of the examples dimensionality $k = \dim \affine(\calF)$ is necessary for satisfying optimality under \restrictionCS, whereas for the other two examples, a 1-dimensional $S$ suffices.
See \Cref{prop_w_proof:2__f_linear__necessary_subset_examples} for the proof.

\begin{proposition}
\label{prop:2__f_linear__necessary_subset_examples}
    Consider designing $S : \calF \to \bbR^k$ to satisfy optimality objective.
    \begin{enumerate}
        \item For $\calF = \set{ \bff \in \bbR^d \mid \norm{\bff}_1 \le 1 }$, $k \ge 1$ is necessary and sufficient for all design restrictions.
        
        \item For $\calF = \set{ \bff \in \bbR^d \mid \norm{\bff}_2 \le 1 }$, $k \ge 1$ is necessary and sufficient for all design restrictions.
        
        \item For $\calF = \set{ \bff \in \bbR^d \mid \norm{\bff}_\infty \le 1 }$, $k \ge d$ is necessary and sufficient for \restrictionCS.
        Moreover, $k \ge 1$ is necessary and sufficient for the \restrictionLM and \restrictionL restrictions.
    \end{enumerate}
\end{proposition}

%% file: obj__both.tex
\section{Minimal design problem for both objectives simultaneously}
\label{sec:feasibility_both}

So far we have separately analyzed the minimal design problems for improvement and optimality objectives.
We now give results for simultaneously satisfying both objectives.

First, we establish a relationship between the improvement and optimality objectives.
\commentAnmol{Actually, a stronger result is true: Improvement + Monotone $\implies \ParetoOpt(S) = \ParetoOpt(\calF)$, a stronger objective than current optimality (and what Nati wanted to understand). So I think that satisfying the improvement objective is the important bottleneck! But don't know if this is minimal design for $PO(S) = PO(F)$ objective.}
This result holds even for score functions $S$ that are not linear in $\calF$.

\begin{theorem}
\label{thm:improvement_monotone__imply__optimality}
    Let $S : \calF \to \bbR^k$ be monotone in $\calF$.
    If $S$ satisfies improvement, then $S$ satisfies optimality.
\end{theorem}
\begin{proof}
    Let score function $S : \calF \to \bbR^k$ be monotone in $\calF$ and satisfy improvement.
    Hence, for all $\bff, \bff' \in \calF$ we have $S(\bff') \ge S(\bff) \iff \bff' \ge \bff$, i.e., the function $S$ preserves the ordering on set $\calF$. 
    We prove by contradiction that such an $S$ satisfies optimality.
    Assume that $\bff^\ast \in \ParetoOpt(S)$ but $\bff^\ast \notin \ParetoOpt(\calF)$.
    That is, there exists $\bff \in \calF$ such that $\bff \ge \bff^\ast$ and $\bff \neq \bff^\ast$.
    Because $S$ preserves the ordering, it must be that $S(\bff) \ge S(\bff^\ast)$ and $S(\bff) \neq S(\bff^\ast)$, which means that $\bff^\ast \notin \ParetoOpt(S)$ and contradicts our assumption.
\end{proof}

We utilize \Cref{thm:improvement_monotone__imply__optimality} to design $S$ that simultaneously satisfies both objectives.
As $S$ is monotone in $\calF$ under \restrictionCS and \restrictionLM restrictions, it suffices to design $S$ that satisfies the improvement objective.
We include the proof in \Cref{cor_w_proof:1_2__f_linear__suff}.

\begin{corollary}
    \label{cor:1_2__f_linear__suff}
    Let columns of $\bfZ$ be an orthonormal basis of linear subspace $\calL$ associated with $\affine(\calF)$.
    For each design restriction, there exists score function $S : \calF \to \bbR^k$ that simultaneously satisfies improvement and optimality objectives with following dimensionalities.
    
    \noindent
    \begin{center}
        \begin{tabular}{l|c}
             & Dimensionality $k \ge$ \\
            \midrule
            \restrictionCS & $\ConeSubsetRank(\bfZ)$ \\
            \restrictionLM & $\ConeGeneratingRank(\bfZ)$ \\
            \restrictionL & $\ConeGeneratingRank(\bfZ)$ \\
        \end{tabular}
    \end{center}
    Moreover, for \restrictionCS and \restrictionLM restrictions, the score design is minimal when $\calF$ has non-empty relative interior.
\end{corollary}

\begin{remark}
    For simultaneously satisfying both objectives under \restrictionL restriction, dimensionality $k = \CR(\bfZ)$ is necessary, when $\calF$ has non-empty relative interior (\Cref{thm:1__f_linear__nonempty_relint_f_with_origin__necsuff}).
    \Cref{cor:1_2__f_linear__suff} states that $k = \CGR(\bfZ)$ is sufficient, and $\CGR \gg \CR$ in general (\Cref{example:matrix_rank__3d__circularcone}).
    We leave to future work to close this gap between necessary and sufficient dimensionality.
\end{remark}

%% file: computing_matrix_ranks.tex
\section{Algorithms to compute matrix ranks}
\label{sec:computing_matrix_ranks}
We present polynomial-time algorithms that, on input $\bfW$, find $\bfV$ that ``attain'' matrix ranks.

For designing scores that the satisfy improvement objective, \Cref{alg:1__design_strategy} calls these algorithms to find $\bfV$ that attain the matrix ranks of $\bfZ$, whose columns are an orthonormal basis of linear subspace $\calL$ associated with performance metrics $\calF$.
The number of rows of $\bfV$ correspond to values of the matrix ranks and the smallest dimensionality $k$ of the designed score function $S : \calF \to \bbR^k$.
The columns of $\bfV$ span the space of $\calL$.
As described in the proof of \Cref{thm:1__f_linear__suff}, we can recover matrix $\bfA$ from $\bfV = \bfA \bfZ$ to design $S : \bff \mapsto \bfA \bff$.

These algorithms to find $\bfV$ augment work in computational geometry on manipulating polyhedral cones~\cite{dula1998algorithm,lopez2011algorithm,wets1967algorithms,orourke1986optimal,toussaint1983solving}, and so are of independent technical interest.
We end this section by comparing $\ConeRank$ with $\NonNegativeRank$, which is extensively studied in the context of non-negative matrix factorization~\cite{gillis2020nonnegative,arora2012computing,vavasis2010complexity,donoho2003does,gillis2013fast,kumar2013fast}.

\subsection{The three matrix ranks}
\label{sec:matrix_rank_examples}

We first illustrate gaps between the three matrix ranks, and build intuition about geometric properties of polyhedral cones that are crucial for designing and understanding the algorithms in \Cref{sec:algorithm__compute_csr,sec:algorithm__compute_cgr,sec:algorithm__compute_cr}.

\begin{definition}[Three matrix ranks]
    \label{defn:matrix_ranks}
    For a matrix $\bfW \in \bbR^{m \times n}$ whose rows generate cone $\calK_W$, we define $\ConeSubsetRank \, (\CSR), \ConeGeneratingRank \,  (\CGR)$, and $\ConeRank  \, (\CR)$:
    \begin{align}
        \begin{split}
            \CSR(\bfW) &\coloneqq \min_k \set{ k \mid \calK_W = \calK_V \text{ for some } \bfV \in\bbR^{k \times n} \text{ such that } \bfV \subseteq \bfW }.\\
            \CGR(\bfW) &\coloneqq \min_k \set{ k \mid \calK_W = \calK_V \text{ for some } \bfV \in\bbR^{k \times n} }.\\
            \CR(\bfW) &\coloneqq \min_k \set{ k \mid \calK_W \subseteq \calK_V \text{ for some } \bfV \in\bbR^{k \times n} }.
        \end{split}
    \end{align}
\end{definition}

The following relationship between the matrix ranks directly follows from the definitions, which place increasingly fewer restrictions on $\bfV$.

\begin{proposition}
\label{prop:matrix_rank_relationships}
    For $\bfW \in \bbR^{m \times n}$, $m \ge \CSR(\bfW) \ge \CGR(\bfW) \ge \CR(\bfW) \ge \rank \bfW$.
\end{proposition}

\begin{example}[Matrix ranks in $\bbR^2$]
\label{example:matrix_rank__2d}
    Assume $\bfW\in \bbR^{m \times n}$ with $n = 2$ columns.
    In each example, the $m$ rows of $\bfW$ are $m$ points in $\bbR^2$, generating cone $\calK_W \subseteq \bbR^2$.
    \begin{enumerate}[(a)]
        \item {$\CSR = \CGR = \CR = r =1$.} $\calK_W$ is a 1-dimensional \textit{pointed} cone\footnote{A cone $\calK$ is called \textit{pointed} if for all nonzero $\bfx \in \calK$, we have $-\bfx \notin \calK$. It is called \textit{non-pointed} otherwise.}.
        Only one point inside the cone is enough to generate it, and hence $\CSR(\bfW) = \CGR(\bfW) = \CR(\bfW) = 1$.
        In case of $\CSR(\bfW)$, the chosen point is restricted to be either the two points in $W$. 
     
        \item {$\CSR = \CGR = \CR > r=1$.} $\calK_W$ is a 1-dimensional linear subspace and hence a \textit{non-pointed} cone.
        Two points inside $\calK_W$ are sufficient and necessary to either generate or enclose $\calK_W$.
        Each of the two points should be on the opposite sides of the origin.
    
        \item {$\CSR = \CGR = \CR = r=2$.} $\calK_W$ is a 2-dimensional pointed cone.
        One of the points in $\bfW$ is inside the cone generated by the other two.
        The two points on the boundary of $\calK_W$ are necessary and sufficient to either generate or enclose $\calK_W$. 
    
        \item {$\CSR = \CGR = \CR > r=2$.} $\calK_W$ is a halfspace, a 2-dimensional non-pointed cone.
        This halfspace is the set addition of a 1-dimensional linear subspace and a 1-dimensional pointed cone orthogonal to the subspace.
        To generate or enclose $\calK_W$, it is sufficient and necessary to pick 2 points in the subspace, and a point for the halfspace direction.
    
        \item {$\CSR > \CGR = \CR > r=2$.} $\calK_W$ is the linear subspace $\bbR^2$.
        For $\CSR$, any choice of 3 points of $\bfW$ will generate a halfspace instead of $\calK_W$, and so $\CSR(\bfW) = m = 4$.
        But in the case of $\CGR$ and $\CR$, we can choose any 3 points in $\bbR^2$ to form a triangle containing origin.
        These 3 points generate and enclose $\calK_W$, and so $\CGR(\bfW) = \CR(\bfW) = 3$.
    \end{enumerate}
    \begin{figure}[!h]
        \centering
        \begin{subfigure}{0.18\textwidth}
            \input{tikz/cone_inR2_1d_pointed}
            \caption{}
        \end{subfigure}
        \hfill
        \begin{subfigure}{0.18\textwidth}
            \input{tikz/cone_inR2_1d_nonpointed}
            \caption{}
        \end{subfigure}
        \hfill
        \begin{subfigure}{0.18\textwidth}
            \input{tikz/cone_inR2_2d_pointed}
            \caption{}
        \end{subfigure}
        \hfill
        \begin{subfigure}{0.18\textwidth}
            \input{tikz/cone_inR2_2d_nonpointed}
            \caption{}
        \end{subfigure}
        \hfill
        \begin{subfigure}{0.18\textwidth}
            \input{tikz/cone_inR2_2d_full}
            \caption{}
        \end{subfigure}
    \end{figure}
\end{example}

\paragraph*{Leveraging pointedness of $\calK_W$}

\Cref{example:matrix_rank__2d,example:matrix_rank__3d__circularcone} introduce a key property of $\calK_W$, \textit{pointedness}.
We find that pointedness of $\calK_W$ enables us to intuitively compute the matrix ranks.
Recall that cone $\calK_W$ is non-pointed when there exists nonzero $\bfx \in \calK_W$ such that $-\bfx \in \calK_W$.
Thus, as proved in \Cref{lemma:alg__is_cone_pointed}, $\calK_W$ is non-pointed if and only if \ref{eqn:LP__pointed} is feasible.
\begin{align}
\label{eqn:LP__pointed}
\tag{LP-1}
    \min_{\bflambda \in \bbR^m} 1 \quad \text{s.t. } \bfzero = \bflambda \bfW, \quad \bflambda \ge \bfzero, \quad \bflambda \cdot \bfone = 1.
\end{align}

\paragraph*{Decomposing non-pointed $\calK_W$}
The maximal linear subspace inside a cone $\calK$ is called its lineality space, $\lineality(\calK) \coloneqq \calK \cap (-\calK)$, which is a non-trivial linear subspace when $\calK$ is non-pointed.
So when $\calK_W$ is non-pointed, we compute the matrix ranks by first decomposing $\calK_W$ into its lineality space and a pointed remnant (see \Cref{fig:cone_decomposition_examples} for examples).
We denote $\calL = \lineality(\calK_W)$ and $\ell = \dim \calL$ throughout this section.

\begin{lemma}[{Schrijver~\cite[Sec.~8.2]{schrijver1998theory}}]
\label{lemma:restated__cone_decomposition}
    A cone $\calK$ can be uniquely decomposed as $\calK = \calL + \calK_P$, where $\calL = \lineality(\calK) \coloneqq \calK \cap (-\calK)$ is the maximal linear subspace inside $\calK$, and $\calK_P = \calL^\perp \cap \calK$ is a pointed cone.
\end{lemma}

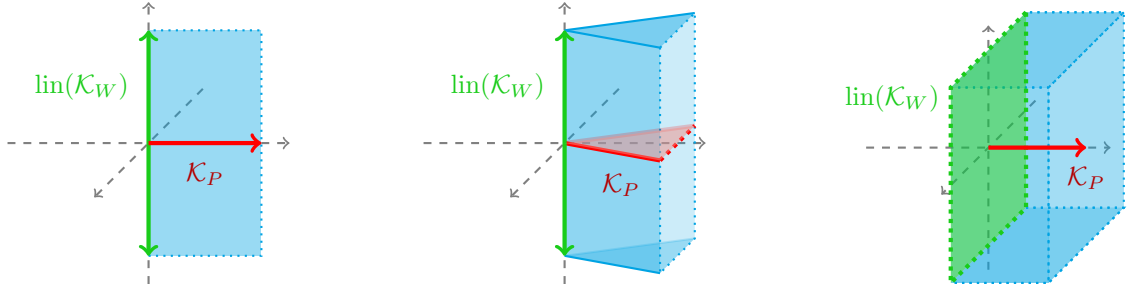
\begin{figure}[!htbp]
    \centering
    \begin{subfigure}{0.33\textwidth}
        \centering
        \input{tikz/decompose_cone_inR3_plane}
    \end{subfigure}\hfill
    \begin{subfigure}{0.33\textwidth}
        \centering
        \input{tikz/decompose_cone_inR3_wedge}
    \end{subfigure}\hfill
    \begin{subfigure}{0.33\textwidth}
        \centering
        \input{tikz/decompose_cone_inR3_halfspace}
    \end{subfigure}
    \caption{Decomposing non-pointed cones $\calK_W$ into lineality space $\lineality(\calK_W)$ and pointed $\calK_P$.}
    \label{fig:cone_decomposition_examples}
\end{figure}

\subsection{Computing $\ConeSubsetRank$}
\label{sec:algorithm__compute_csr}

This is the minimum number of rows of $\bfW$ that generate the cone $\calK_W$.

\paragraph*{Pointed $\calK_W$}
When $\calK_W$ is pointed, we compute $\CSR(\bfW)$ using a greedy algorithm.
We first make two observations.
First, the rows of $\bfV^\ast$ attaining $\CSR(\bfW)$ are positively independent, i.e., no row $\bfv_i$ of $\bfV^\ast$ can be expressed as a nonnegative combination of other rows of $\bfV^\ast$ (see \Cref{claim:csr__property__pos_independent}).
Second, when $\calK_W$ is pointed, any two matrices $\bfU, \bfV \subseteq \bfW$ with positively independent rows and $\calK_U = \calK_V = \calK_W$ have equal number of rows (see \Cref{lemma:csr__property__pointed_positive_basis_has_same_size}).

We sequentially remove rows of $\bfW$ that can be expressed as a nonnegative combination of other rows (we use a Linear Program to check this condition).
Removing such rows does not alter the generated cone $\calK_W$, and we obtain matrix $\bfV \subseteq \bfW$ with positively independent rows and $\calK_W = \calK_V$ (proved in \Cref{lemma:alg__csr__pointed}).
Both $\bfV^\ast$ and the output $\bfV$ have rows chosen from $\bfW$, generate $\calK_W$, and have positively independent rows.
Hence, $\bfV$ attains $\CSR(\bfW)$.

\inbox{%
\textbf{Sketch of \textsc{GetCSR-Pointed} (details in \Cref{alg:csr__pointed})}
\begin{algorithmic}[1]
    \State Initialize $\bfV = \bfW$
    \While{there exists row $\bfv_i$ of $\bfV$ such that $\bfv_i \in \cone(V \setminus \set{\bfv_i} )$}
    \Comment{Use Alg.~\ref{alg:is_in_cone}}
        \State Remove row $\bfv_i$ from $\bfV$
    \EndWhile
    \State \Return $\bfV$
\end{algorithmic}
}

\paragraph*{Non-pointed $\calK_W$}

When $\calK_W$ is non-pointed, it has decomposition $\calK_W = \calL + \calK_P$ where the lineality space $\calL$ is a non-trivial subspace.
Denote with $\bfWinL$ the rows of $\bfW$ lying inside $\calL$, and with $\bfWnotinL$ those lying outside.
\Cref{lemma:csr__property_of_decomposition} shows that any $\bfV \in \bbR^{k \times n}$ with $\bfV \subseteq \bfW$ and $\calK_W = \calK_V$ has number of rows $k \ge \CSR(\bfWinL) + \CSR(\tilde \bfW)$, where $\tilde \bfW$ are projections of rows of $\bfWnotinL$ onto orthogonal complement $\calL^\perp$.
This is because $\bfWinL$ generates $\calL$ (see \Cref{lemma:cone__lineality_space__lineal_points}) and $\tilde \bfW$ generates the pointed cone $\calK_P$ (see \Cref{lemma:cone__lineality_space__nonlineal_points}).

Following this observation, we find $\bfV$ that attains $\CSR(\bfW)$ by decomposing $\calK_W$ and finding matrices that attain $\CSR$ of $\bfWinL$ and $\tilde \bfW$ separately.
First, we find $\bfVinL$ that attains $\CSR(\bfWinL)$ using \Cref{alg:csr__subspace}, which enumerates over subsets of $\bfWinL$ that positively span $\calL$.
This enumeration takes $\poly(\ell,n) \cdot m^{O(\ell)}$ time where $\ell = \dim \calL$ (proved in \Cref{lemma:alg__csr__subspace}).
Second, as $\tilde \bfW$ generates the pointed cone $\calK_P$, we find $\tilde \bfV$ that attains $\CSR(\tilde \bfW)$ using the greedy algorithm for pointed cones.
And then we recover rows $\bfVnotinL \subseteq \bfWnotinL$, whose projections are rows of $\tilde \bfV$.
Together, $\bfV = [\bfVinL; \bfVnotinL]$ attains $\CSR(\bfW)$, proved in \Cref{lemma:alg__csr}.

\inbox{%
\textbf{Sketch of \textsc{GetCSR} when $\calK_W$ is non-pointed (details in \Cref{alg:csr})}
\begin{algorithmic}[1]
    \State Decompose $\calK_W = \calL + \calK_P$
    \Comment{Use Alg.~\ref{alg:cone_minkowski_decomposition}}
    \State Partition $W$ into $\WinL = \set{ \bfw_i \in W \mid \bfw_i \in \calL }$ and $\WnotinL = \set{ \bfw_i \in W \mid \bfw_i \notin \calL }$.
    \State Find $\bfVinL$ that attains $\CSR(\bfWinL)$
    \Comment{Use Alg.~\ref{alg:csr__subspace}}
    \State Find $\tilde \bfV$ that attains $\CSR(\tilde \bfW)$ where $\tilde W = \proj_{\calL^\perp} (\WnotinL)$
    \Comment{Use Alg.~\ref{alg:csr__pointed}}
    \State Recover $\bfVnotinL \subseteq \bfWnotinL$ such that $\tilde V$ are projections of $\VnotinL$
    \State \Return $\bfV = [\bfVinL; \bfVnotinL]$
\end{algorithmic}
}

\subsection{Computing $\ConeGeneratingRank$}
\label{sec:algorithm__compute_cgr}

This is the minimum number of vectors in $\bbR^n$, also called \textit{frame}, that generate the cone $\calK_W$.

\paragraph*{Pointed $\calK_W$}
It turns out that when $\calK_W$ is pointed, the frame of $\calK_W$ are actually vectors in $W$.
This was observed by Border~\cite[Prop.~26.5.4]{border2022convex} (restated in \Cref{lemma:restated__cone_pointed_extreme_rays_are_subset}) and Nemirovski~\cite[Cor.~2.4.2]{nemirovski2023linear}.
We thus get the following result, proved in \Cref{lemma_w_proof:cgr__pointed}.
As an immediate consequence, when $\calK_W$ is pointed, we can find $\bfV$ that attains $\CGR(\bfW)$ using the algorithm for $\ConeSubsetRank$ in \Cref{sec:algorithm__compute_csr}.

\begin{lemma}
\label{lemma:cgr__pointed}
    When $\calK_W$ is pointed, $\ConeGeneratingRank(\bfW) = \ConeSubsetRank(\bfW)$.
\end{lemma}

\paragraph*{Non-pointed $\calK_W$}
Let $\calK_W = \calL + \calK_P$ be the decomposition of $\calK_W$ with $\calL \neq \set{\bfzero}$.
Analogous to $\ConeSubsetRank$, we observe in \Cref{lemma:cgr__property_of_decomposition} that any $\bfV \in \bbR^{k \times n}$ with $\calK_W = \calK_V$ has number of rows $k \ge \CGR(\bfWinL) + \CGR(\tilde \bfW)$.
Here $\bfWinL$ are rows of $\bfW$ lying inside $\calL$ and $\tilde \bfW$ are projections onto $\calL^\perp$ of those lying outside.

Similar to the approach for $\ConeSubsetRank$, we find $\bfV$ that attains $\CGR(\bfW)$ by decomposing $\calK_W$ and finding matrices that attain $\CGR$ of $\bfWinL$ and $\tilde \bfW$ separately.
Note that rows of $\bfWinL$ generate the $\ell$-dimensional lineality space $\calL$.
A classical result by Davis~\cite{davis1954theory} (restated in \Cref{lemma:conerank__subspace_sufficient}) states that there exist $\ell+1$ vectors to generate $\calL$: choose an orthonormal basis $\bfz_1, \dots, \bfz_\ell$ of $\calL$ and let $\bfz_0 = - (\bfz_1 + \cdots + \bfz_\ell)$ to generate $\calL$.
In fact, $\ell+1$ vectors are necessary to do so, as shown in \Cref{lemma:conerank__nonpointed_necessary}.
Hence, $\tilde \bfZ = [\bfz_0; \bfz_1; \dots; \bfz_\ell]$ attains $\CGR(\bfWinL)$.
On the other hand, as $\tilde W = \proj_{\calL^\perp} (\WnotinL)$ generates the pointed cone $\calK_P$, \Cref{lemma:cgr__pointed} states that $\CGR(\tilde \bfW) = \CSR(\tilde \bfW)$.
So we find $\tilde \bfV$ that attains $\CGR(\tilde \bfW)$ using the greedy algorithm for $\ConeSubsetRank$ in \Cref{sec:algorithm__compute_csr}.
Together, $\bfV = [\tilde \bfZ, \tilde \bfV]$ attains $\CGR(\bfW)$, proved in \Cref{lemma:alg__cgr}.

\inbox{%
\textbf{Sketch of \textsc{GetCGR} when $\calK_W$ is non-pointed (details in \Cref{alg:cgr})}
\begin{algorithmic}[1]
    \State Decompose $\calK_W = \calL + \calK_P$ where $\ell = \dim \calL$
    \Comment{Use Alg.~\ref{alg:cone_minkowski_decomposition}}
    \State Partition $W$ into $\WinL = \set{ \bfw_i \in W \mid \bfw_i \in \calL }$ and $\WnotinL = \set{ \bfw_i \in W \mid \bfw_i \notin \calL }$.
    \State Find $\tilde \bfZ$ with $\ell+1$ rows that attains $\CGR(\bfWinL)$
    \State Find $\tilde \bfV$ that attains $\CSR(\tilde \bfW)$ where $\tilde W = \proj_{\calL^\perp} (\WnotinL)$
    \Comment{Use Alg.~\ref{alg:csr__pointed}}
    \State \Return $\bfV = [\tilde \bfZ; \tilde \bfV]$
\end{algorithmic}
}

\subsection{Computing $\ConeRank$}
\label{sec:algorithm__compute_cr}

This is the minimum number of vectors in $\bbR^n$ whose cone \textit{encloses} the cone $\calK_W$.

\paragraph*{Pointed $\calK_W$}
Let $\calK_W$ be an $r$-dimensional cone, i.e., $r = \rank \bfW$.
As vectors $W$ generate a pointed cone, the convex hull $\conv(W)$ does not contain the origin, as shown in \Cref{lemma:separating_hyperplane_theorem_lower_dim}.
So we can find an $(r-1)$-dimensional hyperplane $\bfw^\ast \cdot \bfx = b$ that strictly separates $\conv(W)$ from the origin.
As $\conv(W)$ and origin lie on opposite sides of this hyperplane, we can positively scale each vector $\bfw_i \in W$ to get $\bfu_i \propto \bfw_i$ on the hyperplane.
A classical result by Gale~\cite{gale1953inscribing} (restated in \Cref{lemma:restated__conerank__pointed__convex_simplex_cover}) says that the $r$-dimensional $\conv(U)$ can be enclosed within an $(r-1)$-simplex in the hyperplane.
We prove in \Cref{lemma:alg__cr__pointed} that the cone generated by the $r$ vertices of this $(r-1)$-simplex encloses $\calK_W$.
\Cref{example:matrix_rank__3d__circularcone} illustrated this construction.

\inbox{%
\textbf{Sketch of \textsc{GetCR-Pointed} (details in \Cref{alg:cr__pointed})}
\begin{algorithmic}[1]
    \State Find $(r-1)$-dimensional (strictly) separating hyperplane between $\conv(W)$ and origin
    \State Scale each $\bfw_i$ to get $\bfu_i \propto \bfw_i$ on the hyperplane
    \State In the hyperplane, find $(r-1)$-simplex enclosing $\conv(U)$
    \State \Return $\bfV$ whose rows are $r$ vertices of the simplex
\end{algorithmic}
}

\paragraph*{Non-pointed $\calK_W$}
After decomposing $\calK_W = \calL + \calK_P$ using \Cref{alg:cone_minkowski_decomposition}, we partition rows of $\bfW$ into $\bfWinL$ and $\bfWnotinL$ so that rows of $\bfWinL$ positively span the $\ell$-dimensional $\calL$.
We find $\ell+1$ vectors to generate $\calL$ by selecting an orthonormal basis $\bfz_1, \dots, \bfz_{\ell}$ of $\calL$ and setting $\bfz_0 = - (\bfz_1 + \cdots + \bfz_\ell)$.
\Cref{lemma:conerank__nonpointed_necessary} states that $\ell+1$ vectors are necessary to enclose $\calK_W$.
Hence, $\tilde \bfZ = [\bfz_0; \bfz_1; \dots; \bfz_\ell]$ attains $\CR(\bfWinL)$.
On the other hand, $\tilde W = \proj_{\calL^\perp} (\WnotinL)$ generates the pointed cone $\calK_P$.
Since $\calK_W$ is an $r$-dimensional cone, $\calL$ is $\ell$-dimensional, and $\calK_P = \calL^\perp \cap \calK_W$, the cone $\calK_P$ is an $(r-\ell)$-dimensional pointed cone.
So we find $\tilde \bfV$ that attains $\CR(\tilde \bfW)$ using \Cref{alg:cr__pointed} sketched above.
Note that $\tilde \bfV$ has $r-\ell$ rows.
Together, $\bfV = [\tilde \bfZ; \tilde \bfV]$ has $r+1$ rows and attains $\CR(\bfW)$, proved in \Cref{lemma:alg__cr}.

\inbox{%
\textbf{Sketch of \textsc{GetCR} when $\calK_W$ is non-pointed (details in \Cref{alg:cr})}
\begin{algorithmic}[1]
    \State Decompose $\calK_W = \calL + \calK_P$ where $\ell = \dim \calL$
    \Comment{Use Alg.~\ref{alg:cone_minkowski_decomposition}}
    \State Partition $W$ into $\WinL = \set{ \bfw_i \in W \mid \bfw_i \in \calL }$ and $\WnotinL = \set{ \bfw_i \in W \mid \bfw_i \notin \calL }$.
    \State Find $\tilde \bfZ$ with $\ell+1$ rows that attains $\CR(\bfWinL)$
    \State Find $\tilde \bfV$ that attains $\CR(\tilde \bfW)$ where $\tilde W = \proj_{\calL^\perp} (\WnotinL)$
    \Comment{Use Alg.~\ref{alg:cr__pointed}}
    \State \Return $\bfV = [\tilde \bfZ; \tilde \bfV]$
\end{algorithmic}
}

\paragraph*{Pointedness exactly determines $\ConeRank$}
When $\calK_W$ is pointed, \Cref{alg:cr__pointed} finds $\bfV \in \bbR^{r \times n}$ such that $\calK_W \subseteq \calK_V$ where $r = \rank \bfW$.
When non-pointed, \Cref{alg:cr} finds $\bfV \in \bbR^{(r+1) \times n}$ such that $\calK_W \subseteq \calK_V$.
In fact, these algorithms find $\bfV$ that attains $\ConeRank(\bfW)$ as shown below.

\begin{lemma}
\label{lemma:cr}
    Let $\calK_W$ be an $r$-dimensional cone.
    That is, $\dim \calK_W = \rank \bfW = r$.
    If $\calK_W$ is pointed, then $\ConeRank(\bfW) = r$.
    If $\calK_W$ is non-pointed, then $\ConeRank(\bfW) = r+1$.
\end{lemma}
\begin{proof}
    When $\calK_W$ is pointed, \Cref{alg:cr__pointed} finds $\bfV \in \bbR^{k \times n}$ with $k = r$ such that $\calK_W \subseteq \calK_V$
    And \Cref{lemma:conerank__pointed_necessary} tells us that, if there exists $\bfV \in \bbR^{k \times n}$ with $\calK_W \subseteq \calK_V$, then $k \ge r$.
    So $\ConeRank(\bfW) = r$ if $\calK_W$ is pointed.
    On the other hand when $\calK_W$ is nonpointed, \Cref{alg:cr} finds $\bfV \in \bbR^{k \times n}$ with $k = r+1$ such that $\calK_W \subseteq \calK_V$.
    And \Cref{lemma:conerank__nonpointed_necessary} tells us that, if there exists $\bfV \in \bbR^{k \times n}$ with $\calK_W \subseteq \calK_V$, then $k \ge r+1$.
    So $\ConeRank(\bfW) = r+1$ if $\calK_W$ is nonpointed.
\end{proof}

So $\ConeRank$ is either $r$ or $r+1$ based on the pointedness of $\calK_W$, and does not depend on the number of rows of $\bfW$.
$\ConeSubsetRank$ and $\ConeGeneratingRank$ depend on the number of rows of $\bfW$ as they count the number of extreme rays of $\calK_W$ (when $\calK_W$ is pointed)
In general, $\ConeSubsetRank$ and $\ConeGeneratingRank$ can be much larger than $\ConeRank$.

\subsection{Comparing $\ConeRank$ and $\NonNegativeRank$}
\label{sec:comparison_conerank_nonnegativerank}
In this section, we contrast $\ConeRank$ with $\NonNegativeRank$ (often denoted as $\rank_+$), which is a widely-studied property of non-negative matrices~\cite{gillis2020nonnegative,kumar2013fast,arora2012computing,gillis2013fast,donoho2003does,vavasis2010complexity}.
While $\NonNegativeRank$ is defined for a nonnegative matrix $\bfW \in \bbR^{m \times n}$, we define $\ConeRank$ for \textit{any} matrix $\bfW$:
\begin{align*}
    \NonNegativeRank(\bfW) &\coloneqq \min_k \set{ k \mid \calK_W \subseteq \calK_V \text{ for some nonnegative matrix } \bfV \in \bbR^{k \times n} }\\
    \ConeRank(\bfW) &\coloneqq \min_k \set{ k \mid \calK_W \subseteq \calK_V \text{ for some matrix } \bfV \in \bbR^{k \times n} }.
\end{align*}

The definitions imply that $\NonNegativeRank(\bfW) \ge \ConeRank(\bfW)$ for a nonnegative matrix $\bfW$.
Importantly, $\ConeRank$ does not constrain $\bfV$ to be nonnegative, resulting in these three main distinctions:

\begin{enumerate}
    \item Geometrically, $\NonNegativeRank$ is the minimum number of vectors (arranged as rows of $\bfV$) that generates cone $\calK_V$ \textit{nested} between $\calK_W$ and the nonnegative orthant $\bbR^n_+$~\cite{gillis2020nonnegative}.
    In contrast, $\ConeRank$ does not constrain $\calK_V$ to lie inside $\bbR^n_+$---it only requires $\calK_V$ to \textit{enclose} $\calK_W$.

    \item For non-negative matrix $\bfW$ of rank $r \le 2$, it is known that $\NonNegativeRank(\bfW) = r$.
    However, for rank $r > 2$, it is NP-hard to determine whether $\NonNegativeRank(\bfW) = r$~\cite{vavasis2010complexity,gillis2020nonnegative}.
    In contrast, \Cref{lemma:cr} states that $\ConeRank(\bfW) = r$ or $r+1$ depending on the pointedness of $\calK_W$, which can be determined in $\poly(m,n)$ time with a Linear Program (\ref{eqn:LP__pointed} in \Cref{alg:is_in_cone}).
    
    \item While $\NonNegativeRank(\bfW) = \NonNegativeRank(\bfW^\top)$, we might not generally have $\ConeRank(\bfW) = \ConeRank(\bfW^\top)$.
    This is because pointedness of cone generated by rows does not imply pointedness of cone generated by columns, or vice versa.
\end{enumerate}

%% file: tikz/cone_inR2_1d_pointed.tex
\begin{tikzpicture}[thick,scale=0.6]
    \coordinate (O) at (0, 0);
    \coordinate (right) at (2, 0);
    \coordinate (left) at (-2, 0);
    \coordinate (top) at (0, 2);
    \coordinate (bottom) at (0, -2);

    \draw[<->,color=gray,dashed] (left) -- (right);
    \draw[<->,color=gray,dashed] (bottom) -- (top);
    
    \coordinate (A) at (1, 0.5);
    \coordinate (B) at (2, 1);
    
    \coordinate (Aext) at (2, 1);

    \begin{scope}[ultra thick]
        \draw[->,color=red] (O) -- (Aext);
    \end{scope}

    \begin{scope}[font=\tiny]
    \node[circle,inner sep=2pt,fill=black,label=above:{$(2,1)$}] at (A) {};
    \node[circle,inner sep=2pt,fill=black,label=above:{$(4,2)$}] at (B) {};
    \end{scope}
\end{tikzpicture}

%% file: tikz/cone_inR2_1d_nonpointed.tex
\begin{tikzpicture}[thick,scale=0.6]
    \coordinate (O) at (0, 0);
    \coordinate (right) at (2, 0);
    \coordinate (left) at (-2, 0);
    \coordinate (top) at (0, 2);
    \coordinate (bottom) at (0, -2);

    \draw[<->,color=gray,dashed] (left) -- (right);
    \draw[<->,color=gray,dashed] (bottom) -- (top);
    
    \coordinate (A) at (1, 0.5);
    \coordinate (B) at (-1, -0.5);

    \coordinate (Aext) at (2, 1);
    \coordinate (Bext) at (-2, -1);
    
    \begin{scope}[ultra thick]
        \draw[color=red] (Bext) -- (Aext);
    \end{scope}

    \begin{scope}[font=\tiny]
        \node[circle,inner sep=2pt,fill=black,label=above:{$(2,1)$}] at (A) {};
        \node[circle,inner sep=2pt,fill=black,label=below:{$(-2,-1)$}] at (B) {};
    \end{scope}
\end{tikzpicture}

%% file: tikz/cone_inR2_2d_pointed.tex
\begin{tikzpicture}[thick,scale=0.6]
    \coordinate (O) at (0, 0);
    \coordinate (right) at (2, 0);
    \coordinate (left) at (-2, 0);
    \coordinate (top) at (0, 2);
    \coordinate (bottom) at (0, -2);

    \draw[<->,color=gray,dashed] (left) -- (right);
    \draw[<->,color=gray,dashed] (bottom) -- (top);
    
    \coordinate (A) at (1, 0.5);
    \coordinate (B) at (-1, 0.5);
    \coordinate (C) at (0.5, 1);
    
    \coordinate (Aext) at (2, 1);
    \coordinate (Bext) at (-2, 1);
    \coordinate (topright) at (2, 2);
    \coordinate (topleft) at (-2, 2);

    \begin{scope}[ultra thick]
        \draw[->,color=red] (O) -- (Aext);
        \draw[->,color=red] (O) -- (Bext);
    \end{scope}

    \draw[fill=Salmon,opacity=0.5,draw=none] (O) -- (Aext) -- (topright) -- (topleft) -- (Bext);

    \begin{scope}[font=\tiny]
    \node[circle,inner sep=2pt,fill=black,label=below:{$(2,1)$}] at (A) {};
    \node[circle,inner sep=2pt,fill=black,label=below:{$(-2,1)$}] at (B) {};
    \node[circle,inner sep=2pt,fill=black,label=above:{$(1,2)$}] at (C) {};
    \end{scope}
\end{tikzpicture}

%% file: tikz/cone_inR2_2d_nonpointed.tex
\begin{tikzpicture}[thick,scale=0.6]
    \coordinate (O) at (0, 0);
    \coordinate (right) at (2, 0);
    \coordinate (left) at (-2, 0);
    \coordinate (top) at (0, 2);
    \coordinate (bottom) at (0, -2);

    \draw[<->,color=gray,dashed] (left) -- (right);
    \draw[<->,color=gray,dashed] (bottom) -- (top);
    
    \coordinate (A) at (1, 0.5);
    \coordinate (B) at (-1, -0.5);
    \coordinate (C) at (0.5, 1);
    \coordinate (D) at (-1, 1);

    \coordinate (Aext) at (2, 1);
    \coordinate (Bext) at (-2, -1);
    \coordinate (topright) at (2, 2);
    \coordinate (topleft) at (-2, 2);

    \begin{scope}[ultra thick]
        \draw[color=red] (Bext) -- (Aext);
    \end{scope}

    \draw[fill=Salmon,opacity=0.5,draw=none] (Bext) -- (Aext) -- (topright) -- (topleft);

    \begin{scope}[font=\tiny]
    \node[circle,inner sep=2pt,fill=black,label=below:{$(2,1)$}] at (A) {};
    \node[circle,inner sep=2pt,fill=black,label=below:{$(-2,-1)$}] at (B) {};
    \node[circle,inner sep=2pt,fill=black,label=above:{$(1,2)$}] at (C) {};
    \node[circle,inner sep=2pt,fill=black,label=above:{$(-2,2)$}] at (D) {};
    \end{scope}
\end{tikzpicture}

%% file: tikz/cone_inR2_2d_full.tex
\begin{tikzpicture}[thick,scale=0.6]
    \coordinate (O) at (0, 0);
    \coordinate (right) at (2, 0);
    \coordinate (left) at (-2, 0);
    \coordinate (top) at (0, 2);
    \coordinate (bottom) at (0, -2);

    \draw[<->,color=gray,dashed] (left) -- (right);
    \draw[<->,color=gray,dashed] (bottom) -- (top);
    
    \coordinate (A) at (1, 1);
    \coordinate (B) at (-1, 1);
    \coordinate (C) at (-1, -1);
    \coordinate (D) at (1, -1);
    
    \coordinate (topright) at (2, 2);
    \coordinate (topleft) at (-2, 2);
    \coordinate (bottomright) at (2, -2);
    \coordinate (bottomleft) at (-2, -2);

    % \begin{scope}[ultra thick]
    %     \draw[->] (O) -- (Aext);
    %     \draw[->] (O) -- (Bext);
    % \end{scope}

    \draw[fill=Salmon,opacity=0.5,draw=none] (topright) -- (topleft) -- (bottomleft) -- (bottomright);

    \begin{scope}[font=\tiny]
    \node[circle,inner sep=2pt,fill=black,label=above:{$\bfw_2 (1,1)$}] at (A) {};
    \node[circle,inner sep=2pt,fill=black,label=above:{$\bfw_1 (-1,1)$}] at (B) {};
    \node[circle,inner sep=2pt,fill=black,label=below:{$\bfw_4 (-1,-1)$}] at (C) {};
    \node[circle,inner sep=2pt,fill=black,label=below:{$\bfw_3 (1,-1)$}] at (D) {};
    \end{scope}
\end{tikzpicture}

%% file: tikz/decompose_cone_inR3_plane.tex
% \tdplotsetmaincoords{0}{30}
% \begin{tikzpicture}[tdplot_main_coords,thick,scale=0.75]
\begin{tikzpicture}[thick,scale=1.5]
    \coordinate (O) at (0, 0, 0);

    \draw[->,color=gray,dashed] (-1.25,0,0) -- (1.25,0,0);
    \draw[->,color=gray,dashed] (0,-1.25,0) -- (0,1.25,0);
    \draw[->,color=gray,dashed] (0,0,-1.25) -- (0,0,1.25);
    
    \draw[fill=Cerulean,opacity=0.4,draw=none] (0,-1,0) -- (1,-1,0) -- (1,1,0) -- (0,1,0);

    % boundary
    \begin{scope}[draw=Cerulean]
        \draw[dotted] (0,-1,0) -- (1,-1,0) -- (1,1,0) -- (0,1,0);
    \end{scope}
    
    \draw[<->,color=mygreen,ultra thick] (0,-1,0) -- (0,1,0);
    \draw[->,color=red,ultra thick] (O) -- (1,0,0);
    
    \node[text=mygreen,label={[text=mygreen]left:$\lineality(\calK_W)$}] at (0,0.5,0) {};
    \node[text=red,label={[text=myred]below:$\calK_P$}] at (0.5,0,0) {};
\end{tikzpicture}

%% file: tikz/decompose_cone_inR3_wedge.tex
% \tdplotsetmaincoords{0}{30}
% \begin{tikzpicture}[tdplot_main_coords,thick,scale=0.75]
\begin{tikzpicture}[thick,scale=1.5]
    \coordinate (O) at (0, 0, 0);

    \draw[->,color=gray,dashed] (-1.25,0,0) -- (1.25,0,0);
    \draw[->,color=gray,dashed] (0,-1.25,0) -- (0,1.25,0);
    \draw[->,color=gray,dashed] (0,0,-1.25) -- (0,0,1.25);

    % back face
    \draw[fill=Cerulean,opacity=0.2,draw=none] (0,1,0) -- (1,1,-0.4) -- (1,-1,-0.4) -- (0,-1,0);
    % \draw[dashed,color=blue] (0,-1,-1) -- (1,-1,-1);
    % bottom face
    \draw[fill=Cerulean,opacity=0.4,draw=none] (0,-1,0) -- (1,-1,-0.4) -- (1,-1,0.4);
    % top face
    \draw[fill=Cerulean,opacity=0.4,draw=none] (0,1,0) -- (1,1,-0.4) -- (1,1,0.4);
    % front face
    \draw[fill=Cerulean,opacity=0.3,draw=none] (0,1,0) -- (1,1,0.4) -- (1,-1,0.4) -- (0,-1,0);

    % boundary
    \begin{scope}[draw=Cerulean]
        \draw[dotted] (1,1,0.4) -- (1,1,-0.4) -- (1,-1,-0.4) -- (1,-1,0.4) -- cycle;
        \draw (0,1,0) -- (1,1,0.4);
        \draw (0,1,0) -- (1,1,-0.4);
        \draw (0,-1,0) -- (1,-1,0.4);
        \draw[opacity=0.2] (0,-1,0) -- (1,-1,-0.4);
    \end{scope}
    \begin{scope}[ultra thick,draw=red]
        \draw[dotted] (1,0,0.4) -- (1,0,-0.4);
        \draw (O) -- (1,0,0.4);
        \draw[opacity=0.2] (O) -- (1,0,-0.4);
    \end{scope}
    
    \draw[<->,color=mygreen,ultra thick] (0,-1,0) -- (0,1,0);
    \draw[fill=Salmon,opacity=0.5,draw=none] (O) -- (1,0,-0.4) -- (1,0,0.4);

    \node[text=mygreen,label={[text=mygreen]left:$\lineality(\calK_W)$}] at (0,0.5,0) {};
    \node[text=red,label={[text=myred]below:$\calK_P$}] at (0.5,-0.1,0) {};
\end{tikzpicture}

%% file: tikz/decompose_cone_inR3_halfspace.tex
% \tdplotsetmaincoords{0}{30}
% \begin{tikzpicture}[tdplot_main_coords,thick,scale=0.75]
\begin{tikzpicture}[thick,scale=1.3]
    \coordinate (O) at (0, 0, 0);

    \draw[->,color=gray,dashed] (-1.25,0,0) -- (1.25,0,0);
    \draw[->,color=gray,dashed] (0,-1.25,0) -- (0,1.25,0);
    \draw[->,color=gray,dashed] (0,0,-1.25) -- (0,0,1.25);

    % back face
    \draw[fill=Cerulean,opacity=0.2,draw=none] (0,-1,-1) -- (1,-1,-1) -- (1,1,-1) -- (0,1,-1);
    % bottom face
    \draw[fill=Cerulean,opacity=0.4,draw=none] (0,-1,-1) -- (1,-1,-1) -- (1,-1,1) -- (0,-1,1);
    % top face
    \draw[fill=Cerulean,opacity=0.4,draw=none] (0,1,-1) -- (0,1,1) -- (1,1,1) -- (1,1,-1);
    % front face
    \draw[fill=Cerulean,opacity=0.3,draw=none] (0,1,1) -- (1,1,1) -- (1,-1,1) -- (0,-1,1);
    % side face
    \draw[fill=Cerulean,opacity=0.2,draw=none] (1,1,-1) -- (1,1,1) -- (1,-1,1) -- (1,-1,-1);

    % boundary
    \begin{scope}[draw=Cerulean,dotted]
        \draw (0,-1,-1) -- (1,-1,-1) -- (1,1,-1) -- (0,1,-1);
        % bottom face
        \draw (0,-1,-1) -- (1,-1,-1) -- (1,-1,1) -- (0,-1,1);
        % top face
        \draw (0,1,1) -- (1,1,1) -- (1,1,-1) -- (0,1,-1);
        % front face
        \draw (0,1,1) -- (1,1,1) -- (1,-1,1) -- (0,-1,1);
        % side face
        \draw (1,1,-1) -- (1,1,1) -- (1,-1,1) -- (1,-1,-1);
    \end{scope}
    \begin{scope}[ultra thick,dotted,draw=mygreen]
        \draw (0,1,-1) -- (0,1,1) -- (0,-1,1) -- (0,-1,-1) -- cycle;
    \end{scope}

    \draw[fill=mygreen,opacity=0.5,draw=none] (0,1,-1) -- (0,1,1) -- (0,-1,1) -- (0,-1,-1);
    \draw[->,color=red,ultra thick] (O) -- (1,0,0);

    \node[text=mygreen,label={[text=mygreen]left:$\lineality(\calK_W)$}] at (-0.3,0.5,0) {};
    \node[text=red,label={[text=myred]below:$\calK_P$}] at (1,0,0) {};
\end{tikzpicture}

%% file: conclusion.tex
\section{Conclusion}
\label{sec:conclusion}

We propose a framework to design succinct scores to summarize performance metrics $\calF$, and give polynomial-time algorithms that design scores that are provably minimal under mild assumptions on $\calF$.
Two future directions are to design scores: (1) when metrics takes discrete high-dimensional values, (2) using incomplete, noisy high data from historical samples of  metric values, and (3) when metrics have a non-linear structure.
On a technical note, it remains to identify structural properties of $\calF$ and corresponding minimal designs for the optimality objective.
Designing minimal scores for simultaneously satisfying both objectives under linear restriction is also an open direction.

%% file: acknowledgements.tex
\section*{Acknowledgements}

Anmol Kabra thanks Naren Sarayu Manoj and Max Ovsiankin for pointers on convex analysis and geometry.

%% file: appendix.tex
\appendix

\input{app_omitted_proofs}

\input{app_algorithms}

\input{app_preliminaries}

\input{app_key_lemmas}

%% file: app_omitted_proofs.tex
\clearpage
\section{Omitted Proofs}
\label{app:omitted_proofs}

\subsection{Minimal design problem for improvement objective}

\begin{theorem}[\Cref{thm:1__f_linear__suff}]
    \label{thm_w_proof:1__f_linear__suff}
    Let columns of $\bfZ$ be an orthonormal basis of linear subspace $\calL$ associated with $\affine(\calF)$ and let $r = \dim \affine(\calF)$.
    For each design restriction, there exists $S : \calF \to \bbR^k$, designed using \Cref{alg:1__design_strategy}, that satisfies the improvement objective with the following dimensionalities.
    
    \setlength{\tabcolsep}{3pt}
    \noindent
    \begin{center}
        \begin{tabular}{l|l@{}l} % remove space between last two columns
             & \multicolumn{2}{c}{Dimensionality $k \ge$} \\
            \midrule
            \restrictionCS & $\ConeSubsetRank(\bfZ)$ & $\coloneqq \min_q \set{ q \mid \calK_Z = \calK_V \text{ for some } \bfV \in\bbR^{q \times r} \text{ s.t. } \bfV \subseteq \bfZ }$ \\
            \restrictionLM & $\ConeGeneratingRank(\bfZ)$ & $\coloneqq \min_q \set{ q \mid \calK_Z = \calK_V \text{ for some } \bfV \in \bbR^{q \times r} }$ \\
            \restrictionL & $\ConeRank(\bfZ)$ & $\coloneqq \min_q \set{ q \mid \calK_Z \subseteq \calK_V \text{ for some } \bfV \in \bbR^{q \times r} }$ \\
        \end{tabular}
    \end{center}
\end{theorem}
\begin{proof}
    We give a proof for the \restrictionCS restriction; proofs for the other two restrictions are similar.
    We show that, if $k \ge \CSR(\bfZ)$, then there exists $S(\bff) = \bfA \bff$ satisfying improvement and \restrictionCS.

    Let columns of $\bfZ \in \bbR^{d \times r}$ be an orthonormal basis of $r$-dimensional linear subspace $\calL$ associated with $\affine(\calF)$.
    The definition of $\CSR$ states that $k \ge \CSR(\bfZ)$ when there exists $\bfV \in \bbR^{k \times r}$ such that (i) $\bfV \subseteq \bfZ$ and (ii) $\calK_Z = \calK_V$.
    Property (i) means that $\bfV = \bfA \bfZ$ for some $\bfA \in \bbR^{k \times d}$ with 1-hot rows, and so $S(\bff) = \bfA \bff$ satisfies the \restrictionCS restriction.
    Property (ii) implies that $\calK_Z \subseteq \calK_V$, and so $S$ satisfies improvement:
    \begin{align}
        \label{eqn:1__f_linear__suff__main_arg}
        \calK_Z \subseteq \calK_V
        &\longiff{\text{Lem.~\ref{lemma:1__f_linear__cone_subsets_in_subspace}}}
        \calL \cap \calK_A^\ast \subseteq \calK_I^\ast
        \longimplies{\text{Def.~\ref{defn:affine_hull_f}}}
        \begin{matrix}
            \text{for all } \bff \in \calF, \quad \\
            \calF_{\bff} \cap \calK_A^\ast \subseteq \calK_I^\ast
        \end{matrix}
        \longiff{\text{Eq.~\ref{eqn:1__f_linear__restate_cones}}}
        \text{Improvement}.
    \end{align}
    
    The proof of \Cref{lemma:1__f_linear__cone_subsets_in_subspace} uses $\bfV = \bfA \bfZ$, and the projection of rows of $\bfA$ and $\bfI_d$ in subspace $\calL$ using orthonormal basis $\bfZ$.
    \qedhere
\end{proof}

\begin{example}[Competing metric improvement directions$\implies$dimensionality for \restrictionCS $>$ \restrictionLM]
\label{example:matrix_rank__5d__nonpointed}
    When cone $\calK_Z$ generated by rows of $\bfZ$ is non-pointed, we have $\CSR(\bfZ) > \CGR(\bfZ)$, implying that the score design dimensionality is higher under \restrictionCS restriction than under \restrictionLM.
    The cone $\calK_Z$ can be non-pointed in the presence of competing metric improvement directions, i.e., when improving on one metric degrades another.
    A non-pointed $\calK_Z$ results in a gap between $\CSR(\bfZ)$ and $\CGR(\bfZ)$.

    Consider 8 metrics lying in a 5-dimensional subspace, which has the following orthonormal basis (arranged as columns of $\bfZ$):
    \begin{align*}
        \bfZ &= \frac{1}{2} \cdot \begin{bmatrix}
            1 & 1 & 0 & 0 & 0\\
            -1 & 1 & 0 & 0 & 0\\
            -1 & -1 & 0 & 0 & 0\\
            1 & -1 & 0 & 0 & 0\\
            0 & 0 & 1 & 1 & 1\\
            0 & 0 & 1 & -1 & 1\\
            0 & 0 & 1 & -1 & -1\\
            0 & 0 & 1 & 1 & -1\\
        \end{bmatrix} \in \bbR^{8 \times 5}.
    \end{align*}

    The rows generate a 5-dimensional cone $\calK_Z$ with two orthogonal parts: (i) a 2-dimensional linear subspace due to the first 4 metrics, and (ii) a 3-dimensional ``square'' pointed cone due to the last 4 metrics, as visualized in \Cref{fig:example_matrix_rank__5d_nonpointed}.
    Since $\calK_Z$ contains a 2-dimensional linear subspace within, it is a non-pointed cone.

    A matrix $\bfV$ that attains $\CSR(\bfZ)$ must have rows of $\bfV$ chosen from rows of $\bfZ$ and $\calK_Z = \calK_V$.
    Excluding any row of $\bfZ$ shrinks the generated cone---excluding any row of the first 4 generates a halfspace rather than the 2-dimensional subspace, and excluding any row of the last 4 does not generate the ``square'' pointed cone.
    So $\CSR(\bfZ) = 8$.
    On the other hand, a matrix $\bfV$ that attains $\CGR(\bfZ)$ need not have rows of $\bfV$ chosen from rows of $\bfZ$; $\bfV$ must only satisfy $\calK_Z = \calK_V$.
    We need all last 4 rows to generate the ``square'' cone, but there exists 3 points (the blue and two bottom black points) whose nonnegative combinations generate the 2-dimensional linear subspace.
    So $\CGR(\bfZ) = 7$.
    \begin{figure}[!h]
        \centering
        \input{tikz/cone_inR5_5d_nonpointed}
        \caption{
        A 5-dimensional non-pointed cone $\calK_Z$ with two orthogonal components: a 2-dimensional linear subspace, and a 3-dimensional ``square'' pointed cone.
        }
        \label{fig:example_matrix_rank__5d_nonpointed}
    \end{figure}
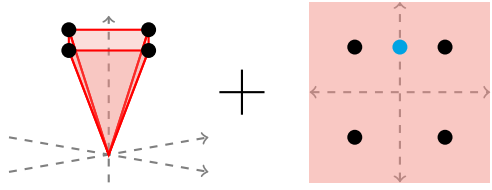
\end{example}

\begin{theorem}[\Cref{thm:1__f_linear__nonempty_relint_f_with_origin__necsuff}]
\label{thm_w_proof:1__f_linear__nonempty_relint_f_with_origin__necsuff}
    Assume metrics $\calF \subseteq \bbR^d$ have non-empty relative interior with respect to $\affine(\calF)$.
    Then the listed dimensionalities $k$ in \Cref{thm:1__f_linear__suff} are necessary.
\end{theorem}
\begin{proof}
    We give a proof for the \restrictionCS restriction; proof for the other two restrictions are similar.
    We show that, when $\calF$ has non-empty relative interior, we get:
    \begin{align}
        \label{eqn:1__f_linear__necsuff__converse_claim}
        \text{for all } \bff \in \calF,\quad \calF_{\bff} \cap \calK_A^\ast \subseteq \calK_I^\ast
        &\implies \calL \cap \calK_A^\ast \subseteq \calK_I^\ast.
    \end{align}
    By adding this implication to \Cref{eqn:1__f_linear__suff__main_arg}, we prove that, when $\calF$ has non-empty relative interior, a score function $S$ satisfies the improvement objective and \restrictionCS restriction \textit{if and only if} $k \geq \CSR(\bfZ)$.

    We now prove the implication in \Cref{eqn:1__f_linear__necsuff__converse_claim}.
    Let $\bfx \in \calL \cap \calK_A^\ast$.
    Since $\calF$ has non-empty relative interior, there exists $\bff^\ast$ in the relative interior.
    \Cref{lemma:subspace_scaled_to_relint} states that, as $\bfx \in \calL$, there exists $a > 0$ such that $a \bfx \in \calF_{\bff^\ast}$.
    Since $\bfx$ is in cone $\calK_A^\ast$ as well, we have $a \bfx \in \calK_A^\ast$.
    Hence, $a \bfx \in \calF_{\bff^\ast} \cap \calK_A^\ast$.
    According to the premise of  \Cref{eqn:1__f_linear__necsuff__converse_claim}, we know that $\calF_{\bff^\ast} \cap \calK_A^\ast \subseteq \calK_I^\ast$, and so $a \bfx \in \calK_I^\ast$.
    As $a > 0$, we get $\bfx \in \calK_I^\ast$, completing the proof.
\end{proof}

\begin{lemma}
\label{lemma_w_proof:matrix_rank__smallest_with_affine_invariant_to_rotation}
    Given affine subspace $\calH$ containing $\calF$, the matrix ranks are invariant to the choice of orthonormal basis of $\calL_{\calH}$.
    Moreover, among all affine subspaces containing $\calF$, the matrix ranks are smallest for $\calH = \affine(\calF)$. 
\end{lemma}
\begin{proof}
    We give a proof for $\CSR$, proofs for the other two matrix ranks are similar.

    \begin{enumerate}
        \item We first give a geometric interpretation for invariance to choice of orthonormal basis of $\calL_{\calH}$.
        Then we give an algebraic proof.\\
        
        \boldheading{Geometric interpretation.}
        For any matrix $\bfW$, note that $\CSR(\bfW)$ is the minimum cardinality of a subset $V$ of $W$ (set of rows of $\bfW$), such that cone $\calK_V$ encloses $\calK_W$.
        By rotating rows of $\bfW$ without altering the column span of $\bfW$, although the row vectors $W$ change, the \textit{relative position of them with respect to each other is the same}.
        So the cone generated by the rotated vectors is just a rotation of cone $\calK_W$.
        As a result, the minimum cardinality of a subset of rotated vectors (to enclose the rotated cone) is unchanged, and so $\CSR(\bfW)$ is unchanged.\\

        \boldheading{Algebraic argument.}
        Let columns of $\bfZ_1$ and $\bfZ_2$ be two sets of orthonormal basis of $r_{\calH}$-dimensional $\calL_{\calH}$.
        We will show that $\CSR(\bfZ_1) = \CSR(\bfZ_2)$.
        The two orthonormal bases have the same column span, and are rotations/reflections of each other.
        So there exists orthogonal matrix $\bfQ \in \bbR^{r_{\calH} \times r_{\calH}}$ such that $\bfZ_1 = \bfZ_2 \bfQ$ and $\bfZ_1 \bfQ^\top = \bfZ_2$.

        We prove that $\CSR(\bfZ_1) \le \CSR(\bfZ_2)$.
        Let $\CSR(\bfZ_2) = k^\ast$.
        Then there exists $\bfV_2 \in \bbR^{k^\ast \times r_{\calH}}$ such that $\bfV_2 \subseteq \bfZ_2$ and $\calK_{Z_2} \subseteq \calK_{V_2}$.
        These two properties mean that $\bfV_2 = \bfA \bfZ_2$ for some $\bfA$ with 1-hot rows, and $\bfZ_2 = \bfB \bfV_2$ for some nonnegative $\bfB$.
        Multiplying with $\bfQ$ on the right, we get $\bfV_2 \bfQ = \bfA \bfZ_2 \bfQ$ and $\bfZ_2 \bfQ = \bfB \bfV_2 \bfQ$.
        Therefore, $\bfV_1 = \bfV_2 \bfQ \in \bbR^{k^\ast \times r_{\calH}}$ has the properties $\bfV_1 \subseteq \bfZ_1$ and $\calK_{Z_1} \subseteq \calK_{V_1}$.
        This proves that $\CSR(\bfZ_1) \le \CSR(\bfZ_2)$.
        With a symmetric argument, we also get $\CSR(\bfZ_1) \ge \CSR(\bfZ_2)$.

        \item Let $\calH_1$ and $\calH_2$ be two non-empty affine subspaces containing $\calF$ such that $\calH_1 \subseteq \calH_2$.
        Let $\calL_1$ and $\calL_2$ be linear subspaces corresponding to $\calH_1$ and $\calH_2$ respectively.
        Since $\calH_1 \subseteq \calH_2$ and for any $\bff \in \calH_1$ we can write $\calL_1 = \calH_1 - \bff$ and $\calL_2 = \calH_2 - \bff$, we find that $\calL_1 \subseteq \calL_2$.
        According to statement (1), $\CSR$ is invariant to the choice of orthonormal basis of linear subspace.
        Hence, pick columns of $\bfZ_1$ and $\bfZ_2$ as orthonormal basis of $\calL_1$ and $\calL_2$ respectively, such that columns of $\bfZ_2$ are a superset of columns of $\bfZ_1$.
        In the definition of $\CSR$, adding vectors to $\bfZ_1$ only increases the number of constraints to satisfy, and so $\CSR$ can only grow.
        Hence, $\CSR(\bfZ_1) \le \CSR(\bfZ_2)$.
        
        Since $\affine(\calF)$ is the unique intersection of all affine subspaces containing $\calF$, we have $\affine(\calF) \subseteq \calH$ for every affine subspace $\calH$ containing $\calF$.
        Thus, $\CSR(\bfZ) \le \CSR(\bfZ_{\calH})$, where columns of $\bfZ$ and $\bfZ_{\calH}$ are orthonormal basis of linear subspaces corresponding to $\affine(\calF)$ and $\calH$ respectively.
        \qedhere
    \end{enumerate}
\end{proof}

\begin{proposition}
\label{prop_w_proof:1__f_linear__empty_relint__counterexample}
    For each design restriction, there exists $\calF \subseteq \bbR^d$ with $\dim \affine(\calF) = d$ and empty relative interior such that there exists function $S : \calF \to \bbR$ that satisfies improvement objective.
\end{proposition}
\begin{proof}
    We first give an example of $\calF \subseteq \bbR^2$, and show that there exists $S : \calF \to \bbR$ that satisfies improvement and the \restrictionCS restriction.
    So $S$ will also satisfy the other two design restrictions.

    Consider $\calF = \set{ (0, 0), (1, 1), (2, 3) } \subseteq \bbR^2$ and let $\bfA = [1, 0] \in \bbR^{1 \times 2}$.
    We now argue that $S(\bff) = \bfA \bff$ satisfies the improvement objective.
    For metric pairs 
    \begin{align*}
        (\bff', \bff) &\in \set{ ((1, 1), (0, 0)), ((2, 3), (1, 1)), ((2, 3), (0, 0)) }
    \end{align*}
    we have $\bfA \bff' \ge \bfA \bff$ and $\bff' \ge \bff$.
    Hence, improvement objective holds for these pairs.
    Whereas for metric pairs
    \begin{align*}
        (\bff', \bff) &\in \set{ ((0, 0), (1, 1)), ((1, 1), (2, 3)), ((0, 0), (2, 3)) }       
    \end{align*}
    the left-hand side of the implication ($\bfA \bff' \ge \bfA \bff$) is not true.
    And so improvement objective holds for these pairs \textit{vacuously}.
    Thus for all $\bff, \bff' \in \calF$ if $\bfA \bff' \ge \bfA \bff$ then $\bff' \ge \bff$.

    We now give a counterexample of $d+1$ points in $\calF \subseteq \bbR^d$.
    Let $\round{\bff}{0} = \bfzero_d$ and $\round{\bff}{1} = \bfone_d$.
    For $i = 2, \dots, d$, construct $\round{\bff}{i}_j = \inparens{ \round{\bff}{i-1}_j }^2 + j$ for each coordinate $j \in [d]$.
    For example, the construction in $\bbR^4$ is:
    \begin{align*}
        \calF &= \inbraces{
            \begin{pmatrix}
                0\\0\\0\\0
            \end{pmatrix},
            \begin{pmatrix}
                1\\1\\1\\1
            \end{pmatrix},
            \begin{pmatrix}
                2\\3\\4\\5
            \end{pmatrix},
            \begin{pmatrix}
                5\\11\\19\\29
            \end{pmatrix},
            \begin{pmatrix}
                26\\123\\364\\845
            \end{pmatrix}
        }
    \end{align*}
    
    Points $\round{\bff}{1}, \dots, \round{\bff}{d}$ are linearly independent, and so $\dim \spanv(\calF) = d$.
    Let $\bfA = [1, 0, \dots, 0] \in \bbR^{1 \times d}$.
    Following a similar argument as the $d=2$ case, we find that $S(\bff) = \bfA \bff$ satisfies the improvement objective (with dimensionality $k = 1$).
\end{proof}

\subsection{Minimal design problem for optimality objective}

\begin{theorem}[{\Cref{thm:2__f_linear__suff}}]
    \label{thm_w_proof:2__f_linear__suff}
    For each design restriction, there exists $S : \calF \to \bbR^k$, designed using \Cref{alg:2__design_strategy}, that satisfies the optimality objective with the following dimensionalities.
    
    \noindent
    \begin{center}
        \begin{tabular}{l|c}
             & Dimensionality $k \ge$ \\
            \midrule
            \restrictionCS & $\dim \affine(\calF)$ \\
            \restrictionLM & $1$ \\
            \restrictionL & $1$ \\
        \end{tabular}
    \end{center}
\end{theorem}
\begin{proof}
    For the last two design restrictions, the minimal design is straightforward.
    Using any vector $\bfa > \bfzero$ of positive entries, design $S : \bff \mapsto \bfa \cdot \bff$~\cite{zadeh1963optimality}.
    Clearly, $S$ is linear in $\bff$. 
    To see that $S$ is also monotone, fix $\bff, \bff' \in \calF$ such that $\bff \ge \bff'$.
    Taking inner product with positive vector $\bfa$, we get $\bfa \cdot \bff \ge \bfa \cdot \bff'$.
    To see that optimality objective is satisfied, fix $\bff^\ast \in \ParetoOpt(S)$.
    Since $S$ is 1-dimensional, by definition of $\ParetoOpt(S)$, we have $\bfa \cdot \bff^\ast \ge \bfa \cdot \bff$ for all $\bff \in \calF$.
    Since $\bfa$ only has positive elements, for any $\bff \in \calF$ either $\bff^\ast = \bff$ or there exists $j \in [d]$ such that $\bff^\ast_j > \bff_j$.
    Therefore, $\bff^\ast \in \ParetoOpt(\calF)$.
    \\
    
    \noindent
    \boldheading{\restrictionCS restriction.}
    We now give a design for the \restrictionCS restriction.
    We first simplify the optimality objective---$\ParetoOpt(S) \subseteq \ParetoOpt(\calF)$ using movement directions $\calF_{\bff} = \set{ \bfg = \bff' - \bff \in \bbR^d \mid \text{ for all } \bff' \in \calF }$, definitions of dual cones $\calK_A^\ast$ and $\calK_I^\ast$, and $\ker \bfA = \set{ \bfx \in \bbR^d \mid \bfA \bfx = \bfzero }$.
    We rewrite $\ParetoOpt(S)$ as follows:
    \begin{align*}
        \ParetoOpt(S) &= \set{\bff \in \calF \mid \text{for all } \bfg \in \calF_{\bff}, \text{ either } \bfA \bfg \not\ge \bfzero \text{ or } \bfA \bfg = \bfzero}\\
        &= \set{\bff \in \calF \mid \calF_{\bff} \subseteq  \stcomplement{(\calK_A^\ast)} \cup \ker \bfA }.
    \end{align*}

    Similarly, $\ParetoOpt(\calF) = \set{\bff \in \calF \mid \calF_{\bff} \subseteq \stcomplement{(\calK_I^\ast)} \cup \ker \bfI }$.
    Thus we get:
    \begin{align}
    \tag{Eq.~\ref{eqn:2__f_linear__restate_cones}}
        \text{Optimality} &\iff \set{ \bff \in \calF \mid \calF_{\bff} \subseteq \stcomplement{(\calK_A^\ast)} \cup \ker \bfA } \subseteq \set{ \bff \in \calF \mid \calF_{\bff} \subseteq \stcomplement{(\calK_I^\ast)} \cup \ker \bfI }.
    \end{align}

    We now identify an isomorphism between movement directions $\calF_{\bff}$ in the ambient space and the coefficient space.
    Let columns of $\bfZ \in \bbR^{d \times r}$ be an orthonormal basis of $r$-dimensional linear subspace $\calL$ associated with $\affine(\calF)$.
    Fix any $\bff \in \calF$. 
    Denote with $\calC_{\bff} \in \bbR^r$ the set of coefficients of $\calF_{\bff}$ w.r.t. orthonormal basis $\bfZ$, i.e., $\calC_{\bff} = \bfZ^\top \inparens{ \calF_{\bff} }$.
    This introduces an isomorphism between the sets $\calF_{\bff}$ and $\calC_{\bff}$, i.e., for every $\bfg \in \calF_{\bff}$ these exists unique $\bfd \in \calC_{\bff}$ such that $\bfg = \bfZ \bfd$.
    With $\bfV = \bfA \bfZ$, we have four equivalences:
    \begin{align*}
        \bfA \bfg \geq \bfzero \iff \bfV \bfd \geq \bfzero \quad \qquad &\text{and} \quad \qquad \bfA \bfg = \bfzero \iff \bfV \bfd = \bfzero\,,\\
        \bfg \ge \bfzero \iff \bfZ \bfd \ge \bfzero \quad\qquad &\text{and} \quad \qquad \bfg = \bfzero \iff \bfZ \bfd = \bfzero.
    \end{align*}

    \Cref{lemma:2__f_linear__cone_decompose_in_subspace} uses these equivalences to state that for any $\bff \in \calF$, we have
    \begin{align}
         \calF_{\bff} \subseteq \stcomplement{(\calK_A^\ast)} \cup \ker \bfA &\iff \calC_{\bff} \subseteq \stcomplement{(\calK_V^\ast)} \cup \ker \bfV
         \\
         \calF_{\bff} \subseteq \stcomplement{(\calK_I^\ast)} \cup \ker \bfI &\iff \calC_{\bff} \subseteq \stcomplement{(\calK_Z^\ast)} \cup \ker \bfZ .
    \end{align}
    
    We further simplify the optimality objective (\Cref{eqn:2__f_linear__restate_cones}):
    \begin{align}
        \text{Optimality} &\iff
        \set{ \bff \in \calF \mid \calF_{\bff} \subseteq \stcomplement{(\calK_A^\ast)} \cup \ker \bfA } \subseteq \set{ \bff \in \calF \mid \calF_{\bff} \subseteq \stcomplement{(\calK_I^\ast)} \cup \ker \bfI }\\
        \label{eqn:2__f_linear__suff__main_arg}
        &\iff
        \set{ \bff \in \calF \mid \calC_{\bff} \subseteq \stcomplement{(\calK_V^\ast)} \cup \ker \bfV } \subseteq \set{ \bff \in \calF \mid \calC_{\bff} \subseteq \stcomplement{(\calK_Z^\ast)} \cup \ker \bfZ }
    \end{align}
    where \Cref{eqn:2__f_linear__suff__main_arg} follows from \Cref{lemma:2__f_linear__cone_decompose_in_subspace}.

    Now, we choose $r$ linear independent rows of $\bfZ$ to create $\bfV \in \bbR^{r \times r}$.
    Since $\bfZ$ has orthonormal columns, we have $\ker \bfV = \ker \bfZ = \set{\bfzero}$.
    Moreover, we have $\bfV \subseteq \bfZ$, implying $\calK_V \subseteq \calK_Z$ and $\calK_Z^\ast \subseteq \calK_V^\ast$ (\Cref{lemma:polyhedral_cones__duals__inclusion}).
    This shows that $\calK_Z^\ast \cup \stcomplement{(\ker \bfZ)} \subseteq \calK_V^\ast \cup \stcomplement{(\ker \bfV)}$.
    As a result, $\stcomplement{(\calK_V^\ast)} \cup \ker \bfV \subseteq \stcomplement{(\calK_Z^\ast)} \cup \ker \bfZ$. Hence,  for any $\bff\in\calF$ for which $\calC_{\bff} \subseteq \stcomplement{(\calK_V^\ast)} \cup \ker \bfV$, we also have $\calC_{\bff} \subseteq \stcomplement{(\calK_Z^\ast)} \cup \ker \bfZ$. 
    This shows that \Cref{eqn:2__f_linear__suff__main_arg} holds with the proposed choice of $\bfV$.
    As $\bfV = \bfA \bfZ$ for $\bfA$ with 1-hot rows, this design satisfies optimality and \restrictionCS restriction.
\end{proof}

\begin{proposition}[\Cref{prop:2__f_linear__necessary_subset_examples}]
\label{prop_w_proof:2__f_linear__necessary_subset_examples}
    Consider designing $S : \calF \to \bbR^k$ to satisfy optimality objective.
    \begin{enumerate}
        \item For $\calF = \set{ \bff \in \bbR^d \mid \norm{\bff}_1 \le 1 }$, $k \ge 1$ is necessary and sufficient for all design restrictions.
        
        \item For $\calF = \set{ \bff \in \bbR^d \mid \norm{\bff}_2 \le 1 }$, $k \ge 1$ is necessary and sufficient for all design restrictions.
        
        \item For $\calF = \set{ \bff \in \bbR^d \mid \norm{\bff}_\infty \le 1 }$, $k \ge d$ is necessary and sufficient for \restrictionCS.
        Moreover, $k \ge 1$ is necessary and sufficient for the \restrictionLM and \restrictionL restrictions.
    \end{enumerate}
\end{proposition}
\begin{proof}
    \Cref{thm:2__f_linear__suff} states $k \ge 1$ is sufficient for \restrictionLM and \restrictionL restrictions for any $\calF$; trivially, $k \ge 1$ is necessary.
    So, we prove the claims for the \restrictionCS restriction.
    For the stated sets $\calF$, we determine $\ParetoOpt(\calF)$ and discuss choice of $S$ to satisfy $\ParetoOpt(S) \subseteq \ParetoOpt(\calF)$.

    We denote the $d$ coordinates of metric value $\bff \in \calF$ with $\bff_1, \dots, \bff_d$.
    Let $\bfe_j$ be the $j^{th}$ canonical basis vector of $\bbR^d$.
    We denote the unit $\ell_p$-norm ball with $\bbB_p^d = \set{ \bff \in \bbR^d \mid \norm{\bff}_p \le 1 }$.
    \begin{enumerate}
        \item Let $\calF = \bbB_1^d$, the unit $\ell_1$-norm ball centered at the origin.
        Note that the $j^{th}$ coordinate of metric value $\bff_j$ is maximized when $\bff = \bfe_j$.
        So vectors $\bfe_1, \dots, \bfe_d$ are pareto-optimal w.r.t. $\calF$.
        In fact, all vectors on the surface of $\bbB_1^d$ in the nonnegative orthant are pareto-optimal w.r.t. $\calF$.
        That is, $\ParetoOpt(\calF) = \set{ \bff \in \bbR^d_+ \mid \bfone_d \cdot \bff = 1 }$.

        We choose any coordinate $j \in [d]$ and design 1-dimensional $S(\bff) = \bff_j$.
        Since $\calF$ is the unit $\ell_1$-norm ball, $\ParetoOpt(S) = \set{ \bfe_j }$, which a subset of $\ParetoOpt(\calF)$ as $\bfone_d \cdot \bfe_j = 1$.
        Hence, this design with dimensionality $k = 1$ satisfies the optimality objective under \restrictionCS restriction.

        Trivially, $k \ge 1$ is necessary as well.
        
        \item Let $\calF = \bbB_2^d$, the unit $L_2$-ball centered at the origin.
        Note that the $j^{th}$ coordinate of metric value $\bff_j$ is maximized when $\bff = \bfe_j$.
        So vectors $\bfe_1, \dots, \bfe_d$ are pareto-optimal w.r.t. $\calF$.
        In fact, all vectors on the unit shell in the nonnegative orthant are pareto-optimal w.r.t. $\calF$.
        That is, $\ParetoOpt(\calF) = \bbS_2^{d-1} \cap \bbR_d^+ = \bbS_2^{d-1} \cap \calK_I$ where $\bfI$ is the identity matrix.
        
        We can similarly determine pareto-optimal points w.r.t. $S(\bff) = \bfA \bff$.
        Let $\bfA$ have $k$ rows $\bfA = [\bfa_1; \dots; \bfa_k] \in \bbR^{k \times d}$.
        The $i^{th}$ coordinate of $S$ is maximized when $\bff = \frac{\bfa_i}{\norm{\bfa_i}_2}$.
        So vectors $\frac{\bfa_1}{\norm{\bfa_1}_2}, \dots, \frac{\bfa_k}{\norm{\bfa_k}_2}$ are pareto-optimal w.r.t. $S$.
        In fact, all vectors on the unit shell and cone $\calK_A$ generated by rows of $\bfA$ are pareto-optimal w.r.t. $S$.
        That is, $\ParetoOpt(S) = \bbS_2^{d-1} \cap \calK_A$.

        So $S$ satisfies optimality if $\bbS_2^{d-1} \cap \calK_A \subseteq \bbS_2^{d-1} \cap \calK_I$.
        Any matrix $\bfA \subseteq \bfI_d$ implies $\calK_A \subseteq \calK_I$.
        Hence, we can choose any coordinate $j \in [d]$ and construct 1-dimensional $S(\bff) = \bff_j$.
        This design with dimensionality $k = 1$ satisfies the optimality objective under \restrictionCS restriction.

        Trivially, $k \ge 1$ is necessary as well.

        \item Let $\calF = \bbB_\infty^d$, the unit $L_\infty$-ball centered at the origin.
        It is easy to see that $\ParetoOpt(\calF) = \set{ \bfone_d }$, a singleton set.

        Under the \restrictionCS restriction, $S : \calF \to \bbR^k$ is such that $S(\bff) = [\bff_{i_1}; \dots; \bff_{i_k}]$ where the every index $i_j \in [d]$.
        Let $I$ be the set of unique indices.
        We will now show that if $k < d$, then there does not exist score function $S$ that satisfies optimality.
        Since $k < d$, we have $\abs{I} < d$.
        The point $\bff \in \bbB_\infty^d$ is pareto-optimal w.r.t. $S$ if $\bff_i = 1$ for every $i \in I$.
        Precisely, $\ParetoOpt(S) = \set{ \bff \in [-1, 1]^d \mid \bff_i = 1 \text{ for all } i \in I }$.
        Since there exists $j \in [d]$ that is not in $I$, $\ParetoOpt(S)$ contains points with $\bff_j = -1$.
        Hence, $\ParetoOpt(S)$ is not a subset of $\ParetoOpt(\calF)$.
        Therefore, for $\calF = \bbB_\infty^d$ and $k < d$ it is not possible to design $S : \calF \to \bbR^d$ that satisfies optimality objective under \restrictionCS restriction. 

        Trivially, $k = d$ is sufficient to satisfy the optimality objective under \restrictionCS restriction: design $S(\bff) = \bff$.
        Hence, $k \ge d$ is both necessary and sufficient when $\calF = \bbB_\infty^d$.
        \qedhere
    \end{enumerate}
\end{proof}

\subsection{Minimal design problem for both objectives simultaneously}

\begin{corollary}
    \label{cor_w_proof:1_2__f_linear__suff}
    Let columns of $\bfZ$ be an orthonormal basis of linear subspace $\calL$ associated with $\affine(\calF)$.
    For each design restriction, there exists score function $S : \calF \to \bbR^k$ that simultaneously satisfies improvement and optimality objectives with following dimensionalities.
    
    \noindent
    \begin{center}
        \begin{tabular}{l|c}
             & Dimensionality $k \ge$ \\
            \midrule
            \restrictionCS & $\ConeSubsetRank(\bfZ)$ \\
            \restrictionLM & $\ConeGeneratingRank(\bfZ)$ \\
            \restrictionL & $\ConeGeneratingRank(\bfZ)$ \\
        \end{tabular}
    \end{center}
    Moreover, for \restrictionCS and \restrictionLM restrictions, the score design is minimal when $\calF$ has non-empty relative interior.
\end{corollary}
\begin{proof}
    For the first two restrictions (\restrictionCS and \restrictionLM), $S$ is monotone in $\calF$.
    So, \Cref{thm:1__f_linear__suff,thm:improvement_monotone__imply__optimality} immediately give the design for simultaneously satisfying both objectives with dimensionality $k = \CSR(\bfZ)$ and $\CGR(\bfZ)$ respectively.
    \Cref{thm:1__f_linear__nonempty_relint_f_with_origin__necsuff} proves the minimality of this design.
    The design for \restrictionLM restriction also applies for the \restrictionL restriction.
\end{proof}

%% file: tikz/cone_inR5_5d_nonpointed.tex
\tdplotsetmaincoords{80}{45}
\tdplotsetrotatedcoords{-90}{180}{-90}
\begin{tikzpicture}[tdplot_main_coords,thick,scale=1.5]
    \coordinate (O) at (0, 0, 0);

    % axes
    \draw[->,color=gray,dashed] (-1.25,0,0) -- (1.25,0,0);
    \draw[->,color=gray,dashed] (0,-1.25,0) -- (0,1.25,0);
    \draw[->,color=gray,dashed] (0,0,-0.25) -- (0,0,1.25);

    % % back face of triangle cone
    % \draw[fill=Cerulean,fill opacity=0.3,draw=none] (O) -- (0,1,1) -- (-1,0,1) -- cycle;
    % \draw[draw=Cerulean] (O) -- (0,1,1);
    % \draw[draw=Cerulean] (O) -- (-1,0,1);

    % draw square cone
    \begin{scope}[thick,draw=red]
        % back face
        \draw[fill=Salmon,fill opacity=0.3] (O) -- (-0.25,0.75,1) -- (-0.75,0.25,1) -- cycle;
        % \draw[canvas is xy plane at z=1,fill=Salmon,fill opacity=0.3] (45-15:0.5) arc (45-15:225+15:0.5) -- (O) -- cycle;
        
        % left face
        \draw[fill=Salmon,fill opacity=0.3] (O) -- (-0.75,0.25,1) -- (0,-0.5,1) -- cycle;
        
        % right face
        \draw[fill=Salmon,fill opacity=0.3] (O) -- (-0.25,0.75,1) -- (0.5,0,1) -- cycle;
        
        % front face
        \draw[fill=Salmon,fill opacity=0.3] (O) -- (0,-0.5,1) -- (0.5,0,1) -- cycle;
    \end{scope}

    % draw black points on the square vertices
    \node[circle,inner sep=2pt,fill=black] at (-0.25,0.75,1) {};
    \node[circle,inner sep=2pt,fill=black] at (-0.75,0.25,1) {};
    \node[circle,inner sep=2pt,fill=black] at (0,-0.5,1) {};
    \node[circle,inner sep=2pt,fill=black] at (0.5,0,1) {};

    % % left face of trianglular cone
    % \draw[fill=Cerulean,fill opacity=0.2,draw=none] (O) -- (-1,0,1) -- (1,-1,1) -- cycle;
    % % right face of trianglular cone
    % \draw[fill=Cerulean,fill opacity=0.2,draw=none] (O) -- (0,1,1) -- (1,-1,1) -- cycle;
    % \draw[draw=Cerulean] (O) -- (1,-1,1);

    % \node[circle,inner sep=2pt,fill=Cerulean] at (-1,0,1) {};
    % \node[circle,inner sep=2pt,fill=Cerulean] at (0,1,1) {};
    % \node[circle,inner sep=2pt,fill=Cerulean] at (1,-1,1) {};
    
    % \node[text=Cerulean] at (-0.5,-0.5,0.5) {$\calK_V$};
    % \node[text=red] at (0.5,0.5,0.5) {$\calK_Z$};
\end{tikzpicture}
\begin{tikzpicture}[thick,scale=0.6]
    % plus sign in the middle
    \draw[-,color=black] (-4,0) -- (-3,0);
    \draw[-,color=black] (-3.5,-0.5) -- (-3.5,0.5);

    \coordinate (O) at (0, 0);
    \coordinate (right) at (2, 0);
    \coordinate (left) at (-2, 0);
    \coordinate (top) at (0, 2);
    \coordinate (bottom) at (0, -2);

    \draw[<->,color=gray,dashed] (left) -- (right);
    \draw[<->,color=gray,dashed] (bottom) -- (top);
    
    \coordinate (A) at (1, 1);
    \coordinate (B) at (-1, 1);
    \coordinate (C) at (-1, -1);
    \coordinate (D) at (1, -1);
    
    \coordinate (topright) at (2, 2);
    \coordinate (topleft) at (-2, 2);
    \coordinate (bottomright) at (2, -2);
    \coordinate (bottomleft) at (-2, -2);

    % \begin{scope}[ultra thick]
    %     \draw[->] (O) -- (Aext);
    %     \draw[->] (O) -- (Bext);
    % \end{scope}

    \draw[fill=Salmon,opacity=0.5,draw=none] (topright) -- (topleft) -- (bottomleft) -- (bottomright);

    \begin{scope}[font=\tiny]
    \node[circle,inner sep=2pt,fill=black] at (A) {};
    \node[circle,inner sep=2pt,fill=black] at (B) {};
    \node[circle,inner sep=2pt,fill=black] at (C) {};
    \node[circle,inner sep=2pt,fill=black] at (D) {};
    \node[circle,inner sep=2pt,fill=Cerulean] at (0,1) {};
    \end{scope}
\end{tikzpicture}

%% file: app_algorithms.tex
\clearpage
\section{Algorithms}
\label{app:algorithms}

\begin{table}[!h]
    \centering
    \begin{tabular}{l l}
        \toprule
        Reference & Description\\
        \midrule
        \Cref{alg:is_in_cone} & Check if a vector is in a polyhedral cone\\
        \Cref{alg:is_cone_pointed} & Determine if cone $\calK_W$ generated by $\bfW$ is pointed\\
        \Cref{alg:cone_minkowski_decomposition} & Decompose a cone $\calK_W$ into its lineality space $\calL$ and pointed $\calK_P$\\
        \midrule
        \Cref{alg:csr} & Find $\bfV$ that attains $\ConeSubsetRank(\bfW)$\\
        \Cref{alg:csr__subspace} & Find $\bfV$ that attains $\ConeSubsetRank(\bfW)$ when $\calK_W$ is a linear subspace\\
        \Cref{alg:csr__pointed} & Find $\bfV$ that attains $\ConeSubsetRank(\bfW)$ when $\calK_W$ is pointed\\
        \midrule
        \Cref{alg:cgr} & Find $\bfV$ that attains $\ConeGeneratingRank(\bfW)$\\
        \midrule
        \Cref{alg:cr} & Find $\bfV$ that attains $\ConeRank(\bfW)$\\
        \Cref{alg:cr__pointed} & Find $\bfV$ that attains $\ConeRank(\bfW)$ when $\calK_W$ is pointed\\
        \bottomrule
    \end{tabular}
    \caption{List of algorithms}
    \label{tab:app_algorithms_list}
\end{table}

\subsection{Properties of polyhedral cones}

\begin{algorithm}[!h]
    \begin{algorithmic}[1]
        \Procedure{IsInCone}{$\bfx$, $\bfW = [\bfw_1; \dots; \bfw_m]$}
            \State Solve LP:
            \begin{align*}
                \min_{\bflambda \in \bbR^m} 1 \text{ s.t. } \bfx = \bflambda \bfW, \bflambda \ge \bfzero.
            \end{align*}
            \State If the LP is feasible \Return True, else \Return False
        \EndProcedure
        \caption{Check if a vector is in a polyhedral cone}
        \label{alg:is_in_cone}
    \end{algorithmic}
\end{algorithm}

\begin{lemma}
\label{lemma:alg__is_in_cone}
    \Cref{alg:is_in_cone} returns True if and only if $\bfx \in \calK_W$.
    The runtime is $\tilde{O} \inparens{ m n^{2.5} }$.
\end{lemma}
\begin{proof}
    The LP is feasible if and only if $\bfx$ can be written as a nonnegative combination of rows of $\bfW$.\\

    \noindent
    \textit{Runtime:} As $\bfW \in \bbR^{m \times n}$, there are $m$ variables and $n$ constraints.
    The runtime of the LP is thus $\tilde{O} (mn^{1.5} + n^{2.5}) = \tilde O(mn^{2.5})$~\cite{lee2015efficient}, where $\tilde O$ hides $\polylog(m,n)$ factors.
\end{proof}

\begin{algorithm}[!h]
    \begin{algorithmic}[1]
        \Procedure{IsConePointed}{$\bfW = [\bfw_1; \dots; \bfw_m]$}
            \State Require: $\bfW$ does not contain any rows of all zeros
            \State Solve \ref{eqn:LP__pointed}:
            \begin{align*}
                \tag{\ref{eqn:LP__pointed}}
                &\min_{\bflambda \in \bbR^m} 1 \text{ s.t. } \bflambda \bfW = \bfzero, \quad \bflambda \ge \bfzero, \quad \bflambda \cdot \bfone = 1.
            \end{align*}
            \State If the LP is feasible \Return False, else \Return True
        \EndProcedure
        \caption{Determine if cone $\calK_W$ generated by $\bfW$ is pointed}
        \label{alg:is_cone_pointed}
    \end{algorithmic}
\end{algorithm}

\begin{lemma}
\label{lemma:alg__is_cone_pointed}
    \Cref{alg:is_cone_pointed} returns True if and only if $\calK_W$ is pointed.
    The runtime is $\tilde{O} \inparens{ m n^{2.5} }$.
\end{lemma}
\begin{proof}
    We will prove the contrapositive: algorithm returns False if and only if $\calK_W$ is non-pointed.

    First, we prove that if algorithm returns False, then $\calK_W$ is non-pointed.
    Algorithm returns False if the LP is feasible, i.e., there exists $\bflambda \ge \bfzero$ such that $\bflambda \bfW = \bfzero$ and $\bflambda \cdot \bfone = 1$.
    Hence, there exists nonzero nonnegative $\bflambda$ such that $\bflambda \bfW = \bfzero$.
    \Cref{lemma:nonpointed__zero_in_cone} then tells us that $\calK_W$ is non-pointed.

    Now we prove: if $\calK_W$ is non-pointed, then algorithm returns False.
    \Cref{lemma:nonpointed__zero_in_cone} tells us that there exists nonzero $\bfmu \ge \bfzero$ such that $\bfzero = \bfmu \bfW$.
    Let $\bflambda = \frac{\bfmu}{\norm{\bfmu}_1}$.
    Note that $\bflambda \neq \bfzero, \bflambda \ge \bfzero$, and $\bflambda \cdot \bfone = 1$.
    And finally, $\bflambda \bfW = \bfzero$.
    The LP is thus feasible, so the algorithm returns False.\\

    \noindent
    \textit{Runtime:} We write \ref{eqn:LP__pointed} in the canonical form as $\min_{\bflambda \in \bbR^m} 1$ such that $\bflambda [\bfW; \bfone] = [\bfzero; 1]$ and $\bflambda \ge \bfzero$.
    As $\bfW \in \bbR^{m \times n}$, there are $m$ variables and $n+1$ constraints.
    The runtime of the LP is thus $\tilde {O}(m n^{1.5} + n^{2.5}) = \tilde O(mn^{2.5})$~\cite{lee2015efficient}.
\end{proof}

\begin{algorithm}[!htbp]
    \begin{algorithmic}[1]
        \Procedure{DecomposeCone}{$\bfW = [w_1; \dots; w_m]$}
        \State Require: $\bfw_i \neq \bfzero$ for all $i \in [m]$
        \State Initialize $\round{\bfL}{0} = \emptyset$ and $\round{\bfP}{0} = \bfW$
        \For{$t = 1, \dots, $}
            \State Solve \ref{eqn:LP__pointed}:
            \begin{align}
                \tag{\ref{eqn:LP__pointed}}
                \bfalpha^\ast &= \argmin_{\bfalpha} 1 \text{ s.t. } \bfalpha \round{\bfP}{t-1} = \bfzero,\quad \bfalpha \ge \bfzero, \quad \bfalpha \cdot \bfone = 1.
            \end{align}
            \If{LP is infeasible}
                \State $T \gets t-1$ and break
            \Else
                \State $\round{\bfL}{t} \gets [\round{\bfL}{t-1}; \bfp_1; \bfp_2; \dots]$ for $\bfp_i \in \round{\bfP}{t-1}$ with $\alpha_i^\ast > 0$
                \State $\round{\bfP}{t} \gets \round{\bfP}{t-1} (\bfI - \bfZ \bfZ^\top)$ where columns of $\bfZ$ are orthonormal basis of $\rowspan(\round{\bfL}{t})$
            \EndIf
        \EndFor
        \State \Return $(\bfL = \round{\bfL}{T}, \bfP = \round{\bfP}{T})$
        \EndProcedure
    \end{algorithmic}
    \caption{Decompose a cone $\calK_W$ into its lineality space $\calL$ and pointed cone $\calK_P = \calL^\perp \cap \calK_W$}
    \label{alg:cone_minkowski_decomposition}
\end{algorithm}

\begin{lemma}
\label{lemma:alg__cone_minkowski_decomposition}
    Let cone $\calK_W$ generated by $\bfW \in \bbR^{m \times n}$ have decomposition $\calK_W = \calH + \calQ$, where $\calH = \lineality(\calK_W)$ is the lineality space and $\calQ = \calL^\perp \cap  \calK_W$ is a pointed cone.
    Then \Cref{alg:cone_minkowski_decomposition} returns matrices $(\bfL, \bfP)$ such that the $\calK_L = \calH$ and $\calK_P = \calQ$.
    The runtime is $\tilde O(m^2 n^{2.5})$.
\end{lemma}
\begin{proof}
    \Cref{lemma:restated__cone_decomposition} states that the cone $\calK_W$ can be uniquely decomposed as $\calL + \calK_P$ where $\calL = \lineality(\calK_W) = \calK_W \cap (-\calK_W)$ is the lineality space and $\calK_P = \calL^\perp \cap \calK_W$ is a pointed cone \cite[Sec.~8.2]{schrijver1998theory}. 
    Here $\calL$ is itself a cone.
    Let the algorithm return matrices $\bfL, \bfP$ that generate cones $\calK_L, \calK_P$ respectively.
    Denote $\calL = \rowspan(\bfL)$.
    Due to unique decomposition result, it suffices to show that $\bfL, \bfP$ satisfy properties: (1) $\calK_L = \calL$, (2) $\calK_P \perp \calL$, (3) $\calK_W = \calL + \calK_P$, and (4) $\calK_P$ is a pointed cone.
    We prove each property below.
    
    \begin{enumerate}
        \item We show that the cone $\calK_L$ generated by $\bfL$ is in fact a linear subspace, implying that $\calK_L = \calL$.
        
        We prove by induction on $t$ that $\calK_L$ is a linear subspace.
        The base case $t = 0$ is trivial as $\bfL = \emptyset$.
        For the inductive step $t$, assume the inductive hypothesis for step $t-1$ that $\calK_{\round{L}{t-1}}$ is a linear subspace.
        We show that $\calK_{\round{L}{t}}$ is a linear subspace by showing that it contains every linear combination of its constituent vectors.
        
        In round $t$ of the algorithm, let rows $\bfp_i$ of $\round{\bfP}{t-1}$ with $\alpha_i^\ast > 0$ be arranged as rows of $\bfV$.
        Let $\bfx, \bfy \in \calK_{\round{L}{t}}$, i.e., $\bfx = \bflambda \round{\bfL}{t-1} + \bflambda' \bfV$ and $\bfy = \bfmu \round{\bfL}{t-1} + \bfmu' \bfV$ for some $\bflambda, \bflambda', \bfmu, \bfmu' \ge \bfzero$.
        For $c_1, c_2 \in \bbR$,
        \begin{align}
            c_1 \bfx + c_2 \bfy &= (c_1 \bflambda + c_2 \bfmu) \round{\bfL}{t-1} + (c_1 \bflambda' + c_2 \bfmu') \bfV.
        \end{align}

        The inductive hypothesis says that cone $\calK_{\round{L}{t-1}}$ is a linear subspace.
        Hence, $(c_1 \bflambda + c_2 \bfmu) \round{\bfL}{t-1} = \bfnu \round{\bfL}{t-1}$ for some $\bfnu \ge \bfzero$.
        Using the fact that $\bfalpha^\ast \bfV = \bfzero$ for $\bfalpha^\ast > \bfzero$ from the LP's feasibility, we get
        
        \begin{align}
            c_1 \bfx + c_2 \bfy &= (c_1 \bflambda + c_2 \bfmu) \round{\bfL}{t-1} + (c_1 \bflambda' + c_2 \bfmu') \bfV\\
            &= \bfnu \round{\bfL}{t-1} + (c_1 \bflambda' + c_2 \bfmu') \bfV\\
            &= \bfnu \round{\bfL}{t-1} + (c_1 \bflambda' + c_2 \bfmu') \bfV - \min_i \inbraces{ \frac{ c_1 \lambda'_i + c_2 \mu'_i }{\alpha_i^\ast} } \bfalpha^\ast \bfV\\
            &= \bfnu \round{\bfL}{t-1} + \underbrace{ \inparens{ c_1 \bflambda' + c_2 \bfmu' - \min_i \inbraces{ \frac{ c_1 \lambda'_i + c_2 \mu'_i }{\alpha_i^\ast} } \bfalpha^\ast } }_{\bfbeta} \bfV.
        \end{align}

        Now observe that $\bfbeta$ is a nonnegative vector, as $\bfalpha^\ast > \bfzero$ and for each coordinate $j$ we have
        \begin{align}
            \beta_j &= c_1 \lambda'_j + c_2 \mu'_j - \min_i \inbraces{ \frac{c_1 \lambda'_i + c_2 \mu'_j}{\alpha^\ast_i} } \alpha^\ast_j\\
            &= \alpha^\ast_j \inbraks{ \frac{c_1 \lambda'_j + c_2 \mu'_j}{\alpha^\ast_i} - \min_i \inbraces{ \frac{c_1 \lambda'_i + c_2 \mu'_j}{\alpha^\ast_i} } }\\
            &\ge \alpha^\ast_j > 0.
        \end{align}
        
        Hence, every linear combination of $\bfx$ and $\bfy$ is in $\calK_{\round{L}{t}}$, proving that it is a linear subspace.

        \item As rows of $\round{\bfP}{t}$ are projections of rows of $\round{\bfP}{t-1}$ onto orthogonal complement $\inparens{ \round{\calL}{t} }^\perp$, the generated cone $\calK_{\round{P}{t}} \perp \round{\calL}{t}$ for all $t$.

        \item To prove that $\calK_W = \calL + \calK_P$, we will show that $\round{\calL}{t-1} + \calK_{\round{P}{t-1}} = \round{\calL}{t} + \calK_{\round{P}{t}}$ for all rounds $t \ge 1$.
        Using the initialization $\round{\calL}{0} = \set{\bfzero}$ and $\calK_{\round{P}{0}} = \calK_W$, we get the desired result $\calK_W = \round{\calL}{T} + \calK_{\round{P}{T}}$.

        In round $t$ of the algorithm, let rows $\bfp_i$ of $\round{\bfP}{t-1}$ with $\alpha_i^\ast > 0$ be arranged as rows of $\bfV$.
        We make two observations.
        First, the proof of property (1) notes that both $\round{\bfL}{t-1}$ and $\round{\bfL}{t}$ generate linear subspaces.
        Second, property (2) notes that rows of $\round{\bfP}{t-1} \subseteq \inparens{ \round{\calL}{t-1} }^\perp$, implying that rows of $\bfV$ are orthogonal to $\round{\calL}{t-1}$.
        Since $\round{\bfL}{t} = [\round{\bfL}{t-1}; \bfV]$, these two observations imply that $\bfV$ generates a linear subspace $\calV$.
        Moreover, this linear subspace satisfies $\calV \perp \round{\calL}{t-1}$ and $\round{\calL}{t} = \calV + \round{\calL}{t-1}$.
        As $\bfV$ are among the rows of $\round{\bfP}{t-1}$, the subspace $\calV$ lies inside the cone $\calK_{\round{P}{t-1}}$.
        So we can decompose this cone as follows:
        \begin{align}
            \calK_{\round{P}{t-1}} &= \proj_{\calV} \inparens{ \calK_{\round{P}{t-1}} } + \proj_{\calV^\perp} \inparens{ \calK_{\round{P}{t-1}} }\\
            &= \calV + \proj_{\calV^\perp} \inparens{ \calK_{\round{P}{t-1}} }\\
            &= \calV + \proj_{\calV^\perp \cap \inparens{ \round{\calL}{t-1} }^\perp } \inparens{ \calK_{\round{P}{t-1}} }\\
            &= \calV + \proj_{\inparens{ \calV + \round{\calL}{t-1} }^\perp } \inparens{ \calK_{\round{P}{t-1}} }
        \end{align}
        where we used the facts $\calK_{\round{P}{t-1}} \subseteq \inparens{ \round{\calL}{t-1} }^\perp$ and $(\calU_1 + \calU_2)^\perp = \calU_1^\perp + \calU_2^\perp$ for any two subspaces $\calU_1, \calU_2$.
        We substitute $\round{\calL}{t} = \calV + \round{\calL}{t-1}$ to simplify as follows
        \begin{align}
            \round{\calL}{t-1} + \calK_{\round{P}{t-1}} &= \round{\calL}{t-1} + \calV + \proj_{\inparens{ \calV + \round{\calL}{t-1} }^\perp } \inparens{ \calK_{\round{P}{t-1}} }\\
            &= \round{\calL}{t} + \proj_{ \inparens{\round{\calL}{t}}^\perp } \inparens{ \calK_{\round{P}{t-1}} }\\
            \label{eqn:cone_minkowski_decomposition__prop4__a}
            &= \round{\calL}{t} + \cone \inparens{ \proj_{ \inparens{\round{\calL}{t}}^\perp } ( \round{P}{t-1} ) }\\
            &= \round{\calL}{t} + \cone (\round{P}{t})
        \end{align}
        where \Cref{eqn:cone_minkowski_decomposition__prop4__a} follows from \Cref{lemma:project_cone_interchangeable}.
        
        \item We will show that $\calK_P = \calK_{\round{P}{T}}$ is a pointed cone by contradiction.
        Assume that it is non-pointed.
        \Cref{lemma:nonpointed__zero_in_cone} tells us that there exists nonzero $\bflambda \ge \bfzero$ such that $\bfzero = \bflambda \round{\bfP}{T}$.
        So there exists $\bfalpha = \frac{\bflambda}{\norm{\bflambda}_1} \ge \bfzero$ with $\norm{\bfalpha}_1 = 1$ and $\bfalpha \round{\bfP}{T} = \bfzero$.
        Hence, LP in the algorithm is feasible, which is a contradiction because the algorithm terminated in round $T$.
        Therefore, $\calK_P$ must be a pointed cone.
    \end{enumerate}

    Using properties (1)--(4) and the unique decomposition result due to \cite{schrijver1998theory}, we get that output $\bfL$ generates $\lineality(\calK_W)$ and output $\bfP$ generates the pointed cone $\calK_W \cap \lineality(\calK_W)^\perp$.\\

    \noindent
    \textit{Runtime:}
    In any round $t$, the number of rows of $\round{\bfL}{t}, \round{\bfP}{t}$ is at most $m$.
    Let $\round{\bfP}{t-1}$ have $k \le m$ rows in round $t$.
    \Cref{lemma:alg__is_cone_pointed} states that solving \ref{eqn:LP__pointed} for $\round{\bfP}{t-1} \in \bbR{k \times n}$ has runtime $\tilde O (kn^{2.5})$.
    After populating $\round{\bfL}{t}$, we find an orthonormal basis of $\rowspan(\round{\bfL}{t})$ in time $O(mn^2)$ and compute $\round{\bfP}{t}$ in time $O(k n^2)$.
    So the runtime of each iteration is $\tilde O(kn^{2.5} + mn^2 + kn^2) = \tilde O(mn^{2.5})$.
    As rows of $\round{\bfP}{t-1}$ get separated into $\round{\bfL}{t}$ and $\round{\bfP}{t}$, there are at most $m$ iterations.
    So the total runtime is $\tilde O(m^2 n^{2.5})$.
\end{proof}

\subsection{Computing $\ConeSubsetRank$}
\label{app:algorithms__csr}

We give an algorithm in \Cref{app:algorithms__csr__general} to find $\bfV$ that attains $\ConeSubsetRank(\bfW)$.
In \Cref{app:algorithms__csr__pointed} we handle the case when the cone $\calK_W$ generated by $\bfW$ is pointed.
In \Cref{app:algorithms__csr__subspace} we handle the case when $\calK_W$ is a linear subspace.

We often write $\CSR(W)$ for set $W = \set{\bfw_1, \dots, \bfw_m } \subseteq \bbR^n$ to denote $\CSR(\bfW)$ of the matrix $\bfW = [\bfw_1; \dots; \bfw_m] \in \bbR^{m \times n}$.
That is, $\CSR$ can be equivalently defined for sets as follows:
\begin{align*}
    \CSR(W) &= \min_k \set{ k \mid \calK_W = \calK_V  \text{ for some } V \subseteq \bbR^n \text{ such that } V \subseteq W, |V| = k }.
\end{align*}

We say that $V^\ast$ \textit{attains} $\CSR(W)$ if $V^\ast \subseteq W, \calK_W = \calK_{V^\ast}$, and $\abs{V^\ast} = \CSR(W)$.

\subsubsection{Computing $\ConeSubsetRank$ for general $\calK_W$}
\label{app:algorithms__csr__general}

\Cref{alg:csr} finds $\bfV$ that attains $\CSR(\bfW)$ (proof in \Cref{lemma:alg__csr}). 
The algorithm first decomposes the $\calK_W$ into its lineality space $\calL$ and a pointed cone. 
It turns out that to find $\bfV$, we can focus on rows of $\bfW$ inside $\calL$ and outside $\calL$ separately, as shown in \Cref{lemma:csr__property_of_decomposition}.
We attain the $\CSR$ of rows of $\bfW$ inside $\calL$ using \Cref{alg:csr__subspace}.
Then we project the rows of $\bfW$ outside $\calL$ onto $\calL^\perp$, and attain their $\CSR$ using \Cref{alg:csr__pointed}.

\begin{algorithm}[!htbp]
    \begin{algorithmic}[1]
        \Procedure{GetCSR}{$\bfW = [\bfw_1; \dots; \bfw_m]$}
        \State $\bfL, \bfP \gets \textsc{DecomposeCone}(\bfW)$ and denote $\calL = \rowspan(\bfL)$
        \Comment{Use \Cref{alg:cone_minkowski_decomposition}}
        \State Partition $W$ into $\WinL = \set{ \bfw_i \in W \mid \bfw_i \in \calL }$ and $\WnotinL = \set{ \bfw_i \in W \mid \bfw_i \notin \calL }$.
        \If{$\calL = \set{\bfzero}$}
            \Comment{$\calK_W$ is pointed}
            \State $\bfVinL \gets \emptyset$ and $\tilde \bfW \gets \bfW$
        \Else
            \Comment{$\calK_W$ is non-pointed}
            \State $\bfVinL \gets \textsc{GetCSR-Subspace}(\bfWinL)$
            \Comment{Use \Cref{alg:csr__subspace}. $\CSR$ of $\calK_W \cap \calL$}
            \State $\tilde \bfW \gets \bfWnotinL (\bfI - \bfZ \bfZ^\top)$ where columns of $\bfZ$ are orthonormal basis of $\calL$
        \EndIf
        \State $\tilde \bfV \gets \textsc{GetCSR-Pointed}(\tilde \bfW)$
        \Comment{Use \Cref{alg:csr__pointed}}
        \State Find rows $\bfVnotinL \subseteq \bfWnotinL$ that correspond to rows of $\tilde \bfV$
        \State \Return $\bfV = [\bfVinL; \bfVnotinL]$
        \EndProcedure
    \end{algorithmic}
    \caption{Find $\bfV$ that attains $\ConeSubsetRank(\bfW)$}
    \label{alg:csr}
\end{algorithm}

\begin{remark}
\label{remark:csr__general_calls_pointed}
    When $\calK_W$ is pointed, its lineality space is $\calL = \set{\bfzero}$.
    So \Cref{alg:csr} sets $\bfVnotinL = \emptyset$ and $\tilde \bfW = \bfW$.
    Then the algorithm finds $\tilde \bfV$ that attains $\CSR(\tilde \bfW) = \CSR(\bfW)$ using \textsc{GetCSR-Pointed} in \Cref{alg:csr__pointed}.
    Since \textsc{GetCSR-Pointed} ensures $\tilde \bfV \subseteq \bfW$ and rows of $\tilde \bfV$ are nonzero, the recovery step in Line~10 becomes redundant.
    \Cref{alg:csr} thus returns $\bfV$ that is the output of \textsc{GetCSR-Pointed}$(\bfW)$.
\end{remark}

\begin{lemma}
\label{lemma:alg__csr}
    Given $\bfW \in \bbR^{m \times n}$, \Cref{alg:csr} finds $\bfV$ that attains $\CSR(\bfW)$ in $m^3 n^{2.5} \cdot \polylog(m,n) \cdot m^{O(\ell)}$ time where $\ell = \dim \lineality(\calK_W)$.
\end{lemma}
\begin{proof}
    When $\calK_W$ is pointed, the algorithm outputs \textsc{GetCSR-Pointed}$(\bfW)$ using \Cref{alg:csr__pointed}, as noted in \Cref{remark:csr__general_calls_pointed}.
    We give the proof for the pointed case in \Cref{lemma:alg__csr__pointed}.
    Here we prove for the case of non-pointed $\calK_W$.
    We first show that the output of the algorithm $V \subseteq \bbR^n$ has the properties $V \subseteq W$ and $\calK_W = \calK_V$.
    Then we show that, when $\calK_W$ is non-pointed, $|V| = \CSR(W)$.

    Cone $\calK_W$ has unique decomposition $\calK_W = \calL + \calK_P$ where $\calL = \lineality(\calK_W)$ is the lineality space of $\calK_W$ and $\calK_P = \calL^\perp \cap \calK_W$ is a pointed cone (see \Cref{lemma:restated__cone_decomposition}).
    \Cref{alg:cone_minkowski_decomposition} outputs matrices $\bfL$ and $\bfP$ such that $\rowspan(\bfL) = \calL$ and $\bfP$ generates the pointed cone $\calK_P$.
    After decomposing, the algorithm partitions $W$ into $\WinL$ and $\WnotinL$, and projects $\WnotinL$ onto $\calL^\perp$ to get $\tilde W = \proj_{\calL^\perp} (\WnotinL)$.
    To prove that the output $V$ has desired properties, we note the properties of $\VinL$ and $\VnotinL$.

    \begin{itemize}
        \item \boldheading{Properties of $\VinL$.}
        According to \Cref{lemma:cone__lineality_space__lineal_points}, vectors $\WinL$ positively span $\calL$, which is a linear subspace.
        The algorithm then uses \Cref{alg:csr__subspace} to find $\VinL$ that attains $\CSR(\WinL)$.
        That is, $\VinL$ has properties $\VinL \subseteq \WinL$, $\calK_{\VinL} = \calK_{\WinL} = \calL$, and $|\VinL| = \CSR(\WinL)$.

        \item \boldheading{Properties of $\VnotinL$.}
        According to \Cref{lemma:cone__lineality_space__nonlineal_points}, vectors $\tilde W$ generate the pointed cone $\calK_P$.
        The algorithm then uses \Cref{alg:csr__pointed} to find $\tilde V$ that attains $\CSR(\tilde W)$.
        So $\tilde V$ has properties $\tilde V \subseteq \tilde W$, $\calK_{\tilde V} = \calK_{\tilde W} = \calK_P$, and $|\tilde V| = \CSR(\tilde W)$.
        As $\tilde W$ are projections of $\WnotinL$ onto $\calL^\perp$ and $\tilde V \subseteq \tilde W$, the algorithm can recover $\VnotinL$ with properties: $\VnotinL \subseteq \WnotinL$, $\cone(\proj_{\calL^\perp} (\VnotinL)) = \calK_{\tilde W}$, and $|\VnotinL| = \CSR(\tilde W)$.
    \end{itemize}

    Clearly, we have $V = \VinL \cup \VnotinL \subseteq W$.
    We now show that $\calK_V = \calK_W$.
    Below, `$+$' is the Minkowski sum.
    \begin{align}
        \calK_V &= \calK_{\VinL \cup \VnotinL} = \calK_{\VinL} + \calK_{\VnotinL}\\
        \label{eqn:lemma__alg__csr__a}
        &= \calL + \cone \inparens{ \proj_{\calL} (\VnotinL) + \proj_{\calL^\perp} (\VnotinL) }\\
        \label{eqn:lemma__alg__csr__b}
        &= \calL + \proj_{\calL} \inparens{ \cone(\VnotinL) } + \cone (\tilde V)\\
        \label{eqn:lemma__alg__csr__c}
        &= \calL + \calK_P = \calK_W.
    \end{align}
    
    \Cref{eqn:lemma__alg__csr__a} uses $\calL = \calK_{\VinL}$ and the decomposition of vectors $\VnotinL = \proj_{\calL} (\VnotinL) + \proj_{\calL^\perp} (\VnotinL)$.
    \Cref{eqn:lemma__alg__csr__b} is due to \Cref{lemma:project_cone_interchangeable} and $\tilde V = \proj_{\calL^\perp} (\VnotinL)$.
    And \Cref{eqn:lemma__alg__csr__c} follows from the fact that $\proj_{\calL} (\VnotinL) \subseteq \calL$, and so $\proj_{\calL} (\cone(\VnotinL)) \subseteq \calL$ as $\calL$ is a linear subspace.

    \Cref{lemma:csr__property_of_decomposition} shows that any $V \subseteq W$ such that $\calK_V = \calK_W$ must have size $|V| \ge \CSR(\WinL) + \CSR(\tilde W)$.
    Above, we showed the output $V$ of \Cref{alg:csr} has the properties: $V \subseteq W$, $\calK_V = \calK_W$, and $|V| = |\VinL| + |\VnotinL| = \CSR(\WinL) + \CSR(\tilde W)$.
    Hence, $V$  attains $\CSR(W)$ and $\CSR(W)=\CSR(\WinL) + \CSR(\tilde W)$.\\

    \noindent
    \textit{Runtime:}
    \Cref{lemma:alg__cone_minkowski_decomposition} states that \textsc{DecomposeCone} in \Cref{alg:cone_minkowski_decomposition} has runtime $\tilde O(m^2 n^{2.5})$.
    We can check with Gaussian elimination if a row $\bfw_i$ of $\bfW$ is inside $\calL$ or not---this has runtime $O(mn \min(m,n))$ as $\bfL$ has atmost $m$ rows.
    So we can partition $W$ into $\WinL, \WnotinL$ in time $O(m^2 n \min(m,n))$.

    Let $\dim \calL = \ell$.
    \Cref{lemma:alg__csr__subspace} states that \textsc{GetCSR-Subspace} in \Cref{alg:csr__subspace} has runtime $\ell^2 n^{2.5} \cdot \polylog(\ell,n) \cdot m^{O(\ell)}$.
    We find an orthonormal basis of $\calL$ in time $O(mn^2)$ and compute $\tilde \bfW$ in time $O(mn^2)$.

    The matrix $\tilde \bfW$ has at most $m$ rows.
    Then \Cref{lemma:alg__csr__pointed} states that \textsc{GetCSR-Pointed} in \Cref{alg:csr__pointed} has runtime $\tilde O(m^3 n^{2.5})$.
    We can finally recover $\bfVnotinL$ by inspection in time $O(mn)$.

    Adding all runtimes so far, the total runtime is:
    \begin{align}
        &\tilde O(m^2 n^{2.5}) + O(m^2 n \min(m,n))\\
        &\quad + \ell^2 n^{2.5} \cdot \polylog(\ell,n) \cdot m^{O(\ell)} + O(mn^2) + O(mn^2)\\
        &\quad + \tilde O(m^3 n^{2.5}) + O(mn)\\
        &= m^3 n^{2.5} \cdot \polylog(m,n) \cdot m^{O(\ell)}.
        \qedhere
    \end{align}
\end{proof}

\begin{lemma}
\label{lemma:csr__property_of_decomposition}
    Let $W = \set{\bfw_1, \dots, \bfw_m} \subseteq \bbR^n$ generate cone $\calK_W$ and have lineality space $\calL = \lineality (\calK_W)$.
    Denote $\WinL = \set{ \bfw_i \in W \mid \bfw_i \in \calL }$, and $\tilde W  = \proj_{\calL^\perp} \inparens{  W \setminus \WinL }$.
    For any $V$ with properties $V \subseteq W$ and $\calK_{V} = \calK_W$, we have:
    \begin{align*}
        |V| &\ge \CSR(\WinL) + \CSR(\tilde W).
    \end{align*}
\end{lemma}
\begin{proof}
    Denote $\VinL = \set{ \bfv_i \in V \mid \bfv_i \in \calL }$, and $\tilde V  = \proj_{\calL^\perp} \inparens{  V \setminus \VinL }$.
    We will show that (1) $|\VinL| \ge \CSR(\WinL)$, and (2) $|\tilde V| \ge \CSR(\tilde W)$.
    As $\tilde V$ is the projection of $V \setminus \VinL$ onto $\calL^\perp$, we have $|V \setminus \VinL| \ge |\tilde V|$.
    Joining these results, we get the desired result:
    \begin{align*}
        |V| &= |\VinL| + |V \setminus \VinL| \ge |\VinL| + |\tilde V| \ge \CSR(\WinL) + \CSR(\tilde W).
    \end{align*}

    Let $\calK_W = \calL + \calK_P$ be the unique decomposition where $\calL = \lineality(\calK_W)$ and $\calK_P = \calL^\perp \cap \calK_W$ is a pointed cone (see \Cref{lemma:restated__cone_decomposition}).
    To prove that $|\VinL| \ge \CSR(\WinL)$, we show that $\VinL$ satisfies properties $\VinL \subseteq \WinL$ and $\calK_{\VinL} = \calK_{\WinL}$, and so the first statement follows from the definition of $\CSR$.
    We analogously prove that $|\tilde V| \ge \CSR(\tilde W)$.
    \begin{enumerate}
        \item
        Recall that ${\VinL}$ and ${\WinL}$ are the elements of $V$ and $W$ respectively in $\calL$, the lineality space of $\calK_{V} = \calK_W$. 
        As $V \subseteq W$, we have $\VinL \subseteq \WinL$. \Cref{lemma:cone__lineality_space__lineal_points} states that $\calK_{\WinL} = \calL$ and $\calK_{\VinL} = \calL$, and so $\calK_{\VinL} = \calK_{\WinL}$.
        By definition of $\CSR$, we have $|\VinL| \ge \CSR(\WinL)$.

        \item
        As $V \subseteq W$, we have $(V \setminus \VinL) \subseteq (W \setminus \WinL)$.
        Since $\tilde V$ and $\tilde W$ are the respective projections, we have $\tilde V \subseteq \tilde W$.
        As $\calK_P$ is the pointed cone from decomposing $\calK_W = \calK_V$, \Cref{lemma:cone__lineality_space__nonlineal_points} states that $\cone(\tilde W) = \calK_P$ and $\cone(\tilde V) = \calK_P$, and so $\cone(\tilde W) = \cone(\tilde V)$.
        By definition of $\CSR$, we have $|\tilde V| \ge \CSR(\tilde W)$.
        \qedhere
    \end{enumerate}
\end{proof}

\subsubsection{Computing $\ConeSubsetRank$ when $\calK_W$ is a linear subspace}
\label{app:algorithms__csr__subspace}

When $\calK_W$ is a linear subspace, \Cref{alg:csr__subspace} finds $\bfV$ that attains $\CSR(\bfW)$ (proof in \Cref{lemma:alg__csr__subspace}).
This algorithm iterates over all subsets of $\bfW$ that are linearly independent and positively span a linear subspace, and outputs the smallest subset.
In the proof of \Cref{lemma:alg__csr__subspace}, we argue that the $\bfV^\ast$ that attains $\CSR(\bfW)$ has number of rows between $\rank \bfW + 1$ and $2 \rank \bfW$.
Hence, the algorithm only searches over subsets of size $\set{ \rank \bfW + 1, \dots, 2 \rank \bfW}$.

\begin{algorithm}[!htbp]
    \begin{algorithmic}[1]
        \Procedure{GetCSR-Subspace}{$\bfW = [\bfw_1; \dots; \bfw_m]$}
        \State Let $t = \rank \bfW = \dim \rowspan (\bfW)$
        \For{$k = t+1, \dots, 2t$}
            \For{Subsets $\bfU \subseteq \bfW$ of $k$ rows}
                \If{$\rank \bfU = t$ and \textsc{IsInCone}$\inparens{-\sum_{\bfu_i \in \bfU} \bfu_i, \bfU} =$ True}
                \Comment{Use \Cref{alg:is_in_cone}}
                    \State \Return $\bfV = \bfU$
                \EndIf
            \EndFor
        \EndFor
        \EndProcedure
    \end{algorithmic}
    \caption{Find $\bfV$ that attains $\ConeSubsetRank(\bfW)$ when $\calK_W = \rowspan(\bfW)$ is a linear subspace}
    \label{alg:csr__subspace}
\end{algorithm}

\begin{lemma}
\label{lemma:alg__csr__subspace}
    Let cone $\calK_W$ generated by $\bfW \in \bbR^{m \times n}$ be such that $\calK_W = \rowspan(\bfW)$.
    Then \Cref{alg:csr__subspace} finds $\bfV$ that attains $\CSR(\bfW)$ in $t^2 n^{2.5} \cdot \polylog(t,n) \cdot m^{O(t)}$ time where $t = \rank \bfW$.
\end{lemma}
\begin{proof}
    We denote the set of rows of $\bfW$ with $W$, and so $\rowspan(\bfW) = \spanv(W)$ and $\rank \bfW = \dim \spanv(W)$.
    The algorithm searches over subsets of $W$ to find $V^\ast \subseteq W$ attaining $\CSR(W)$. 
    The proof has two parts.
    First, we show that $V^\ast$ satisfies the condition in Line~5 of \Cref{alg:csr__subspace}, and $t+1 \le |V^\ast| \le 2 t$ where $t = \dim \spanv(W)$.
    Second, we show that any $U \subseteq W$ satisfying this condition has the property $\calK_U = \calK_W$.
    So by iterating over all subsets $U$ of size $\set{t+1, \dots, 2t}$ and checking whether they satisfy this condition, the algorithm finds $V$ that attains $\CSR(W)$.

    \begin{itemize}
        \item Let $V^\ast$ attain $\CSR(W)$, i.e., $V^\ast \subseteq W$, $\calK_W = \calK_{V^\ast}$, and $|V^\ast| = \CSR(W)$.
        We first show that when $\calK_W = \spanv(W)$ is a linear subspace, $V^\ast$ satisfies the condition in Line~5.
        That is, $\dim \spanv(V^\ast) = t$ and $- \sum_{\bfv_i \in V^\ast} \bfv_i \in \calK_{V^\ast}$.

        As $\calK_W$ is a $t$-dimensional linear subspace, we have $\spanv(W) = \calK_W$, which is equal to $\calK_{V^\ast}$.
        Since $V^\ast \subseteq W$, we get that $\calK_{V^\ast} = \spanv(V^\ast) = \spanv(W)$, implying that $\dim \spanv(V^\ast) = t$.
        Vectors $V^\ast$ positively span the linear subspace $\spanv(V^\ast)$, and so $-\sum_{\bfv_i \in V^\ast} \bfv_i$ is in $\spanv(V^\ast) = \calK_{V^\ast}$.
        So $V^\ast$ satisfies the condition in Line~5.
    
        Now we show that $t+1 \leq |V^\ast| \leq 2t$.
        As vectors $V^\ast$ positively span the $t$-dimensional linear subspace, \Cref{lemma:conerank__nonpointed_necessary} implies that $|V^\ast| \ge t+1$.
        Moreover, as vectors $V^\ast$ are positively independent (see \Cref{claim:csr__property__pos_independent}), Regis~\cite[Lemma~6.6]{regis2016properties} and Audet~\cite{audet2011short} state that $|V^\ast| \le 2t$ (restated in \Cref{lemma:restated__csr__subspace__max_size}).

        \item 
        We show that any $U \subseteq W$ satisfying the condition in Line~5 has the property $\calK_U = \calK_W$.
        Regis~\cite[Thm.~2.5--(v) implies (i)]{regis2016properties} states that for a given set of nonzero vectors $U$, if $-\sum_{\bfu_i \in U} \bfu_i$ is in $\calK_U$, then $\calK_U = \spanv(U)$.
        Hence, for any $U \subseteq W$ satisfying the condition in Line~5, we have $\calK_U = \spanv(U)$.
        Moreover, properties $U \subseteq W$ and $\dim \spanv(U) = t = \dim \spanv(W)$ imply that $\spanv(U) = \spanv(W)$.
        Hence, $\calK_U = \calK_W$.
    \end{itemize}
    
    The algorithm searches over all possible subsets $U \subseteq W$ with $\calK_U = \calK_W$, and the search space includes $V^\ast$ that attains $\CSR(W)$.
    The algorithm outputs $V = U$ of the smallest size with the desired properties, and so $V$ attains $\CSR(W)$.\\

    \noindent
    \textit{Runtime:} 
    In each inner loop (searching over subsets $V \subseteq W$), we can use Gaussian elimination to check if a set of $n$-dimensional vectors of size $k$ linearly spans a $t$-dimensional space---this has runtime $O (kn \min(k,n))$.
    To check if $-\sum_{\bfv_i \in V} \bfv_i \in \calK_V$, we use \Cref{alg:is_in_cone}, which checks if an LP is feasible---this has runtime $\tilde O(kn^{2.5})$ where $|V| = k \le m$.
    There are $\binom{m}{k} = O(m^k)$ subsets of size $k$.
    The total runtime is thus
    \begin{align*}
        \sum_{k=t+1}^{2t} \tilde O(m^k \cdot ( kn \min(k,n) + kn^{2.5} )) &= \sum_{k=t+1}^{2t} \tilde O( k^2 n^{2.5} m^k ) \\
        &= \tilde O( (t+1)^2 n^{2.5} \cdot m^{t+1} + \cdots + (2t)^2 n^{2.5} \cdot m^{2t})\\
        &= t^2 n^{2.5} \cdot \polylog(t,n) \cdot m^{O(t)}.
        \qedhere
    \end{align*}
\end{proof}

\subsubsection{Computing $\ConeSubsetRank$ for pointed $\calK_W$}
\label{app:algorithms__csr__pointed}

When $\calK_W$ is pointed, \Cref{alg:csr__pointed} finds $\bfV$ that attains $\CSR(\bfW)$ (proof in \Cref{lemma:alg__csr__pointed}).
This algorithm iteratively removes the vectors in the interior of the cone $\calK_W$, i.e., the rows of $\bfW$ that can be written as a nonnegative combination of other rows of $\bfW$.
Regardless of the order of removing vectors in the interior of the cone, this greedy algorithm outputs a set of rows that attains $\CSR(\bfW)$.

\begin{algorithm}[!h]
    \begin{algorithmic}[1]
        \Procedure{GetCSR-Pointed}{$\bfW = [\bfw_1; \dots; \bfw_m]$}
        \State Let $\bfU = \bfW$ and $k$ be the number of rows of $\bfU$. Initially $k = m$.
        \While{true}
            \State Denote $\bfU_{-i} = [\bfu_1; \dots; \bfu_{i-1}; \bfu_{i+1}; \dots; \bfu_k]$ for every $i \in [k]$
            \If{for some $i \in [k]$, \textsc{IsInCone}$(\bfu_i, \bfU_{-i}) =$ True}
            \Comment{Use \Cref{alg:is_in_cone}}
                \State Remove row $\bfu_i$ from the matrix $\bfU$
            \Else
                \State Break
            \EndIf
        \EndWhile
        \State \Return $\bfV = \bfU$
        \EndProcedure
    \end{algorithmic}
    \caption{Find $\bfV$ that attains $\ConeSubsetRank(\bfW)$ when $\calK_W$ is pointed}
    \label{alg:csr__pointed}
\end{algorithm}

\begin{lemma}
\label{lemma:alg__csr__pointed}
    Let cone $\calK_W$ generated by $\bfW \in \bbR^{m \times n}$ be pointed.
    Then \Cref{alg:csr__pointed} finds $\bfV$ that attains $\CSR(\bfW)$ in $\tilde O(m^3 n^{2.5})$ time.
\end{lemma}
\begin{proof}
    We first show that the output of the algorithm $V \subseteq \bbR^n$ has the properties $V \subseteq W$, $\calK_W = \calK_V$, and vectors $V$ are positively independent.
    \Cref{lemma:csr__property__pointed_positive_basis_has_same_size} shows that when $\calK_W$ is pointed, any set of vectors with these properties has the same number of vectors in it.
    Since the set of vectors that attains $\CSR(W)$ has the same properties (see \Cref{claim:csr__property__pos_independent}), we know that $|V| = \CSR(W)$.

    At initialization, $U = W$ and so $\calK_U = \calK_W$.
    In an iteration, the algorithm removes a vector $\bfu_i \in U$ if $\bfu_i \in \calK_{U_{-i}}$, which is the cone generated by the other vectors $U_{-i} = U \setminus \set{\bfu_i}$.
    Since $\bfu_i$ is a nonnegative combination of $U_{-i}$, removing row $\bfu_i$ does not change the generated cone, i.e. $\calK_U = \calK_{U_{-i}}$.
    Hence, the algorithm ensures the following invariant at the end of each iteration: $U$ has the properties $U \subseteq W$ and $\calK_U = \calK_W$.
    The algorithm terminates when no row of $U$ is a nonnegative combination of the other rows of $U$. Hence, the output of the algorithm $V$ has positively independent vectors. 

    \Cref{claim:csr__property__pos_independent} guarantees that the set of vectors attaining the $\CSR(W)$ have the same properties as $V$.
    Consequently, \Cref{lemma:csr__property__pointed_positive_basis_has_same_size} guarantees that $|V| = \CSR(W)$.\\

    \noindent
    \textit{Runtime:} 
    Since there are $m$ vectors at initialization in $U$ and the algorithm removes 1 row in each iteration until no rows can be removed, the algorithm terminates in at most $m$ iterations.
    To check if $\bfv_i \in \calK_{V_{-i}}$, we use \Cref{alg:is_in_cone}, which checks if an LP is feasible---this has runtime $\tilde O (kn^{2.5})$ where $|V| = k \le m$.
    There are $k \le m$ such checks in each iteration of the algorithm, and there are at most $m$ iterations.
    The total runtime is thus $\tilde O(m^2 \cdot mn^{2.5}) = \tilde O(m^3 n^{2.5})$.
\end{proof}

\begin{claim}
\label{claim:csr__property__pos_independent}
    Let cone $\calK_W$ generated by $W = \set{\bfw_1, \dots, \bfw_m}$, and let $V^\ast$ attain $\CSR(W)$.
    Then vectors of $V^\ast$ are positively independent.
\end{claim}
\begin{proof}
    We prove by contradiction.
    Assume that vectors of $V^\ast$ are not positively independent, i.e., there exists $\bfv_i \in V^\ast$ that can be expressed as a nonnegative combination of vectors $V^\ast \setminus \set{\bfv_i}$.
    We can remove $\bfv_i$ from $V^\ast$ to obtain $\tilde V \subseteq W$ with $|\tilde V| = |V^\ast| - 1$ and $\calK_W = \calK_{\tilde{V}}$.
    But this contradicts the assumption that $V^\ast$ attains $\CSR(W)$, i.e., $|V^\ast|$ is minimal.
    So vectors of $V^\ast$ must be positively independent.
\end{proof}

\subsection{Computing $\ConeGeneratingRank$}
\label{app:algorithms__cgr}

We give an algorithm in \Cref{app:algorithms__cgr__general} to find $\bfV$ that attains $\ConeGeneratingRank(\bfW)$.
This algorithm uses a submodule to find the $\ConeSubsetRank$ for the pointed cones.
In \Cref{app:algorithms__cgr__pointed} we show that when the cone $\calK_W$ generated by $\bfW$ is pointed, we have $\ConeGeneratingRank(\bfW) = \ConeSubsetRank(\bfW)$.

We often write $\CGR(W)$ for set $W = \set{\bfw_1, \dots, \bfw_m } \subseteq \bbR^n$ to denote $\CGR(\bfW)$ of the matrix $\bfW = [\bfw_1; \dots; \bfw_m] \in \bbR^{m \times n}$.
That is, $\CGR$ can be equivalently defined for sets as follows:
\begin{align*}
    \CGR(W) &= \min_k \set{ k \mid \calK_W = \calK_V  \text{ for some } V \subseteq \bbR^n \text{ such that } |V| = k }.
\end{align*}

We say that $V^\ast$ \textit{attains} $\CGR(W)$ if $\calK_W = \calK_{V^\ast}$, and $\abs{V^\ast} = \CGR(W)$.

\subsubsection{Computing $\ConeGeneratingRank$ for general $\calK_W$}
\label{app:algorithms__cgr__general}

\Cref{alg:cgr} finds $\bfV$ that attains $\CGR(\bfW)$ (proof in \Cref{lemma:alg__cgr}). 
The algorithm first decomposes the $\calK_W$ into its lineality space $\calL$ and a pointed cone. 
To find $\bfV$, again we can focus on rows of $\bfW$ inside $\calL$ and outside $\calL$ separately. 
We attain the $\CGR$ of rows of $\bfW$ inside $\calL$ using an orthonormal basis of $\calL$.
Then we project the rows of $\bfW$ outside $\calL$ onto $\calL^\perp$, and attain their $\CGR$ using \Cref{alg:csr__pointed}.
It turns out that this decomposition into $\calL$ and $\calL^\perp$ and generating the cone in $\calL$ and $\calL^\perp$ separately does indeed give the minimal frame (generating set) as a consequence of \Cref{lemma:cgr__property_of_decomposition}.

\begin{algorithm}[!h]
    \begin{algorithmic}[1]
        \Procedure{GetCGR}{$\bfW = [\bfw_1; \dots; \bfw_m]$}
            \State $\bfL, \bfP \gets \textsc{DecomposeCone}(W)$ and denote $\calL = \rowspan(\bfL)$, $\ell = \dim \calL$
            \Comment{Use \Cref{alg:cone_minkowski_decomposition}}
            \If{$\calL = \set{\bfzero}$}
                \Comment{$\calK_W$ is pointed}
                \State $\tilde \bfW \gets \bfW$ and $\tilde \bfZ \gets \emptyset$
            \Else
                \Comment{$\calK_W$ is non-pointed}
                \State Partition $W$ into $\WinL = \set{ \bfw_i \in W \mid \bfw_i \in \calL }$ and $\WnotinL = \set{ \bfw_i \in W \mid \bfw_i \notin \calL }$
                \State $\tilde \bfW \gets \bfWnotinL (\bfI - \bfZ \bfZ^\top)$ where columns of $\bfZ$ are orthonormal basis of $\calL$
                \State $\bfz_0 \gets - (\bfz_1 + \cdots + \bfz_\ell)$
                \State $\tilde \bfZ \gets [\bfz_0^\top; \bfZ^\top]$
                \Comment{Frame of $\calK_W \cap \calL$}
            \EndIf
            \State $\tilde \bfV \gets \textsc{GetCSR-Pointed}(\tilde \bfW)$
            \Comment{Use \Cref{alg:csr__pointed}. Frame of $\calK_W \cap \calL^\perp$}
            \State \Return $\bfV = [\tilde \bfZ; \tilde \bfV]$
        \EndProcedure
    \end{algorithmic}
    \caption{Find $\bfV$ that attains $\ConeGeneratingRank(\bfW)$}
    \label{alg:cgr}
\end{algorithm}

\begin{remark}
\label{remark:cgr__general_calls_pointed}
    When $\calK_W$ is pointed, its lineality space is $\calL = \set{\bfzero}$.
    So \Cref{alg:cgr} sets $\tilde \bfZ = \emptyset$ and $\tilde \bfW = \bfW$.
    Then the algorithm finds $\tilde \bfV$ that attains $\CSR(\tilde \bfW) = \CSR(\bfW)$ using \textsc{GetCSR-Pointed} in \Cref{alg:csr__pointed}.
    \Cref{alg:cgr} thus returns $\bfV$ that is the output of \textsc{GetCSR-Pointed}$(\bfW)$.
\end{remark}

\begin{lemma}
\label{lemma:alg__cgr}
    Given $\bfW \in \bbR^{m \times n}$, \Cref{alg:cgr} finds $\bfV$ that attains $\CGR(\bfW)$ in $\tilde O(m^3 n^{2.5})$ time.
\end{lemma}
\begin{proof}
    When $\calK_W$ is pointed, the algorithm outputs \textsc{GetCSR-Pointed}$(\bfW)$ using \Cref{alg:csr__pointed}, as noted in \Cref{remark:cgr__general_calls_pointed}.

    To get the frame of the 
    $\calK_W$ in general, we look at the union of the frame of the lineality space and the pointed part (Line~11).
    It turns out that this decomposition is minimal as proved below.
    We will also show that $\tilde{Z}$ is the frame for lineality space of $\calK_W$ (Line~9).
    \Cref{lemma:cgr__pointed} shows that $\CGR(\tilde W) = \CSR(\tilde W)$ as $\tilde W$ generates the pointed part of $\calK_W$ (Line~10).
    
    Cone $\calK_W$ has unique decomposition $\calK_W = \calL + \calK_P$ where $\calL = \lineality(\calK_W)$ is the lineality space of $\calK_W$ and $\calK_P = \calL^\perp \cap \calK_W$ is a pointed cone (see \Cref{lemma:restated__cone_decomposition}).
    \Cref{alg:cone_minkowski_decomposition} outputs matrices $\bfL$ and $\bfP$ such that $\rowspan(\bfL) = \calL$ and $\bfP$ generates the pointed cone $\calK_P$.
    After decomposing, the algorithm partitions $W$ into $\WinL$ and $\WnotinL$, and projects $\WnotinL$ onto $\calL^\perp$ to get $\tilde W = \proj_{\calL^\perp} (\WnotinL)$.
    To prove that the output $V$ has desired properties, we note the properties of $\tilde Z$ and $\tilde V$.

    \begin{itemize}
        \item \boldheading{Properties of $\tilde Z$.}
        According to \Cref{lemma:cone__lineality_space__lineal_points}, vectors $\WinL$ positively span $\calL$, which is a linear subspace.
        According to \Cref{lemma:cone__lineality_space__lineal_points}, vectors $\WinL$ positively span $\calL$, which is a linear subspace.
        As columns of $\bfZ = [\bfz_1, \dots, \bfz_\ell]$ are an orthonormal basis of $\calL$ and $\bfz_0 = -(\bfz_1 + \cdots + \bfz_\ell)$, the set $\tilde Z = \set{\bfz_0, \bfz_1, \dots, \bfz_\ell}$ generates $\calK_{\tilde Z} = \calL$ (due to Davis~\cite{davis1954theory}, see \Cref{lemma:conerank__subspace_sufficient}).
        \Cref{lemma:conerank__nonpointed_necessary} states that $\ell+1$ vectors are necessary to generate an $\ell$-dimensional linear subspace.
        Hence, $\tilde Z$ attains $\CGR(\WinL)$.
        That is, $\tilde Z$ satisfies $\calK_{\tilde Z} = \calK_{\WinL} = \calL$, and $|\tilde Z| = \CGR(\WinL)$.

        \item \boldheading{Properties of $\tilde V$.}
        According to \Cref{lemma:cone__lineality_space__nonlineal_points}, vectors $\tilde W$ generate $\calK_P$, which is a pointed cone.
        The algorithm then uses \Cref{alg:csr__pointed} to find $\tilde V$ that attains $\CSR(\tilde W)$.
        That is, $\tilde V$ satisfies $\calK_{\tilde V} = \calK_{\tilde W} = \calK_P$ and $|\tilde V| = \CSR(\tilde W)$.
        As $\calK_P$ is pointed, \Cref{lemma:cgr__pointed} implies that $\tilde V$ also attains $\CGR(\tilde W)$.
    \end{itemize}

    We now show that $\calK_V = \calK_W$.
    Below, `$+$' is the Minkowski sum.
    \begin{align}
        \calK_V &= \calK_{\tilde Z \cup \tilde V} = \calK_{\tilde Z} + \calK_{\tilde V} = \calL + \calK_P = \calK_W.
    \end{align}
    
    \Cref{lemma:cgr__property_of_decomposition} shows that any $V$ such that $\calK_V = \calK_W$ must have size $|V| \ge \CGR(\WinL) + \CGR(\tilde W)$.
    Above, we showed the output $V$ of \Cref{alg:cgr} has the properties: $\calK_V = \calK_W$ and $|V| = |\tilde Z| + |\tilde V| = \CGR(\WinL) + \CGR(\tilde W)$.
    Hence, $V$ attains $\CGR(W)$ and $\CGR(W)=\CGR(\WinL) + \CGR(\tilde W)$.\\

    \noindent
    \textit{Runtime:}
    \Cref{lemma:alg__cone_minkowski_decomposition} states that \textsc{DecomposeCone} in \Cref{alg:cone_minkowski_decomposition} has runtime $\tilde O(m^2 n^{2.5})$.
    We can check with Gaussian elimination if a row $\bfw_i$ of $\bfW$ is inside $\calL$ or not---this has runtime $O(mn \min(m,n))$ as $\bfL$ has atmost $m$ rows.
    So we can partition $W$ into $\WinL, \WnotinL$ in time $O(m^2 n \min(m,n))$.
    We find an orthonormal basis of $\calL$ in time $O(mn^2)$ and compute $\tilde \bfW$ in time $O(mn^2)$.
    Computing $\bfz_0$ takes time $O(\ell n) = O(mn)$.
    The matrix $\tilde \bfW$ has at most $m$ rows.
    Then \Cref{lemma:alg__csr__pointed} states that \textsc{GetCSR-Pointed} in \Cref{alg:csr__pointed} has runtime $\tilde O(m^3 n^{2.5})$.
    Adding all runtimes so far, the total runtime is:
    \begin{align}
        \tilde O(m^2 n^{2.5}) + O(m^2 n \min(m,n)) + O(mn^2) + O(mn^2) + O(mn) 
 + \tilde O(m^3 n^{2.5}) &= \tilde O(m^3 n^{2.5}).
        \qedhere
    \end{align}
\end{proof}

\begin{lemma}
\label{lemma:cgr__property_of_decomposition}
    Let $W = \set{\bfw_1, \dots, \bfw_m} \subseteq \bbR^n$ generate cone $\calK_W$ and have lineality space $\calL = \lineality (\calK_W)$.
    Denote $\WinL = \set{ \bfw_i \in W \mid \bfw_i \in \calL }$, and $\tilde W  = \proj_{\calL^\perp} \inparens{  W \setminus \WinL }$.
    For any $V$ satisfying $\calK_{V} = \calK_W$, we have:
    \begin{align*}
        |V| &\ge \CGR(\WinL) + \CGR(\tilde W).
    \end{align*}
\end{lemma}
\begin{proof}
    Denote $\VinL = \set{ \bfv_i \in V \mid \bfv_i \in \calL }$, and $\tilde V  = \proj_{\calL^\perp} \inparens{  V \setminus \VinL }$.
    We will show that (1) $|\VinL| \ge \CGR(\WinL)$, and (2) $|\tilde V| \ge \CGR(\tilde W)$.
    As $\tilde V$ is the projection of $V \setminus \VinL$ onto $\calL^\perp$, we have $|V \setminus \VinL| \ge |\tilde V|$.
    Joining these results, we get the desired result:
    \begin{align*}
        |V| &= |\VinL| + |V \setminus \VinL| \ge |\VinL| + |\tilde V| \ge \CGR(\WinL) + \CGR(\tilde W).
    \end{align*}

    Let $\calK_W = \calL + \calK_P$ be the unique decomposition where $\calL = \lineality(\calK_W)$ and $\calK_P = \calL^\perp \cap \calK_W$ is a pointed cone (see \Cref{lemma:restated__cone_decomposition}).
    To prove that $|\VinL| \ge \CGR(\WinL)$, we show that $\VinL$ satisfies  $\calK_{\VinL} = \calK_{\WinL}$, and so the first statement follows from the definition of $\CGR$.
    We analogously prove that $|\tilde V| \ge \CGR(\tilde W)$.
    \begin{enumerate}
        \item
        Since $\calL$ is the lineality space of $\calK_{V} = \calK_W$, \Cref{lemma:cone__lineality_space__lineal_points} states that $\calK_{\WinL} = \calL$ and $\calK_{\VinL} = \calL$, and so $\calK_{\VinL} = \calK_{\WinL}$.
        By definition of $\CGR$, we have $|\VinL| \ge \CGR(\WinL)$.

        \item
        As $\calK_P$ is the pointed cone from decomposing $\calK_W = \calK_V$, \Cref{lemma:cone__lineality_space__nonlineal_points} states that $\cone(\tilde W) = \calK_P$ and $\cone(\tilde V) = \calK_P$, and so $\cone(\tilde W) = \cone(\tilde V)$.
        By definition of $\CGR$, we have $|\tilde V| \ge \CGR(\tilde W)$.
        \qedhere
    \end{enumerate}
\end{proof}

\subsubsection{Computing $\ConeGeneratingRank$ for pointed $\calK_W$}
\label{app:algorithms__cgr__pointed}

When $\calK_W$ is pointed, \Cref{lemma:cgr__pointed}, proved in \Cref{lemma_w_proof:cgr__pointed}, states that $\CGR(\bfW) = \CSR(\bfW)$.
Hence, we can use \textsc{GetCSR-Pointed} in \Cref{alg:csr__pointed} to find $\bfV$ that attains $\CGR(\bfW) = \CSR(\bfW)$.

\begin{lemma}[\Cref{lemma:cgr__pointed}]
\label{lemma_w_proof:cgr__pointed}
    Let $\calK_W$ be a pointed cone.
    Then, $\CGR(\bfW) = \CSR(\bfW)$.
\end{lemma}
\begin{proof}
This result follows from standard properties of extreme rays of pointed cones~\cite{border2022convex,nemirovski2023linear}, namely \textit{all extreme rays of a pointed cone come from its generating set}.

In Border~\cite{border2022convex}, an \textit{extreme ray} of a convex cone is defined as a vector in the cone (unique up to positive scaling) such that it cannot be written as nonnegative linear combination of other vectors in the cone. 
A pointed cone can be generated by different sets of generators. 
Border~\cite[Prop.~26.5.4]{border2022convex}, restated in \Cref{lemma:restated__cone_pointed_extreme_rays_are_subset}, shows that the set of extreme rays of the pointed cone are always included in any of the generating sets of the cone.
As a result, the number of extreme rays $k$ of $\calK_W$ is such that $k \le \CGR(W)$. 

According to the definition of $\CSR$, set $V$ attains $\CSR(W)$ if it is the smallest subset of $W$ that generates $\calK_W$. 
Also, since $W$ generates $\calK_W$, the set of extreme rays of the pointed cone are always included in any of the generating sets of the cone.
That is, the set of extreme rays of $\calK_W$ is a subset of $W$. 
Border~\cite[Prop.~26.5.4]{border2022convex} also states that the set of extreme rays generates the cone $\calK_W$. 
Hence, $\CSR(W) \le k$.

These two statements imply $\CSR(W) \le \CGR(W)$.
\Cref{prop:matrix_rank_relationships} states that $\CSR(W) \ge \CGR(W)$, implying that the set of extreme rays attain $\CSR(W)$.
\end{proof}

\subsection{Computing $\ConeRank$}
\label{app:algorithms__cr}

We give an algorithm in \Cref{app:algorithms__cr__general} to find $\bfV$ that attains $\ConeRank(\bfW)$.
This algorithm uses a submodule to find the $\ConeRank$ for the pointed cones, described in \Cref{app:algorithms__cr__pointed}.

We often write $\CR(W)$ for set $W = \set{\bfw_1, \dots, \bfw_m } \subseteq \bbR^n$ to denote $\CR(\bfW)$ of the matrix $\bfW = [\bfw_1; \dots; \bfw_m] \in \bbR^{m \times n}$.
That is, $\CR$ can be equivalently defined for sets as follows:
\begin{align*}
    \CR(W) &= \min_k \set{ k \mid \calK_W \subseteq \calK_V  \text{ for some } V \subseteq \bbR^n \text{ such that } |V| = k }.
\end{align*}

We say that $V^\ast$ \textit{attains} $\CR(W)$ if $\calK_W \subseteq \calK_{V^\ast}$, and $\abs{V^\ast} = \CR(W)$.

\subsubsection{Computing $\ConeRank$ for general $\calK_W$}
\label{app:algorithms__cr__general}

\Cref{alg:cr} finds $\bfV$ that attains $\CR(\bfW)$ (proof in \Cref{lemma:alg__cr}). 
The algorithm first decomposes the $\calK_W$ into its lineality space $\calL$ and a pointed cone. 
To find $\bfV$, again we can focus on rows of $\bfW$ inside $\calL$ and outside $\calL$ separately.
We attain the $\CR$ of rows of $\bfW$ inside $\calL$ using an orthonormal basis of $\calL$.
Then we project the rows of $\bfW$ outside $\calL$ onto $\calL^\perp$, and attain their $\CR$ using \Cref{alg:cr__pointed}.

\begin{algorithm}[!h]
    \begin{algorithmic}[1]
        \Procedure{GetCR}{$\bfW = [\bfw_1; \dots; \bfw_m]$}
            \State $\bfL, \bfP \gets \textsc{DecomposeCone}(W)$ and denote $\calL = \rowspan(\bfL)$, $\ell = \dim \calL$
            \Comment{Use \Cref{alg:cone_minkowski_decomposition}}
            \If{$\calL = \set{\bfzero}$}
                \Comment{$\calK_W$ is pointed}
                \State $\tilde \bfW \gets \bfW$ and $\tilde \bfZ \gets \emptyset$
            \Else
                \Comment{$\calK_W$ is non-pointed}
                \State Partition $W$ into $\WinL = \set{ \bfw_i \in W \mid \bfw_i \in \calL }$ and $\WnotinL = \set{ \bfw_i \in W \mid \bfw_i \notin \calL }$
                \State $\tilde \bfW \gets \bfWnotinL (\bfI - \bfZ \bfZ^\top)$ where columns of $\bfZ$ are orthonormal basis of $\calL$
                \State $\bfz_0 \gets - (\bfz_1 + \cdots + \bfz_\ell)$
                \State $\tilde \bfZ \gets [\bfz_0^\top; \bfZ^\top]$
                \Comment{Frame of $\calK_W \cap \calL$}
            \EndIf
            \State $\tilde \bfV \gets \textsc{GetCR-Pointed}(\tilde \bfW)$
            \Comment{Use \Cref{alg:cr__pointed}. Encloses $\calK_W \cap \calL^\perp$}
            \State \Return $\bfV = [\tilde \bfZ; \tilde \bfV]$
        \EndProcedure
    \end{algorithmic}
    \caption{Find $\bfV$ that attains $\ConeRank(\bfW)$}
    \label{alg:cr}
\end{algorithm}

\begin{remark}
\label{remark:cr__general_calls_pointed}
    When $\calK_W$ is pointed, its lineality space is $\calL = \set{\bfzero}$.
    So \Cref{alg:cr} sets $\tilde \bfZ = \emptyset$ and $\tilde \bfW = \bfW$.
    Then the algorithm finds $\tilde \bfV$ that attains $\CR(\tilde \bfW) = \CR(\bfW)$ using \textsc{GetCR-Pointed} in \Cref{alg:cr__pointed}.
    \Cref{alg:cr} thus returns $\bfV$ that is the output of \textsc{GetCR-Pointed}$(\bfW)$.
\end{remark}

\begin{lemma}
\label{lemma:alg__cr}
    Given $\bfW \in \bbR^{m \times n}$, \Cref{alg:cr} finds $\bfV$ that attains $\CR(\bfW)$ in $\tilde O(m^3 n^3)$ time.
\end{lemma}
\begin{proof}
    When $\calK_W$ is pointed, the algorithm outputs \textsc{GetCR-Pointed}$(\bfW)$ using \Cref{alg:cr__pointed}, as noted in \Cref{remark:cr__general_calls_pointed}.
    We give the proof for the pointed case in \Cref{lemma:alg__cr__pointed}.

    To find a cone enclosing $\calK_W$ in general, we look a the union of the cones enclosing the lineality space and the pointed part (Line~11).
    It turns out that this decomposition is minimal as proved below.
    We will show that $\tilde Z$ encloses the lineality space of $\calK_W$ (Line~9), and $\tilde V$ encloses the pointed part of $\calK_W$ (Line~10).
    
    Let $r = \dim \calK_W = \rank \bfW$.
    Cone $\calK_W$ has unique decomposition $\calK_W = \calL + \calK_P$ where $\calL = \lineality(\calK_W)$ is the lineality space of $\calK_W$ and $\calK_P = \calL^\perp \cap \calK_W$ is a pointed cone (see \Cref{lemma:restated__cone_decomposition}).
    \Cref{alg:cone_minkowski_decomposition} outputs matrices $\bfL$ and $\bfP$ such that $\rowspan(\bfL) = \calL$ and $\bfP$ generates the pointed cone $\calK_P$.
    After decomposing, the algorithm partitions $W$ into $\WinL$ and $\WnotinL$, and projects $\WnotinL$ onto $\calL^\perp$ to get $\tilde W = \proj_{\calL^\perp} (\WnotinL)$.
    To prove that the output $V$ has desired properties, we note the properties of $\tilde Z$ and $\tilde V$.

    \begin{itemize}
        \item \boldheading{Properties of $\tilde Z$.}
        According to \Cref{lemma:cone__lineality_space__lineal_points}, vectors $\WinL$ positively span $\calL$, which is a linear subspace.
        As columns of $\bfZ = [\bfz_1, \dots, \bfz_\ell]$ are an orthonormal basis of $\calL$ and $\bfz_0 = -(\bfz_1 + \cdots + \bfz_\ell)$, the set $\tilde Z = \set{\bfz_0, \bfz_1, \dots, \bfz_\ell}$ generates $\calK_{\tilde Z} = \calL$ (due to Davis~\cite{davis1954theory}, see \Cref{lemma:conerank__subspace_sufficient}).
        \Cref{lemma:conerank__nonpointed_necessary} states that $\ell+1$ vectors are necessary to positively span an $\ell$-dimensional linear subspace.
        Hence, $\tilde Z$ attains $\CR(\WinL)$.
        That is, $\tilde Z$ satisfies $\calK_{\tilde Z} = \calK_{\WinL} = \calL$, and $|\tilde Z| = \CR(\WinL) = \ell+1$.
    
        \item \boldheading{Properties of $\tilde V$.}
        According to \Cref{lemma:cone__lineality_space__nonlineal_points}, vectors $\tilde W$ generate $\calK_P$, which is a pointed cone.
        Note that $\calK_{\tilde W} = \calK_P$ is an $(r-\ell)$-dimensional pointed cone, since $\calL$ is an $\ell$-dimensional linear subspace inside $r$-dimensional $\calK_W$ and $\calK_P = \calL^\perp \cap \calK_W$.
        The algorithm uses \Cref{alg:cr__pointed} to find $\tilde V$ that attains $\CR(\tilde W)$.
        That is, $\calK_{\tilde W} \subseteq \calK_{\tilde V}$ and $|\tilde V| = \CR(\tilde W)$.
        According to \Cref{lemma:alg__cr__pointed}, we have $\CR(\tilde W) = r-\ell$ and so $|\tilde V| = r-\ell$.
    \end{itemize}

    We now show that $\calK_W \subseteq \calK_V$.
    Below, `$+$' is the Minkowski sum.
    \begin{align}
        \calK_W &= \calL + \calK_P = \calK_{\tilde Z} + \calK_P \subseteq \calK_{\tilde Z} + \calK_{\tilde V} = \calK_{\tilde Z \cup \tilde V} = \calK_V.
    \end{align}

    \Cref{lemma:conerank__nonpointed_necessary} states that any $V$ such that $\calK_W \subseteq \calK_V$ must have size $|V| \ge r+1$ where $r = \dim \calK_W = \rank \bfW$.
    Above, we show that the output $V$ of \Cref{alg:cr} has the properties:
    $\calK_W \subseteq \calK_V$ and $|V| = |\tilde Z| + |\tilde V| = (\ell+1) + (r-\ell) = r+1$.
    Hence, $V$ attains $\CR(W)$ and $\CR(W) = r+1$.\\

    \noindent
    \textit{Runtime:}
    \Cref{lemma:alg__cone_minkowski_decomposition} states that \textsc{DecomposeCone} in \Cref{alg:cone_minkowski_decomposition} has runtime $\tilde O(m^2 n^{2.5})$.
    We can check with Gaussian elimination if a row $\bfw_i$ of $\bfW$ is inside $\calL$ or not---this has runtime $O(mn \min(m,n))$ as $\bfL$ has atmost $m$ rows.
    So we can partition $W$ into $\WinL, \WnotinL$ in time $O(m^2 n \min(m,n))$.
    We find an orthonormal basis of $\calL$ in time $O(mn^2)$ and compute $\tilde \bfW$ in time $O(mn^2)$.
    Computing $\bfz_0$ takes time $O(\ell n) = O(mn)$.
    The matrix $\tilde \bfW$ has at most $m$ rows.
    Then \Cref{lemma:alg__cr__pointed} states that \textsc{GetCR-Pointed} in \Cref{alg:cr__pointed} has runtime $O(m^2 n^3)$.
    Adding all runtimes so far, the total runtime is:
    \begin{align}
        \tilde O(m^2 n^{2.5}) + O(m^2 n \min(m,n)) + O(mn^2) + O(mn^2) + O(mn) 
 + O(m^2 n^3) &= \tilde O(m^3 n^3).
        \qedhere
    \end{align}
\end{proof}

\subsubsection{Computing $\ConeRank$ for pointed $\calK_W$}
\label{app:algorithms__cr__pointed}

When $\calK_W$ is pointed, \Cref{alg:cr__pointed} finds $\bfV$ that attains $\CR(\bfW)$ when $\rank \bfW = r$ (proof in \Cref{lemma:alg__cr__pointed}).
This algorithm first finds an $(r-1)$-dimensional hyperplane that strictly separates the origin from the convex hull of rows of $\bfW$.
Then the algorithm scales rows of $\bfW$ to lie on the hyperplane, and finds an $(r-1)$-simplex that encloses the convex hull of the scaled rows.
It turns out that the cone generated by the $r$ vertices of the simplex enclose $\calK_W$.
We discuss how to find the separating hyperplane in \Cref{remark:cr__pointed_separating_hyperplane}, and the simplex in \Cref{remark:cr__pointed_simplex_alg}.

\begin{algorithm}[!h]
    \begin{algorithmic}[1]
        \Procedure{GetCR-Pointed}{$\bfW = [\bfw_1; \dots; \bfw_m]$}
            \State Let $r = \rank \bfW$
            \State Find $(r-1)$-dimensional hyperplane $\bfw^\ast \cdot \bfx = b$ with $b > 0$ such that $\bfw^\ast \cdot \bfw_i = b_i > b$ for all $i \in [m]$
            \State Scale each $\bfw_i$ to get $\bfu_i = \frac{b}{b_i} \bfw_i$ lying on the hyperplane
            \State Find $(r-1)$-simplex in the hyperplane (with vertices $V = \set{\bfv_1, \dots, \bfv_r}$) so that $\conv(U) \subseteq \conv(V)$
            \State \Return $\bfV$
        \EndProcedure
    \end{algorithmic}
    \caption{Find $\bfV$ that attains $\ConeRank(\bfW)$ when $\calK_W$ is pointed}
    \label{alg:cr__pointed}
\end{algorithm}

\begin{remark}[Separating hyperplane]
\label{remark:cr__pointed_separating_hyperplane}
    \Cref{alg:cr__pointed} relies on finding an $(r-1)$-dimensional hyperplane that strictly separates $\conv(W)$ from the origin.
    We can find such a hyperplane using a Support Vector Machine (SVM).
    W.l.o.g. let vectors $\bfw_1, \dots, \bfw_m$ be nonzero.
    To do so, we assign $y=+1$ label to vectors $\bfw_1, \dots, \bfw_m$ and $y = -1$ label to vector $\bfzero$.
    In the proof of \Cref{lemma:alg__cr__pointed}, we argue that a strictly separating hyperplane exists.
    Hence, the two classes $(W, +1)$ and $(\set{\bfzero}, -1)$ are linearly separable by some hyperplane $\bfw^\ast \cdot \bfx = b$ with unit $\ell_2$-norm $\bfw^\ast$ and $b > 0$.
    The SVM algorithm thus optimizes a quadratic objective given linear constraints~\cite[Lec.~9]{weinberger2018ml}, and can be optimized with Interior-Point Methods in $O(m^2 n^3)$ time~\cite[Sec.~10.2]{nemirovski2023ipm}.
\end{remark}
\begin{remark}[Numerical stability]
    We comment on the numerical stability of finding the separating hyperplane in Line~3.
    While a strictly separating hyperplane always exists as $\calK_W$ is pointed, the parameters determining this hyperplane can be numerically unstable.
    In particular, two issues could cause numerical instability:
    \begin{enumerate}
        \item 
        When there are vectors in $W$ that are too close to the origin, the parameter $b$ of the hyperplane is very small.
        Scaling vectors $\bfw_i$ by $b/b_i$ to get $\bfu_i$ in Line~4 could lead to numerical instability in finding a simplex to the vectors $\bfu_i$ (Line~5).
        This issue can be improved by a preprocessing step to scale all vectors $W$ to be unit norm.
        This rescaling does not change the generated cone. 

        \item 
        When the pointed cone is close to being nonpointed, the numerical instability becomes unavoidable.
        For instance, in \Cref{example:matrix_rank__2d}~(c), $\calK_W$ is a pointed cone.
        With a larger angle between the extreme rays, the cone is still pointed as long as angle $<\pi$.
        But with a larger angle between the extreme rays, the cone 
        But with larger angle between the extreme rays, the cone is still pointed (if the angle $<\pi$), but is getting close to a half space which is a nonpointed cone. Theoretically, as long as the angle $<\pi$, the $\CR$ is two. But finding the two points which enclose the cone requires higher numerical precision. This issue is unavoidable and is due to the \textit{almost} nondegeneracy of the cones. 
    \end{enumerate}
\end{remark}

\begin{remark}[Enclosing Simplex]
\label{remark:cr__pointed_simplex_alg}
    \Cref{alg:cr__pointed} relies on finding an $(r-1)$-simplex in the $(r-1)$-dimensional hyperplane so that the simplex encloses the convex hull of scaled vectors $U = \set{\bfu_1, \dots, \bfu_m}$.
    We briefly sketch the procedure for finding the vertices of such a simplex.
    A crude approach is to circumscribe $\conv(U)$ with a sphere and find a simplex that inscribes this sphere.
    So we circumscribe $\conv(U)$ with a $(r-1)$-dimensional sphere, which lies in $\affine(U)$, centered at $\overline{\bfu} = \frac{1}{m} \sum_{i \in [m]} \bfu_i$ of radius $R = \max_{i \in [m]} \norm{\bfu_i - \overline{\bfu}}_2$.
    
    We now find a simplex centered at $\overline{\bfu}$ with inradius equal to $R$ to enclose this sphere.
    A regular $(r-1)$-simplex of side length $a$ centered has inradius $\nicefrac{a}{\sqrt{2 (r-1) r}}$~\cite{nevskii2019some}.
    \footnote{\url{https://math.stackexchange.com/questions/2222844/inradius-of-regular-simplex-in-mathbbrn}}
    Conversely, a regular $(r-1)$-simplex of inradius $R$ has side length $R \sqrt{2(r-1)r}$.
    So, we can construct the simplex in $\bbR^{r-1}$ with $r$ vertices $\set{ \bfzero, (R \sqrt{2(r-1)r} \cdot \bfe_1), \dots, (R \sqrt{2(r-1)r} \cdot \bfe_{r-1}) }$ where $\bfe_i$ is the $i^{th}$ canonical basis vector of $\bbR^{r-1}$.
    And then we apply an affine transformation so that this simplex is centered at $\overline{\bfu}$ and lies in the $(r-1)$-dimensional hyperplane of $\bbR^n$ containing $\conv(U)$.
    The transformed simplex encloses the sphere containing $\conv(U)$.
    
    Computing the desired center and inradius take $O(mn)$ time.
    For each vertex of the simplex, applying the affine transformation (mapping $\bbR^{r-1}$ to $\bbR^n$) takes $O(nr)$ time.
    Thus the total runtime is $O(mnr)$.
\end{remark}

\begin{lemma}
\label{lemma:alg__cr__pointed}
    Let cone $\calK_W$ generated by $\bfW \in \bbR^{m \times n}$ be pointed.
    Then \Cref{alg:cr__pointed} finds $\bfV$ that attains $\CR(\bfW)$ in $O(m^2 n^3)$ time.
\end{lemma}
\begin{proof}
    Let $r = \rank \bfW = \dim \spanv(W)$.

    Consider the convex hull of $W$, the set $\calC_W = \conv(\set{\bfw_1, \dots, \bfw_m})$. 
    W.l.o.g. we remove any zero vectors from $W$.
    We claim that $\bfzero \notin \calC_W$.
    To see this, we give a brief proof by contradiction. 
    Assume that $\bfzero \in \calC_W$.
    Then there exists some nonzero $\alpha_1, \dots, \alpha_m \ge 0$ with $\sum_i |\alpha_i| = 1$ such that $\bfzero = \sum_{i=1}^m \alpha_i \bfw_i$.
    The origin can thus be written as a nonzero, nonnegative combination of vectors $W$.
    This contradicts the fact that $\calK_W$ is pointed, as stated in \Cref{lemma:nonpointed__zero_in_cone}.

    The convex hull $\calC_W$ lies in an $r$-dimensional linear subspace.
    This is because $\dim \spanv(\calC_W) = \dim \spanv(\calK_W) = r$, where $\spanv$ is the linear span.
    As $\bfzero$ is a point in $\spanv(\calC_W)$ but not in $\calC_W$, \Cref{lemma:separating_hyperplane_theorem_lower_dim} tells us that $\bfzero$ and $\calC_W$ are strictly separated by an $(r-1)$-dimensional hyperplane $\bfw^\ast \cdot \bfx = b$.
    W.l.o.g. let $\bfw^\ast$ have unit $\ell_2$-norm and $b > 0$, i.e., origin lies on the negative side of the hyperplane.
    As each $\bfw_i$ lies on the positive side of the hyperplane, we have $\bfw^\ast \cdot \bfw_i = b_i \ge b > 0$.
    We scale each $\bfw_i$ onto this hyperplane to get $\bfu_i = \frac{b}{b_i} \bfw_i$.
    Note that $\bfu_i$ lie on the hyperplane as $\bfw^\ast \cdot \bfu_i = \frac{b}{b_i} \bfw^\ast \cdot \bfw_i = b$.
    Each $\bfu_i$ is thus a positive scaling of $\bfw_i$.

    Let $\calC_U = \conv(\set{\bfu_1, \dots, \bfu_m})$ be the convex hull of the scaled points.
    This convex set is $(r-1)$-dimensional as its affine hull is the hyperplane $\bfw^\ast \cdot \bfx = b$ itself.
    Gale~\cite{gale1953inscribing} tells us that this $(r-1)$-dimensional bounded set $\calC_U$ can be enclosed within an $(r-1)$-simplex having $(r-1)+1 = r$ vertices (see  \Cref{lemma:restated__conerank__pointed__convex_simplex_cover} and \Cref{remark:cr__pointed_simplex_alg}).
    Let the vertices be $V = \set{\bfv_1, \dots, \bfv_r}$.
    Each $\bfv_i$ thus lies in the hyperplane $\bfw^\ast \cdot \bfx = b$ and we have $\calC_U \subseteq \calC_V$.
    
    We finally show that $\calK_W \subseteq \calK_V$.
    Let $\bfx \in \calK_W$, i.e. $\bfx = \sum_{i=1}^m \lambda_i \bfw_i$ for nonzero $\lambda_1, \dots, \lambda_m \ge 0$.
    Because $\calC_U \subseteq \calC_V$, we have that for all $i \in [m]$, the scaled point $\bfu_i = \sum_{j=1}^r \alpha_{i,j} \bfv_i$ for some nonzero $\alpha_{i,1}, \dots, \alpha_{i,r} \ge 0$ with $\sum_{j=1}^r \alpha_{i,j} = 1$.
    We can rewrite $\bfx$ as:
    \begin{align}
        \bfx &= \sum_{i=1}^m \lambda_i \bfw_i = \sum_{i=1}^m \lambda_i \cdot \frac{b_i}{b} \cdot \bfu_i
        = \sum_{i=1}^m \lambda_i \cdot \frac{b_i}{b} \cdot \sum_{j=1}^r \alpha_{i,j} \bfv_j
        = \sum_{j=1}^r \underbrace{ \inparens{ \sum_{i=1}^m \lambda_i \cdot \frac{b_i}{b} \cdot \alpha_{i,j} } }_{\mu_j} \bfv_j
    \end{align}
    where coefficients $\mu_j \ge 0$ as $\lambda_i, \alpha_{i,j}, b_i, b \ge 0$ for all $i \in [m], j \in [r]$.
    So $\bfx \in \calK_V$, implying that $\calK_W \subseteq \calK_V$.

    Since $\calK_W$ is an $r$-dimensional pointed cone, \Cref{lemma:conerank__pointed_necessary} states that $r$ vectors are necessary to generate a cone enclosing $\calK_W$, i.e., $\CR(\bfW) \ge r$.
    In fact, the output of the algorithm $V$ has exactly $r$ vectors, so it must attain $\CR(W)$.\\

    \noindent
    \textit{Runtime:}
    We can compute $\rank \bfW$ using Gaussian elimination in time $O(mn \min(m,n))$.
    \Cref{remark:cr__pointed_separating_hyperplane} notes that we can find the separating hyperplane using Support Vector Machines in $O(m^2 n^3)$ time.
    After scaling each $\bfw_i$ to get $\bfu_i$, we can find the simplex enclosing $\calC_U$ in time $O(mnr)$ as noted in \Cref{remark:cr__pointed_simplex_alg}.
    Thus the total runtime is $O(mn \min(m,n) + m^2 n^3 + mnr) = O(m^2n^3)$.
\end{proof}

%% file: app_preliminaries.tex
\clearpage
\section{Preliminaries}
\label{app:preliminaries}

\subsection{Background on convex analysis}
\label{app:background_convex_analysis}

A convex set $\calK \subseteq \bbR^n$ is a cone if for all $\bfx \in \calK, \lambda \ge 0$, we have $\lambda \bfx \in \calK$.
The Minkowski-Weyl theorem states that a cone is polyhedral if and only if it is finitely generated.
The cone generated by vectors in $W \subseteq \bbR^n$ is denoted as $\calK_W$.
Representing the elements of a cone as a conic combination of the generating vectors, as in $\calK_W$, is called the Vertex-representation (V-representation).
Cones can equivalently be defined in terms of inequalities, $\set{\bfx \in \bbR^n \mid \bfA \bfx \ge \bfzero}$ for some $\bfA \in \bbR^{m \times n}$.
This representation is referred to as the Halfspace-representation (H-representation).

The dimension of a convex set $\calC$ is defined as the dimension of the affine hull of $\calC$, i.e., $\dim \calC = \dim \affine(\calC)$ where the affine hull is simply the set of all affine combinations of elements of a set.
The affine hull of a cone $\calK$ is the linear subspace spanned by the cone, and so $\dim \calK = \dim \spanv(\calK)$.
It then follows that, for cone $\calK_W$ generated by rows of $\bfW$, we have $\dim \calK_W = \dim \spanv(\calK_W) = \rank \bfW$.

The dual cone of $\calK$ is $\calK^\ast = \set{ \bfy \in \bbR^n \mid \innerproduct{\bfy}{\bfx} \ge 0 \text{ for all } \bfx \in \calK }$.
Thus, for cone $\calK_W$ generated by rows of $\bfW = [\bfw_1; \dots; \bfw_m]$, the dual cone is:
\begin{align*}
    \calK_W^\ast = \set{\bfy \in \bbR^n \mid \text{for all } \bfx \in \calK_W, \innerproduct{\bfx}{\bfy} \ge 0} = \set{\bfy \in \bbR^n \mid \bfW \bfy \ge \bfzero}
\end{align*}
where the last equality is true due to \Cref{lemma:cones__dual_cone_eq_aux_cone}.

\subsection{Auxiliary Lemmas}

\begin{lemma}
\label{lemma:cones__dual_cone_eq_aux_cone}
    Let rows of $\bfW \in \bbR^{m \times n}$ generate cone $\calK_W$.
    Then $\calK_W^\ast = \set{ \bfy \in \bbR^n \mid \bfW \bfy \ge \bfzero }$.
\end{lemma}
\begin{proof}
    Denote $\calK_W' = \set{ \bfy \in \bbR^n \mid \bfW \bfy \ge \bfzero }$.
    First, we prove that $\calK_W^\ast \subseteq \calK_W'$.
    Let $\bfy \in \calK_W^\ast$. 
    Then for all $\bfx \in \calK_W, \innerproduct{\bfx}{\bfy} \ge 0$. 
    Since for all $i \in [m], \bfw_i \in \calK_W$, we have that $\innerproduct{\bfw_i}{\bfy} \ge 0$.
    Therefore, $\bfy \in \calK_W'$.
    Finally, we prove that $\calK_W' \subseteq \calK_W^\ast$.
    Let $\bfy \in \calK_W'$, so $\bfW \bfy \ge \bfzero$.
    Let $\bfx \in \calK_W$, i.e. $\bfx = \bflambda \bfW$ for some nonnegative $\bflambda$.
    Therefore, we have $\innerproduct{\bfx}{\bfy} = \innerproduct{\bflambda \bfW}{\bfy} = \innerproduct{\bflambda}{\bfW \bfy} \ge 0$.
    This implies that $\bfy \in \calK_Z^\ast$.
\end{proof}

\begin{lemma}
\label{lemma:nonpointed__zero_in_cone}
    Let rows of $\bfW \in \bbR^{m \times n}$ generate cone $\calK_W$.
    Cone $\calK_W$ is nonpointed if and only if there exists nonzero $\bflambda \ge \bfzero$ such that $\bfzero = \bflambda \bfW$.
\end{lemma}
\begin{proof}
    We first prove that if $\calK_W$ is nonpointed then there exists nonzero $\bflambda \ge \bfzero$ such that $\bfzero = \bflambda \bfW$.
    Cone $\calK_W$ is nonpointed when there exists nonzero $\bfx \in \calK_W$ such that $-\bfx \in \calK_W$.
    That is, there exists nonzero and nonnegative $\bfalpha, \bfbeta$ such that $\bfx = \bfalpha \bfW$ and $-\bfx = \bfbeta \bfW$.
    Adding the two equations, we get $\bfzero = (\bfalpha + \bfbeta) \bfW$.
    Hence, there exists nonzero $\bflambda = (\bfalpha + \bfbeta) \ge \bfzero$ such that $\bfzero = \bflambda \bfW$.

    We now prove that if there exists nonzero $\bflambda \ge \bfzero$ such that $\bfzero = \bflambda \bfW$ then cone $\calK_W$ is nonpointed.
    W.l.o.g. assume $\bfW$ does not contain any rows of all zeros and $\lambda_1 > 0$.
    Let $\bfx = \lambda_1 \bfw_1$.
    Since $\bfzero = \bflambda \bfW$, we have $\bfx = \lambda_1 \bfw_1 = - \sum_{i=2}^m \lambda_i \bfw_i$.
    Therefore, there exists nonzero $\bfx = \lambda_1 \bfw_1 \in \calK_W$ such that $-\bfx = \sum_{i=2}^m \lambda_i \bfw_i \in \calK_W$.
    So cone $\calK_W$ is nonpointed.
\end{proof}

\begin{lemma}
\label{lemma:polyhedral_cones__duals__inclusion}
    For two polyhedral cones $\calK_1$ and $\calK_2$, we have $\calK_1 \subseteq \calK_2 \iff \calK_2^\ast \subseteq \calK_1^\ast$.
\end{lemma}
\begin{proof}
    Since the two cones are polyhedral, they are closed and convex.
    For any closed and convex cone $\calK$, the dual of its dual cone is the cone itself: $\calK^{\ast\ast} = \calK$.
    The result then follows from the fact that for any two convex cones $\calK_1 \subseteq \calK_2 \implies \calK_2^\ast \subseteq \calK_1^\ast$~\cite[Sec.~2.6.1]{boyd2004convex}.
\end{proof}

\begin{lemma}
\label{lemma:affine_subspace_invariant_origin_choice}
    Let $\affine(X)$ be the affine hull of set $X \subseteq \bbR^n$.
    Then for any $\bfx, \bfx' \in X, \affine(X)_{\bfx} = \affine(X)_{\bfx'}$.
\end{lemma}
\begin{proof}
    We will prove that $\affine(X)_{\bfx} \subseteq \affine(X)_{\bfx'}$.
    Inclusion in the other direction will follow from a symmetry argument.
    We first note that for any $\bfx \in X$, the set $\affine(X)_{\bfx}$ is a linear subspace.
    Pick arbitrary $\bfx, \bfx' \in X$, and let $\bfy \in \affine(X)_{\bfx}$.
    Then $\bfy + \bfx \in \affine(X)$.
    By centering $\affine(X)$ at $\bfx'$ and rearranging terms, we get that $\bfy \in (\bfx' - \bfx) + \affine(X)_{\bfx'}$.
    By definition, we have that $\bfx - \bfx' \in \affine(X)_{\bfx'}$.
    Since $\affine(X)_{\bfx'}$ is a linear subspace, the vector $-(\bfx - \bfx')$ is also in $\affine(X)_{\bfx'}$.
    Hence, $\bfy \in \affine(X)_{\bfx'}$.
\end{proof}

\begin{lemma}[{Border~\cite[Prop.~26.5.4]{border2022convex}}]
\label{lemma:restated__cone_pointed_extreme_rays_are_subset}
    Let $W = \set{\bfw_1, \dots, \bfw_m} \subseteq \bbR^n$ be nonzero vectors that generate cone $\calK_W$.
    If $\calK_W$ is pointed, then it has nondegenerate extreme rays, and each is of the form $\langle \bfw_i \rangle$ for some $i \in [m]$.
    That is, every extreme ray is one of the generators $\bfw_i$.
    (But not every every $\bfw_i$ need be extreme.)
    Moreover, the cone $\calK_W$ is the convex hull of its extreme rays.
\end{lemma}

\begin{lemma}[{Regis~\cite[Lemma~6.6]{regis2016properties}, Audet~\cite{audet2011short}}]
\label{lemma:restated__csr__subspace__max_size}
    Let $S$ be a finite set of vectors that positively span the linear subspace $V$ of $\bbR^n$.
    Then $S$ contains a subset that positively spans $V$ and that contains at most $2 \dim V$ elements.
\end{lemma}

\begin{lemma}
\label{lemma:separating_hyperplane_theorem_lower_dim}
    Let $\calC$ be a convex set lying in an $r$-dimensional linear subspace in $\bbR^n$, i.e. $\dim \spanv(\calC) = r$, and let $\bfx_0$ be a point in $\spanv(\calC)$ but not in $\calC$.
    Then there exists an $(r-1)$-dimensional affine hyperplane $\bfw^\star \cdot \bfx = b$ strictly separating $\calC$ and $\bfx_0$.
    That is, $\bfw^\star \cdot \bfy > b$ for all $\bfy \in \calC$ and $\bfw^\star \cdot \bfx_0 < b$.
\end{lemma}
\begin{proof}
    Since $\bfx_0 \notin \calC$, by the separating hyperplane theorem~\cite[Example~2.20]{boyd2004convex}, we get that $\bfx_0$ and $\calC$ are strictly separated by a hyperplane $\bfw \cdot \bfx = b$.
    This means that $\bfw \cdot \bfy > b$ for all $\bfy \in \calC$ and $\bfw \cdot \bfx_0 < b$.
    This affine hyperplane is $(n-1)$-dimensional.
    
    We will now show that an $(r-1)$-dimensional affine hyperplane exists that separates $\bfx_0$ and $\calC$.
    Denote by $H$ the affine hyperplane $\bfw \cdot \bfx = b$ and $A = \spanv(\calC)$ the linear subspace in which $\calC$ lies.
    Note that the hyperplane $H$ neither is parallel to $A$ nor includes $A$.
    Why?
    If $H$ were parallel to $A$, then $\bfw \cdot \bfy = \bfw \cdot \bfx_0 = 0$ for all $\bfy \in \calC$, contradicting strict separation.
    If $H$ included $A$, then $\bfw \cdot \bfy = \bfw \cdot \bfx_0 = b$ for all $\bfy \in \calC$, again contradicting strict separation.
    Therefore, $H \cap A$ is neither empty nor equal to $A$.
    Along with the fact that $A$ is an $r$-dimensional subspace, we get that $H \cap A$ is an $(r-1)$ hyperplane in $A$ and strictly separates $\bfx_0$ and $\calC$.
\end{proof}

\begin{lemma}[Gale~\cite{gale1953inscribing}]
\label{lemma:restated__conerank__pointed__convex_simplex_cover}
    Let $S$ be a closed subset of $\bbR^n$ of diameter 1.
    Then $S$ can be inscribed in a regular $n$-simplex of diameter $d \le \sqrt{n(n+1)/2}$.
\end{lemma}

%% file: app_key_lemmas.tex
\clearpage
\section{Technical Lemmas}
\label{app:key_lemmas}

\subsection{Relationships of cones in a subspace}

\begin{lemma}
\label{lemma:1__f_linear__cone_subsets_in_subspace}
    Let $\calL \subseteq \bbR^d$ be an $r$-dimensional linear subspace, and let columns of $\bfZ \in \bbR^{d \times r}$ be an orthonormal basis of $\calL$.
    Let $\calK_{A_1}$ and $\calK_{A_2}$ be cones in $\bbR^d$ generated by rows of matrices $\bfA_1 \in \bbR^{m_1 \times d}$ and $\bfA_2 \in \bbR^{m_2 \times d}$ respectively.
    With $\bfV_1 = \bfA_1 \bfZ$ and $\bfV_2 = \bfA_2 \bfZ$, we have,
    \begin{align*}
        \calL \cap \calK_{A_1}^\ast \subseteq \calK_{A_2}^\ast \iff \calK_{V_1}^\ast \subseteq \calK_{V_2}^\ast \iff \calK_{V_2} \subseteq \calK_{V_1}.
    \end{align*}
\end{lemma}
\begin{proof}
    We can simplify this condition $\calL \cap \calK_{A_1}^\ast \subseteq \calK_{A_2}^\ast$ further by expressing vectors in the basis $\bfZ$.
    
    First, every $\bfx \in \calL$ has a unique representation in the basis $\bfZ$.
    That is, $\bfx = \bfZ \bfc$ for some $\bfc \in \bbR^r$.
    Second, every $d$-dimensional row $\bfa$ of $\bfA_1$ and $\bfA_2$ can be written as $\bfa^\parallel + \bfa^\perp$, where $\bfa^\parallel = \bfa \bfZ \bfZ^\top \in \calL$ and $\bfa^\perp = \bfa (\bfI - \bfZ \bfZ^\top) \in \calL^\perp$.
    Therefore, $\bfA_1 = \bfA_1^\parallel + \bfA_1^\perp$ where $\bfA_1^\parallel = \bfA_1 \bfZ \bfZ^\top$ and $\bfA_1^\perp = \bfA_1 (\bfI - \bfZ \bfZ^\top)$.
    Note that $\bfA_1^\perp \bfZ = \bfzero_{m_1 \times r}$.
    Similarly we can decompose the matrix $\bfA_2 = \bfA_2^\parallel + \bfA_2^\perp$.
    Denote the coefficients as $\bfV_1 = \bfA_1 \bfZ$ and $\bfV_2 = \bfA_2 \bfZ$.
    Using these simplifications, we get:
    \begin{align}
        \calL \cap \calK_{A_1}^\ast \subseteq \calK_{A_2}^\ast &\iff \text{for all } \bfx \in \calL, \bfA_1 \bfx \ge \bfzero \implies \bfA_2 \bfx \ge \bfzero\\
        &\iff \text{for all } \bfc \in \bbR^r, \bfA_1 \bfZ \bfc \ge \bfzero \implies \bfA_2 \bfZ \bfc \ge \bfzero\\
        &\iff \text{for all } \bfc, (\bfA_1^\parallel + \bfA_1^\perp) \bfZ \bfc \ge \bfzero \implies (\bfA_2^\parallel + \bfA_2^\perp) \bfZ \bfc \ge \bfzero\\
        &\iff \text{for all } \bfc, \bfV_1 \bfZ^\top \bfZ \bfc \ge \bfzero \implies \bfV_2 \bfZ^\top \bfZ \bfc \ge \bfzero\\
        &\iff \text{for all } \bfc, \bfV_1 \bfc \ge \bfzero \implies \bfV_2 \bfc \ge \bfzero\\
        &\iff \calK_{V_1}^\ast \subseteq \calK_{V_2}^\ast\\
        &\iff \calK_{V_2} \subseteq \calK_{V_1}.
    \end{align}
    where the last equivalence follows from \Cref{lemma:polyhedral_cones__duals__inclusion}.
\end{proof}

\begin{lemma}
\label{lemma:subspace_scaled_to_relint}
    Let $\calL$ be the linear subspace corresponding to $\affine(X)$.
    For any $\bfx^\ast$ in the relative interior of $X$ and any $\bfx \in \calL$, there exists $a > 0$ such that $a \bfx \in X_{\bfx^\ast}$.
\end{lemma}
\begin{proof}
    We use the definition of relative interior.
    Since $\bfx^\ast$ is in relative interior of $X$, there exists $R > 0$ such that $(\bfx^\ast + R \cdot \bbB_2^d) \; \cap \; \affine(X) \subseteq X$.
    Centering the sets at $\bfx^\ast$, there exists $R > 0$ such that $R \cdot \bbB_2^d \; \cap \; \affine(X)_{\bfx^\ast} \subseteq X_{\bfx^\ast}$.
    We note that $\calL = \affine(X)_{\bfx^\ast}$.

    Let $\bfx \in \calL$.
    If $\bfx = \bfzero$ then we are done as $a \bfx = \bfzero \in X_{\bfx^\ast}$ for any $a > 0$.
    If $\bfx$ is nonzero, then we can normalize it so that $\tilde{\bfx} = R \cdot \frac{\bfx}{\norm{\bfx}} \in R \cdot \bbB_2^d \cap \calL$.
    From the definition of relative interior, we get that $\tilde{\bfx} \in X_{\bfx^\ast}$.
    Thus for any nonzero $\bfx \in \calL$ there exists $a = R / \norm{\bfx}$ such that $a \bfx \in X_{\bfx^\ast}$.
\end{proof}

\begin{lemma}
\label{lemma:2__f_linear__cone_decompose_in_subspace}
    Let $\calL$ be the linear subspace corresponding to $r$-dimensional $\affine(X) \subseteq \bbR^d$, and let columns of $\bfZ \in \bbR^{d \times r}$ be an orthonormal basis of $\calL$.
    For any $\bfx \in X$, denote with $\calC_{\bfx} \subseteq \bbR^r$ the preimage of $X_{\bfx}$ under the orthonormal basis $\bfZ$.
    Let $\calK_A \subseteq \bbR^d$ be generated by rows of $\bfA \in \bbR^{m \times d}$, and let $\bfV = \bfA \bfZ$.
    Then for every $\bff \in \calF$,
    \begin{align*}
        X_{\bfx} \cap \calK_A^\ast \cap \stcomplement{(\ker \bfA)} = \emptyset &\iff \calC_{\bfx} \cap \calK_V^\ast \cap \stcomplement{(\ker \bfV)} = \emptyset.
    \end{align*}
\end{lemma}
\begin{proof}
    Note that for every $\bfx \in X$, the linear subspace spanned by the set $X_{\bfx}$ is $\calL$, and columns of $\bfZ$ are an orthonormal basis of $\calL$.
    That is, for every $\bfy \in X_{\bfx}$ these exists unique $\bfd \in \calC_{\bfx}$ such that $\bfy = \bfZ \bfd$.
    Moreover, we can decompose rows of $\bfA$ in the linear subspace $\calL$ and its orthogonal complement $\calL^\perp$, as in proof of \Cref{lemma:1__f_linear__cone_subsets_in_subspace}.
    We decompose $\bfA = \bfA \bfZ \bfZ^\top + \bfA (\bfI_d - \bfZ \bfZ^\top)$.

    We use these decomposition results to prove the desired result.
    We first prove the forward direction by contradiction.
    Let $\bfx \in X$ and assume that $X_{\bfx} \cap \calK_A^\ast \cap \stcomplement{(\ker \bfA)} = \emptyset$.
    Now assume that there exists $\bfd \in \calC_{\bfx} \cap \calK_V^\ast \cap \stcomplement{(\ker \bfV)}$.
    So $\bfV \bfd \ge \bfzero$ and $\bfV \bfd \neq \bfzero$, implying that $\bfA \bfZ \bfd \ge \bfzero$ and $\bfA \bfZ \bfd \neq \bfzero$.
    Hence, there exists $\bfy = \bfZ \bfd \in X_{\bfx}$ such that $\bfy \in \calK_A^\ast$ and $\bfy \in \stcomplement{(\ker \bfA)}$.
    This contradicts our assumption that $X_{\bfx} \cap \calK_A^\ast \cap \stcomplement{(\ker \bfA)} = \emptyset$, and so we must have $\calC_{\bfx} \cap \calK_V^\ast \cap \stcomplement{(\ker \bfV)} = \emptyset$.

    We also prove the backward direction by contradiction.
    Let $\bfx \in X$ and assume that $\calC_{\bfx} \cap \calK_V^\ast \cap \stcomplement{(\ker \bfV)} = \emptyset$.
    Now assume that there exists $\bfy \in X_{\bfx} \cap \calK_A^\ast \cap \stcomplement{(\ker \bfA)}$.
    So $\bfA \bfy \ge \bfzero$ and $\bfA \bfy \neq \bfzero$.
    Using decomposition of rows of $\bfA$ and $\bfy$ in the basis $\bfZ$, we get that $\bfA \bfy = \bfA \bfZ \bfd$ where $\bfy = \bfZ \bfd$ for $\bfd \in \calC_{\bfx}$.
    So there exists $\bfd \in \calC_{\bfx}$ such that $\bfA \bfZ \bfd \ge \bfzero$ and $\bfA \bfZ \bfd \neq \bfzero$.
    Since $\bfV = \bfA \bfZ$, we get that there exists $\bfd \in \calC_{\bfx} \cap \calK_V^\ast \cap \stcomplement{(\ker \bfV)}$.
    This contradicts our assumption that $\calC_{\bfx} \cap \calK_V^\ast \cap \stcomplement{(\ker \bfV)} = \emptyset$, and so we must have $X_{\bfx} \cap \calK_A^\ast \cap \stcomplement{(\ker \bfA)} = \emptyset$.
\end{proof}

\subsection{Properties of polyhedral cones}

\begin{lemma}
\label{lemma:conerank__pointed_necessary}
    Let $\calK$ be an $r$-dimensional polyhedral cone in $\bbR^n$ such that $\calK$ is pointed.
    If there exists $\bfV \in \bbR^{k \times n}$ such that $\calK \subseteq \calK_V$, then $k \ge r$.
\end{lemma}
\begin{proof}
    We prove by contradiction.
    Let $\bfV \in \bbR^{k \times n}$ be such that $\calK \subseteq \calK_V$.
    Assume $k < r$.
    Hence, we have $\dim \calK_V = \dim \spanv(\calK_V) \le k < r$.
    But $\calK$ is an $r$-dimensional cone and $r = \dim \calK \le \dim \calK_V$, which is a contradiction.
    Therefore, it must be that $k \ge r$.
\end{proof}

\begin{lemma}
\label{lemma:conerank__nonpointed_necessary}
    Let $r > 0$ and $\calK$ be an $r$-dimensional polyhedral cone in $\bbR^n$ such that $\calK$ is nonpointed.
    If there exists $\bfV \in \bbR^{k \times n}$ such that $\calK \subseteq \calK_V$, then $k \ge r+1$.
\end{lemma}
\begin{proof}
    We prove by contradiction.
    Let $\bfV \in \bbR^{k \times n}$ be such that $\calK \subseteq \calK_V$.
    Assume $k \le r$.
    Note that $r = \dim \calK \le \dim \calK_V$ implies that $\rank (\bfV) \ge r$.

    Since $\calK$ is nonpointed, $\calK_V$ is nonpointed.
    \Cref{lemma:nonpointed__zero_in_cone} tells us that there exists nonzero $\bflambda \ge \bfzero$ such that $\bfzero = \bflambda \bfV$.
    This implies that $\rank(\bfV) < k \le r$.
    This a contradiction, and so it must be that $k \ge r+1$.
\end{proof}

\begin{lemma}[{Davis~\cite{davis1954theory}}]
\label{lemma:conerank__subspace_sufficient}
    Let $r > 0$ and $\calL$ be an $r$-dimensional subspace in $\bbR^n$.
    If $k \ge r+1$, then there exists $\bfV \in \bbR^{k \times n}$ such that $\calL = \calK_V$.
\end{lemma}
\begin{proof}
    We include the proof for completion here.
    We will show that there exists $\bfV \in \bbR^{k \times n}$ with $k = r+1$ such that $\calL = \calK_V$.
    Let the columns of $\bfZ = [\bfz_1, \dots, \bfz_r]$ be an orthonormal basis of $\calL$.
    Consider vectors $\bfv_1 = \bfz_1, \dots, \bfv_r = \bfz_r$, and $\bfv_{r+1} = - (\bfz_1 + \cdots + \bfz_r)$. 
    We now show that $\calL = \calK_V$ where $\bfV = [\bfv_1; \dots; \bfv_{r+1}]$.
    \begin{itemize}
        \item $\calL \subseteq \calK_V$.
        Let $\bfx \in \calL$.
        Then there exists $\bfc \in \bbR^r$ such that $\bfx = \bfc \bfZ$.
        Let $c^\ast = \min_{i \in [r]} c_i$. If $c^\ast \ge 0$ then we are done as $\bfx$ is a conic combination of $\bfv_1, \dots, \bfv_r$. Otherwise, we can rewrite $\bfx$ as:
        \begin{align}
            \bfx &= \inparens{ \sum_{i=1}^r (c_i - c^\ast) \bfz_i} + (- c^\ast) \cdot \inparens{ - \sum_{i=1}^r \bfz_i } = \inparens{ \sum_{i=1}^r (c_i - c^\ast) \bfv_i } + (- c^\ast) \bfv_{r+1}
        \end{align}
        noting that the coefficients $(c_1 - c^\ast), \dots, (c_r - c^\ast)$ and $(-c^\ast)$ are nonnegative.
        Hence, $\bfx \in \calK_V$.

        \item $\calK_V \subseteq \calL$.
        Let $\bfx \in \calK_V$.
        Then there exists $\bflambda \in \bbR^k_+$ such that $\bfx = \bflambda \bfV$.
        We can rewrite $\bfx$ as:
        \begin{align}
            \bfx &= \inparens{ \sum_{i=1}^r \lambda_i \bfv_i } + \lambda_{r+1} \cdot \inparens{- \sum_{i=1}^r \bfz_i} = \sum_{i=1}^r (\lambda_i - \lambda_{r+1}) \bfz_i.
        \end{align}

        Therefore, $\bfx$ is a linear combination of $\bfz_1, \dots, \bfz_r$, and $\bfx \in \calL$.
    \end{itemize}

    Therefore, $k = r+1$ vectors are sufficient to generate a cone equalling $r$-dimensional subspace $\calL$.
\end{proof}

\begin{lemma}
\label{lemma:project_cone_interchangeable}
    Let $V \subseteq \bbR^n$ be a vector space and $W = \set{\bfw_1, \dots, \bfw_m}$ be a set of vectors in $\bbR^n$. Then
    \begin{align*}
        \cone \inparens{ \proj_V (W) } = \proj_V (\cone(W)).
    \end{align*}
\end{lemma}
\begin{proof}
    Let $\bfv_1, \dots, \bfv_k$ be an orthonormal basis of $V$. 
    We prove set inclusion in both directions.

    $(\subseteq)$. Let $\bfx \in \cone (\proj_V (W))$. Following definitions of projection, for some nonnegative $\lambda_1, \dots, \lambda_m$,
    \begin{align}
        \bfx &= \sum_{i=1}^m \lambda_i \inparens{ \sum_{j=1}^k \innerproduct{\bfw_i}{\bfv_j} \bfv_j } = \sum_{j=1}^k \innerproduct{ \sum_{i=1}^m \lambda_i \bfw_i }{\bfv_j} \bfv_j \in \proj_V (\cone(W)).
    \end{align}

    $(\supseteq)$. Let $\bfx \in \proj_V (\cone(W))$. Again following definitions, for some nonnegative $\lambda_1, \dots, \lambda_m$,
    \begin{align}
        \bfx &= \sum_{j=1}^k \innerproduct{ \sum_{i=1}^m \lambda_i \bfw_i }{\bfv_j} \bfv_j = \sum_{i=1}^m \lambda_i \inparens{ \sum_{j=1}^k \innerproduct{\bfw_i}{\bfv_j} \bfv_j } \in \cone \inparens{ \proj_V (W) }. \qedhere
    \end{align}
\end{proof}

\begin{lemma}
\label{lemma:csr__property__pointed_positive_basis_has_same_size}
    Let cone $\calK_W$ generated by $W = \set{\bfw_1, \dots, \bfw_m}$ be pointed.
    Let $U$ be such that $U \subseteq W, \calK_U = \calK_W$, and the vectors $U$ are positively independent.
    Let $V$ have the same properties.
    Then $|U| = |V|$.
\end{lemma}
\begin{proof}
    We will show that (i) every vector in $V$ is a positive multiple of some vector in $U$, and (ii) every vector in $U$ is a positive multiple of some vector in $V$.
    As $U, V$ are sets of positively independent vectors, no $\bfu_i \in U$ can be expressed as a nonnegative combination of $U \setminus \set{\bfu_i}$ (and similarly for $V$).
    These statements together imply that $V$ and $U$ are the same set upto positive scaling of vectors, implying that $|U| = |V|$.

    Note that $U, V$ do not contain the zero vector, as the zero vector would be a (trivial) nonnegative combination of other vectors in the set.
    We now show that every vector in $V$ is a positive multiple of some vector in $U$.
    We prove this statement for vector $\bfv_1 \in V$, the proof is the same for the other vectors in $\bfV$.

    Let $|U| = s$ and $|V| = t$.
    Since $\calK_U = \calK_V$, let $\bfv_1 = \sum_{i=1}^s a_i \bfu_i$ for some $a_i \ge 0$ that are not all zero.
    Also, for all $i \in [s]$ let $\bfu_i = \sum_{j=1}^{t} b_{i,j} \bfv_j$ for some $b_{i, j} \ge 0$ such that $b_{i, 1}, \dots, b_{i, t}$ are not all zero.
    So,
    \begin{align}
        \bfv_1 &= \sum_{i=1}^{s} a_i \sum_{j=1}^{t} b_{i,j} \bfv_j\\
        \label{eqn:csr__minimality_proof__v1_in_cone_others}
        \implies \underbrace{\inbraks{ 1 - \sum_{i=1}^{s} a_i b_{i,1} }}_{\rho} \bfv_1 &= \sum_{j=2}^{t} \underbrace{\inparens{ \sum_{i=1}^{s} a_i b_{i,j} }}_{c_j} \bfv_j.
    \end{align}

    Observe that the coefficient $c_j \ge 0$ as $a_i, b_{i, j} \ge 0$ for all $i \in [s], j \in \set{2, \dots, t}$.
    Moreover, when $\rho \neq 0$, not all coefficients $c_j$ are zero simultaneously.
    This is because $c_j = 0$ for all $j \in \set{2, \dots, t}$ would mean that $\bfv_1 = \bfzero$, which contradicts our assumption about vectors in $V$.
    Hence, if $\rho \neq 0$, then there exists $c_j > 0$.

    There are three cases: $\rho > 0$, $\rho < 0$, and $\rho = 0$.
    
    \begin{enumerate}
        \item Case $\rho > 0$.
        We can divide both sides of \Cref{eqn:csr__minimality_proof__v1_in_cone_others} by $\rho$ to find that $\bfv_1$ is a nonnegative combination of vectors $\bfv_2, \dots, \bfv_t$.
        This contradicts our assumption that vectors $V$ are positively independent, and hence this case cannot happen.

        \item Case $\rho < 0$.
        We again divide both sides of \Cref{eqn:csr__minimality_proof__v1_in_cone_others} by $\rho$ to find that $-\bfv_1$ is a nonnegative combination of vectors $\bfv_2, \dots, \bfv_t$.
        This implies that $-\bfv_1 \in \calK_W$ and $\bfv_1 \in \calK_W$, contradicting our assumption that $\calK_W$ is pointed.
        So this case cannot happen.

        \item Case $\rho = 0$.
        \Cref{eqn:csr__minimality_proof__v1_in_cone_others} simplifies to $0 = \sum_{j=2}^{t} c_j \bfv_j$ where $c_j = \sum_{i=1}^{s} a_i b_{i,j}$.
        Since $\calK_W$ is pointed and $\bfv_j \neq \bfzero$, we must have $c_j = \sum_{i=1}^s a_i b_{i, j} = 0$ for all $j \in \set{2, \dots, t}$.
        
        Coefficient $c_j$ is a sum of nonnegative terms $a_i b_{i, j}$, and for $c_j$ to be exactly zero, each term must be exactly zero.
        At least some $a_i > 0$, as otherwise, $a_i = 0$ for all $i \in [s]$ implies $\bfv_1 = \bfzero$.
        W.l.o.g. let $a_1, \dots, a_{\ell} > 0$ for some $1 < \ell \le s$ and $a_{\ell+1}, \dots, a_s = 0$.
        This observation about $a_1, \dots, a_s$ says that $c_j = \sum_{i=1}^{\ell} a_i b_{i, j} = 0$.
        Hence, for all $j \in \set{2, \dots, t}, i \in [\ell]$, we have $b_{i, j} = 0$.

        Using observation about $a_i$ and $b_{i, j}$, we find that $\bfv_1 = \sum_{i=1}^{\ell} a_i \bfu_i$.
        Moreover for $i \in [\ell]$, we have $\bfu_i = \sum_{j=1}^t b_{i, j} \bfv_j = b_{i, 1} \bfv_1$.
        Therefore, $\bfu_1, \dots, \bfu_{\ell}$ are positive multiples of $\bfv_1$.
        If $\ell > 1$ then $\bfu_1, \dots, \bfu_{\ell}$ would be positive multiples of $\bfv_1$, and thus of each other, contradicting assumption that vectors $U$ are positively independent.
        For the same reason, $\bfu_{\ell+1}, \dots, \bfu_s$ cannot be positive multiples of $\bfv_1$.
        Thus, there exists exactly one vector $\bfu_{i_1}$ that is a positive multiple of $\bfv_1$.
    \end{enumerate}

    So, every vector $\bfv_i$ in $V$ is a positive multiple of some vector in $U$.
    In fact, two vectors $\bfv_i, \bfv_j$ where $i \neq j$ are positively multiples of different vectors in $U$, as otherwise $\bfv_i$ would be a positive multiple of $\bfv_j$.
    This means that $|V| \le |U|$.
    By a symmetric argument we can prove that $|V| \ge |U|$, implying that $|U| = |V|$.
\end{proof}

\subsection{Decomposition of polyhedral cones}

\begin{lemma}
\label{lemma:cone__decomposition_pointed_part_is_projection}
    Let cone $\calK \subseteq \bbR^n$ have lineality space $\calL = \calK \cap (-\calK)$.
    Then $\calK \cap \calL^\perp = \proj_{\calL^\perp} (\calK)$.
\end{lemma}
\begin{proof}
    The lineality space $\calL$ is a linear subspace and so $\calL + \calL^\perp = \bbR^n$, which is the Minkowski sum of $\calL$ and its orthogonal complement $\calL^\perp$.
    We can thus write $\calK$ as a Minkowski sum of two sets: $\calK = \calK \cap (\calL + \calL^\perp) = (\calK \cap \calL) + (\calK \cap \calL^\perp)$.
    As $\calL$ is a linear subspace inside $\calK$, we have $\calK \cap \calL = \calL$, implying that $\calK = \calL + (\calK \cap \calL^\perp)$.
    We now project $\calK$ onto $\calL^\perp$, and use the fact that $\proj_{\calL^\perp}(\cdot)$ is a linear operator.
    So,
    \begin{align}
        \proj_{\calL^\perp} (\calK) &= \proj_{\calL^\perp} (\calL + (\calK \cap \calL^\perp))\\
        &= \proj_{\calL^\perp} (\calL) + \proj_{\calL^\perp} (\calK \cap \calL^\perp)\\
        &= \set{\bfzero} + (\calK \cap \calL^\perp).
    \end{align}
    where the last equality follows from $\calL^\perp$ being orthogonal complement of $\calL$, and $\calK \cap \calL^\perp$ lying inside $\calL^\perp$ so that projecting $\calK \cap \calL^\perp$ onto $\calL^\perp$ has no effect.
    We note that $\bfzero \in \calK \cap \calL^\perp$, which is the intersection of a cone and a linear subspace, and so $\set{\bfzero} + (\calK \cap \calL^\perp) = \calK_W \cap \calL^\perp$.
    We thus get the desired result $\proj_{\calL^\perp} (\calK) = \calK \cap \calL^\perp$.
\end{proof}

\begin{lemma}
\label{lemma:cone__lineality_space__lineal_points}
    Let $W = \set{\bfw_1, \dots, \bfw_m} \subseteq \bbR^n$ generate cone $\calK_W$.
    Let $\calL = \lineality (\calK_W)$ be the lineality space of $\calK_W$, and denote $\WinL = \set{ \bfw_i \in W \mid \bfw_i \in \calL }$.
    Then $\calL = \calK_{\WinL}$.
\end{lemma}
\begin{proof}
    This lemma is based on the following observation: since $\calL$ is the maximal linear subspace inside $\calK_W$, for any point $\bfw_i$ of $W$ not in $\calL$, we have $-\bfw_i \notin \calK_W$. 
    Hence, the lineality space $\calL$ can only be generated by the points of $W$ that are inside $\calL$. 

    As $\WinL \subseteq \calL$ and $\calL$ is a linear subspace, we have $\calK_{\WinL} \subseteq \calL$.

    We now show that $\calL \subseteq \calK_{\WinL}$.
    Here $\calL$ is the lineality space of $\calK_W$, i.e., $\calL = \calK_W \cap (-\calK_W)$.
    Denote $\WnotinL = \set{ \bfw_i \in W \mid \bfw_i \notin \calL }$, and so $\calK_W = \calK_{(\WinL \cup \WnotinL)} = \calK_{\WinL} + \calK_{\WnotinL}$, which is the Minkowski sum of two cones.
    Let $\bfx \in \calL$.
    As $\calL \subseteq \calK_W = \calK_{\WinL} + \calK_{\WnotinL}$, we have $\bfx = \bfy + \bfz$ for some $\bfy \in \calK_{\WinL}$ and $\bfz \in \calK_{\WnotinL}$.
    From earlier we know that $\calK_{\WinL} \subseteq \calL$, and so $\bfx - \bfy = \bfz \in \calL$.
    So $\bfz \in \calL \cap \calK_{\WnotinL}$.
    We next show that $\bfz$ is exactly zero, implying that $\bfx = \bfy \in \calK_{\WinL}$ and completing the proof.

    Suppose that $\bfz \in \calL \cap \calK_{\WnotinL}$ is such that $\bfz \neq \bfzero$.
    Then $\bfz = \sum_{\bfw_i \in \WnotinL} \lambda_i \bfw_i$ for some $\lambda_i \ge 0$ that are not all zero.
    W.l.o.g. let $\lambda_1 \neq 0, \bfw_1 \neq \bfzero$ and rearrange to get:
    \begin{align}
        -\bfw_1 &= -\frac{\bfz}{\lambda_1} + \sum_{\bfw_i \in \WnotinL, i \neq 1} \frac{\lambda_i}{\lambda_1} \bfw_i.
    \end{align}

    Note that as $\bfz \in \calL$, which is a linear subspace inside $\calK_W$, we have $-\bfz / \lambda_1 \in \calL \subseteq \calK_W$.
    As every $\bfw_i \in \WinL$ is also in $W$, we get that $-\bfw_1 \in \calK_W$.
    Hence, $\bfw_1 \in \calK_W \cap (-\calK_W) = \calL$.
    This contradicts the choice of $\bfw_1$, which is an element in the set $\WnotinL = \set{ \bfw_i \in W \mid \bfw_i \notin \calL }$.
    So our assumption that $\bfz \neq \bfzero$ was wrong.
\end{proof}

\begin{lemma}
\label{lemma:cone__lineality_space__nonlineal_points}
    Let $W = \set{\bfw_1, \dots, \bfw_m} \subseteq \bbR^n$ generate cone $\calK_W$.
    Let $\calK_W = \calL + \calK_P$ be the decomposition such that $\calL = \lineality(\calK_W)$ is the lineality space of $\calK_W$ and $\calK_P = \calL^\perp \cap \calK_W$ is a pointed cone.
    Denote $\WnotinL = \set{ \bfw_i \in W \mid \bfw_i \notin \calL }$.
    Then $\cone(\tilde W) = \calK_P$ where $\tilde W = \proj_{\calL^\perp} (\WnotinL)$.
\end{lemma}
\begin{proof}
    We denote $\WinL = \set{ \bfw_i \in W \mid \bfw_i \in \calL }$, and note that $W = \WinL \cup \WnotinL$.
    Here $\WinL$ only contains vectors that lie in $\calL$, and $\WnotinL$ does not contain vectors that lie in $\calL$.
    Hence, we get
    \begin{align}
        \proj_{\calL^\perp} (W) &= \proj_{\calL^\perp} (\WinL) \cup \proj_{\calL^\perp} (\WnotinL) = \set{\bfzero} \cup \proj_{\calL^\perp} (\WnotinL).
    \end{align}

    \Cref{lemma:cone__decomposition_pointed_part_is_projection} states that $\calK_P = \proj_{\calL^\perp} (\calK_W)$, and thus we have
    \begin{align}
        \calK_P &= \proj_{\calL^\perp} (\cone(W)) \longeq{(a)} \cone \inparens{ \proj_{\calL^\perp} (W) } \longeq{(b)} \cone \inparens{ \proj_{\calL^\perp} (\WnotinL) } = \cone(\tilde W)
    \end{align}
    where equality $(a)$ follows from \Cref{lemma:project_cone_interchangeable}, and $(b)$ from the fact: $\cone(\set{\bfzero} \cup X) = \cone(X)$ for set $X$.
\end{proof}